\documentclass[a4paper,fleqn]{cas-sc}
\usepackage[authoryear]{natbib}

\usepackage{cite}
\usepackage{comment}
\usepackage{relsize}
\usepackage{subcaption}
\usepackage{amsmath,amssymb,amsfonts}
\usepackage{graphicx}
\usepackage{textcomp}
\usepackage{xcolor}
\usepackage{rotating}
\usepackage{adjustbox}
\usepackage{hyperref}
\usepackage{placeins}
\usepackage{cancel}
\usepackage{enumitem}
\usepackage{soul}
\usepackage{algorithmic}
\usepackage{array}
\usepackage{url}
\usepackage{multicol}
\usepackage{booktabs}

\usepackage{algorithm}%
\usepackage[algo2e,ruled,vlined]{algorithm2e}
\SetArgSty{textup}

\usepackage{amsthm}
\theoremstyle{plain}
\newtheorem{theorem}{Theorem}
\newtheorem{congestion}{Rule}
\newtheorem{definition}{Definition}
\newtheorem{assumption}{Assumption}

\begin{document}
\let\WriteBookmarks\relax
\def\floatpagepagefraction{1}
\def\textpagefraction{.001}
\shorttitle{Congestion-Aware Path Re-routing Strategy for Dense Urban Airspace}
\shortauthors{Sajid Ahamed et~al.}

\title[mode = title]{Congestion-Aware Path Re-routing Strategy for Dense Urban Airspace}

\author[1]{Sajid Ahamed Mohammed Abdul}[style = chinese, orcid=0000-0002-6261-9541]
\cormark[1]
\ead{sajidahamed@iisc.ac.in}

\credit{Conceptualization of this study, Methodology, Software, Formal Analysis, Writing-Original Draft}

\affiliation[1]{organization={Department of Aerospace Engineering, Indian Institute of Science}, 
                city={Bengaluru},
                postcode={560012}, 
                state={Karnataka},
                country={India}}

\author[2]{Prathyush P. Menon}[style = chinese, orcid= 0000-0003-3804-9291]
\ead{P.M.Prathyush@exeter.ac.uk}
\credit{Conceptualization, Supervision, Writing - Review \& Editing}

\affiliation[2]{organization={Cooperative Robotics \& Autonomous NEtworks Laboratory (CRANE)},
                addressline={University of Exeter}, 
                postcode={EX4 4QF}, 
                postcodesep={}, 
                city={Streatham},
                country={United Kingdom}}

\author[3]{Debasish Ghose}[style = chinese, orcid=0000-0001-5022-4123]
\ead{dghose@iisc.ac.in}
\credit{Conceptualization, Supervision, Writing - Review \& Editing}

\affiliation[3]{organization={Robert Bosch Center for Cyber-Physical Systems (RBCCPS) \& Department of Aerospace Engineering},
                addressline={Indian Institute of Science}, 
                city={Bengaluru},
                postcode={560012}, 
                state={Karnataka}, 
                country={India}}

\cortext[cor1]{Corresponding author}

\begin{abstract}
Existing UAS Traffic Management (UTM) frameworks designate preplanned flight paths to uncrewed aircraft systems (UAS), enabling the UAS to deliver payloads. However, with increasing delivery demand between the source-destination pairs in the urban airspace, UAS will likely experience considerable congestion on the nominal paths. We propose a rule-based congestion mitigation strategy that improves UAS safety and airspace utilization in congested traffic streams. The strategy relies on nominal path information from the UTM and positional information of other UAS in the vicinity. Following the strategy, UAS opts for alternative local paths in the unoccupied airspace surrounding the nominal path and avoids congested regions. The strategy results in UAS traffic exploring and spreading to alternative adjacent routes on encountering congestion. The paper presents queuing models to estimate the expected traffic spread for varying stochastic delivery demand at the source, thus helping to reserve the airspace around the nominal path beforehand to accommodate any foreseen congestion. Simulations are presented to validate the queuing results in the presence of static obstacles and intersecting UAS streams.
\end{abstract}


\maketitle

\section{Introduction}
With the rapid growth in the number of uncrewed aircraft systems (UAS) being deployed for on-demand delivery applications, severe congestion is foreseen in the urban airspace. There is a need for structuring the airspace to mitigate this congestion. Delivery applications in urban airspaces would normally require UAS to travel on preplanned paths between source-destination pairs.  Geofencing these paths can physically separate a UAS trajectory from other UAS, as well as static obstacles. However, intersecting UAS trajectories are inevitable, especially in dense traffic scenarios. Increased traffic on these paths due to increasing delivery demand results in severe conflicts among the UAS. Further, the increasing number of conflicts in confined geofenced air volumes could be detrimental to UAS safety. Such limitations are majorly due to treating UAS traffic as a road-like network or air traffic, where UAS are restricted to the preplanned paths and air volumes. To overcome this problem, we propose in this paper an airspace design and a congestion mitigation strategy that utilizes the airspace surrounding the preplanned path and dynamically generates local paths for UAS. These paths are opted for when the preplanned paths are congested. We also suggest an approach for estimating the expected air volume that would be utilized as a function of the probabilistic occurrence of congestion. Such airspace can be adaptively reserved beforehand to accommodate any foreseen congestion.  

\textit{Background and Literature:} Several government and private agencies (\citet{kopardekar2016unmanned,catapult2020enabling,balakrishnan2018blueprint}) have developed different airspace designs to enable beyond-visual range UAS operations. These designs impose structural constraints on UAS's degree of freedom (DoF) to maximize airspace safety and capacity at the expense of introducing time delays in UAS flights. Broadly, the structures proposed suggest segmenting the airspace and restricting the UAS in an air matrix (discrete UAS DoF) (\citet{pang2020concept}); in a network of virtual air geofenced corridors (unidirectional UAS traffic flow in an enclosed air corridor volume) (\citet{tony2021lane}); or in geo-vectored altitude layers, where UAS heading vector range is limited in each altitude layer and the altitude layers are interconnected by vertical ascend or descend corridors) (\citet{hoekstra2018geovectoring}).  

UAS safety risk is a probabilistic belief on collision occurrence among UAS. In \citet{doole2021constrained}, the event where the separation between UAS is below a threshold is termed an intrusion. An anticipated intrusion in a prescribed look-ahead time window is termed a conflict. Whereas congestion is an accumulation of UAS in a predefined region that eventually leads to increased conflict risk. The airspace designs have an intrinsic ability to reduce the occurrence of collision due to the constraints that it imposes on the UAS traffic.

In air matrix design, the airspace is tessellated into a cellular grid, and online path-planning strategies such as Rapidly exploring Random Trees (RRT), skeleton maps (\citet{li2023bi, jana2023numerical}) were proposed that generate optimal and risk-aware trajectories for UAS. However, with the increasing number of UAS, these strategies lead to an unstructured and unpredictable traffic pattern. 
In altitude-layered airspace, the intrinsic safety can be improved if the UAS is capable of sensing and avoiding (SAA) the intrusions (\citet{agnel2021unmanned}) or can mitigate conflict risk by preflight or in-flight conflict detection and resolution (CDR) (\citet{sacharny2022lane, nagrare2024intersection}). 
Though the SAA and CDR are adequate at individual UAS levels, these approaches perform poorly in congested traffic scenarios. In the study conducted by \citet{watkins2021investigative}, beyond a critical demand, the SAA can no longer execute avoidance maneuvers without colliding with neighboring UAS. Whereas the in-flight CDR suffers from domino effects and spillback conflicts, which can severely destabilize the UAS traffic (\citet{sedov2017decentralized, aarts2023capacity}). 
A solution to this problem as proposed in this paper, is to mitigate congestion and prevent high UAS density regions from forming than to allow UAS to proceed and encounter conflicts in dense regions.  

The analytical study in \citet{sunil2018three} suggests segregating UAS into altitude layers as prima facie for mitigating congestion. 
In \citet{sedov2017decentralized}, when two or more UAS conflict, these UAS are collectively enclosed in a virtual disk and labeled as congested. Any UAS in its vicinity locally re-route, avoiding this disk until the conflicts in the disk are resolved. In \citet{lee2022congestion}, the airspace is tesselated into cells, and the cell is said to be congested if the number of pre-planned paths occupying the same cell is greater than the cell capacity. Similarly, in \citet{egorov2019encounter}, a region is congested if the expectation of UAS encountering conflict with two or more UAS is above some threshold. 
The papers  \citet{zhou2020resilient,gharibi2021density,wang2021air} present strategies for mitigating congestion in the air corridor networks. In \citet{zhou2020resilient}, the UAS traffic is modeled as fluid queues, and the congestion is mitigated by controlling the UAS inflow and outflow rate for each corridor in the network. In \citet{gharibi2021density}, 
the velocity of the trailing UAS is reduced proportional to the weighted sum of UAS inter-separation distances in the congested volume. In \citet{wang2021air}, 
the congestion is mitigated by solving a traffic assignment problem such that the number of UAS routed on each corridor results in an equitable UAS density distribution in all corridors. In \citet{egorov2019encounter, lee2022congestion,wang2021air}, the mitigation strategy is spatial and does not capture congestion in the temporal sense, forcing new incoming UAS to take long detours. Whereas in \citet{zhou2020resilient,gharibi2021density}, intermittent local congestion events affect the throughput of the entire corridor network. This limitation is because the UAS is always restricted within air corridors despite airspace being available in the surroundings. 

\textit{Our Contributions:} In this paper,  we present a distributed congestion mitigation strategy where UAS decisively opts for local alternative paths upon encountering congestion in its preplanned flight path and reverts when congestion reduces. Such path re-routing response of UAS ensures no further UAS inflow into congested regions and upper bounds the number of conflicts in the region. The proposed strategy can handle congestion contingencies due to system failure, rogue agents, intersecting UAS streams, and time-varying stochastic delivery demand. The main contributions of this paper are as follows. 

\begin{enumerate}
    \item A novel distributed congestion mitigation strategy is proposed, imposing congestion-aware and preference-based heading structural constraints on UAS degree of motion. 
    \item Demonstrating how the proposed strategy adapts to the increasing UAS demand and dynamically creates parallel air corridors, thus offering intrinsic safety for high UAS traffic demand. 
    \item Development of queuing-based models for estimating the expected spread of the UAS traffic stream when every UAS implements the proposed congestion mitigation strategy.  
\end{enumerate}

The paper is organized as follows. The workspace for the congestion mitigation problem is defined in Section \ref{sec: problem definition}. Section \ref{sec: congestion mitigation methodology} introduces congestion, the heading structural constraints being imposed, and the proposed congestion mitigation strategy employed by UAS. Section \ref{sec: congestion queuing model} presents queuing models to study the emerging traffic behavior as a function of the UAS traffic demand and structural constraints. Section \ref{sec: simulation results} presents simulations that demonstrate the increased intrinsic safety and resemblance of UAS traffic to the parallel air corridor network. Section \ref{sec: conclusion} concludes the paper, emphasizing the advantages of the proposed mitigation strategy. 

\section{Problem Definition} \label{sec: problem definition}

This paper presents traffic analysis for a UAS traffic management (UTM) framework driven by stochastic delivery demand. The UTM environment has static obstacles and several source-destination pairs. The delivery service provider deploys a UAS for every customer requesting delivery on a specific source-destination route, and such requests recur stochastically. Accordingly, at the source location, a time-varying geometrically distributed request for UAS is assumed with bounded request rate $\Lambda(t) \in [0,\Lambda_{\max}]$, where $\Lambda_{\max}$ is the maximum request rate. The time axis is divided into slots of $\Delta T = 1/\Lambda_{\max}$ time units and indexed with $k \in \mathbb{N}_0$. We assume that though multiple requests may be received in a given timeslot, only one request among them is randomly chosen, for which a UAS is deployed at the start of the next timeslot. The probability that a UAS would be deployed in a timeslot is given by $\lambda = \Lambda \Delta T$.   

The workspace is assumed to be a bounded horizontal stretch of air volume containing the given source-destination pair. It is tessellated into a connected grid of a finite number of identical cells. As the static obstacles are of varying heights, some may intersect the cells of this workspace. We define a geometric center called the cell center for each static obstacle-free cell (refer Fig. \ref{fig: 1a}). The motion of a UAS leaving a cell center and traveling in a straight line to reach a neighboring cell center is referred to as a transition. The traversal of UAS in the workspace is modeled as successive transitions from one cell center to any of its neighboring cell centers. We assume that the UAS takes $\Delta T$ time units to transition between neighboring cell centers. Two UAS that are $\tau$ cells spatial separated are equivalently $\tau$ slots ($\tau \Delta T$) temporal separated. The UTM prescribes a nominal path connecting cell centers between any source-destination pair. We assume this path would be approximated by connected straight-line segments passing through cell centers. Since the UAS deployment is geometrically distributed with parameter $\lambda$, the inter-separation distance between UAS transitioning on the nominal path (in timeslots) is also geometrically distributed and is given by the distribution $A(t)$. 
\begin{align}
    A(t) = \mathrm{P}[\tau = t] = \lambda (1-\lambda)^t
    \label{eqn: inter-separation}
\end{align}
where $\mathrm{P}[\tau = t]$ is the probability that the inter-separation distance $\tau$ between UAS is $t$ timeslots. If $\tau \leq 1$ between any two UAS, these two UAS are present in a cell and counted as an intrusion. Analogously, $\mathrm{P}[\tau \leq 1]$, the probability of an intrusion in the $\Delta T$ look-ahead time window is the conflict risk ($R_{conflict}$).  

Let $S \geq 1 $ timeslots be the threshold on UAS inter-separation distance for defining congestion. Consider an arbitrary region in which more than $M (\geq 2)$ UAS are present, with the inter-separation distances between them being less than $S$. 
We assume positive correlation of $\tau$ with $M$ and negative correlation with $S$, that is, $\mathrm{P}[\tau \leq 1 \vert \tau \leq S, M \geq 2] > \mathrm{P}[\tau \leq 1 \vert \tau \leq S',  M = 2],\ S < S'$. The relation means there is an increased risk of conflict $(R_{conflict} = \mathrm{P}[\tau \leq 1])$ when there are $M$ UAS that are $\leq S$ timeslots closer to each other. Avoiding such higher-risk regions is preferable. These regions are referred to as congested regions. Here, $S$ is a UTM design parameter.

The requests are recurring, hence there is a stream of UAS traffic flowing in the nominal path from source to destination. The traffic in the path, right ahead in front of a UAS, is its upstream traffic. The traffic trailing behind the UAS is its downstream traffic.
When $\lambda$ is large, $\mathrm{P}[\tau \leq S]$ is high; thus, the chances of a UAS encountering congestion on its upstream is high. When the probability of congestion increases, there is a requirement to generate alternative parallel paths (dependent on the request rate) that avoid congested regions. The nominal path is a reference path given to the UAS. However, the UAS is free to opt for any of those cell centers in its neighborhood so long as the choice does not violate the constraints set by the UTM. The main objective of this paper is to propose a distributed congestion mitigation strategy for the UAS. Using this strategy, individual UAS can dynamically generate local parallel paths for itself in the unoccupied airspace in a structured manner, following which the UAS bypasses congested regions in its local neighborhood. When the request rate $\lambda$ is low, the UAS traverses along the nominal flight path. However, in the event of congestion on the nominal path, a certain number of UAS present in the downstream are diverted onto the parallel paths. As the request rate increases, the frequency of the congestion event increases. The number of UAS that would be diverted on parallel paths also increases, causing congestion on the parallel paths as well. Then, UAS downstream would have to opt for parallel paths further outward from the nominal path. When congestion reduces, the UAS on parallel paths tends to close in on the nominal path. As the request rate is time-varying, a change in the area occupied by the UAS stream around the nominal path will be observed in the process of congestion mitigation.
 
 The above idea is illustrated in Fig. \ref{fig: 1b}. The UAS only relies on the nominal path information from UTM and the positional information of UAS present in its vicinity to determine congested regions and generate local parallel paths. For the UAS traffic routed between source-destination pairs, we demonstrate and analyze the extent of the UAS traffic spread around the nominal path as a function of the request rate. The paper extends congestion mitigation to scenarios where the inter-separation distances between $M$ UAS of intersecting traffic streams fall below the threshold $S$.
 
\begin{figure}[h!]
\centering
     \begin{subfigure}[b]{0.55\textwidth}
    \centering
    \includegraphics[width = \linewidth]{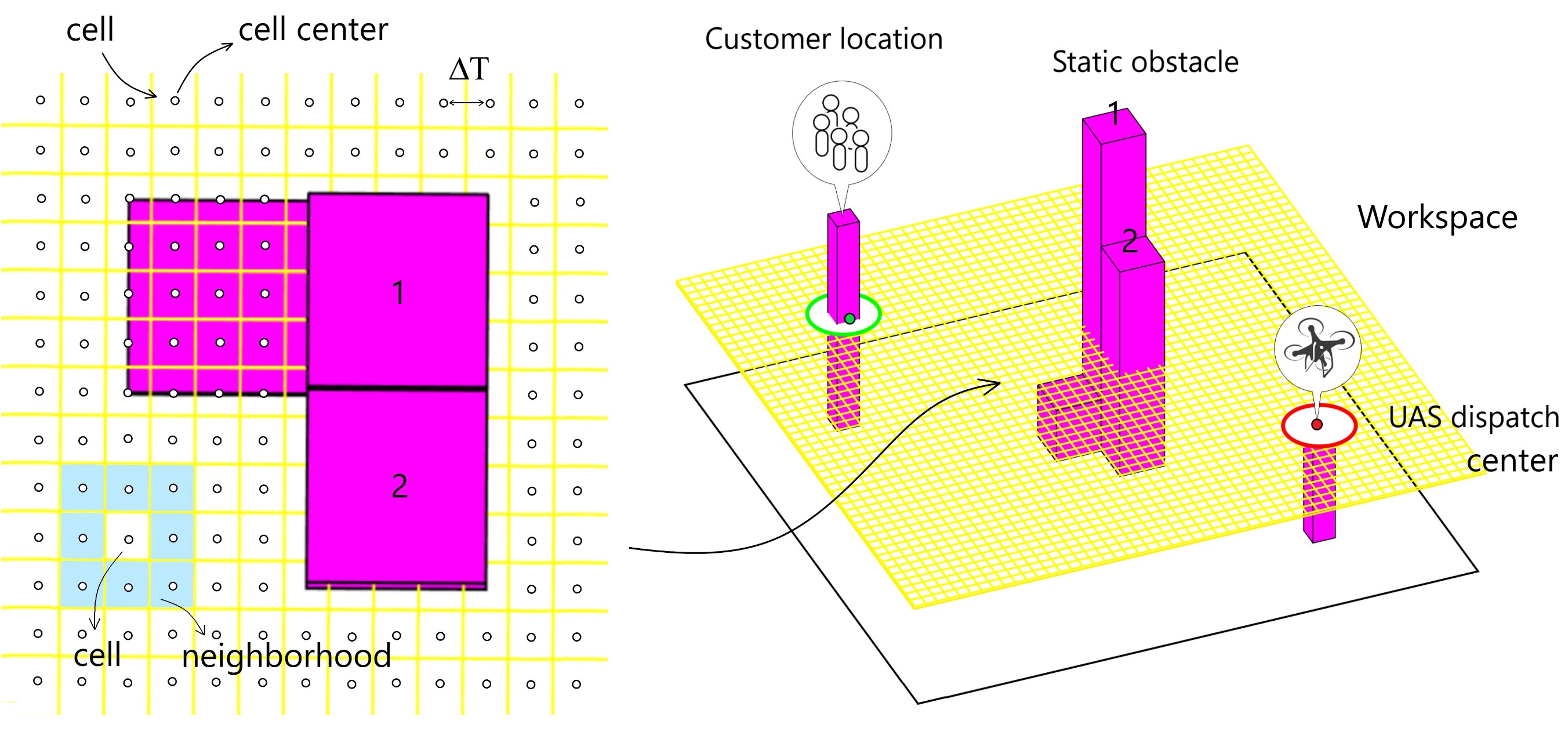}
    \caption{}
    \label{fig: 1a}
    \end{subfigure}
     \hfill
     \begin{subfigure}[b]{0.44\textwidth}
\centering
    \includegraphics[width = 0.65\linewidth]{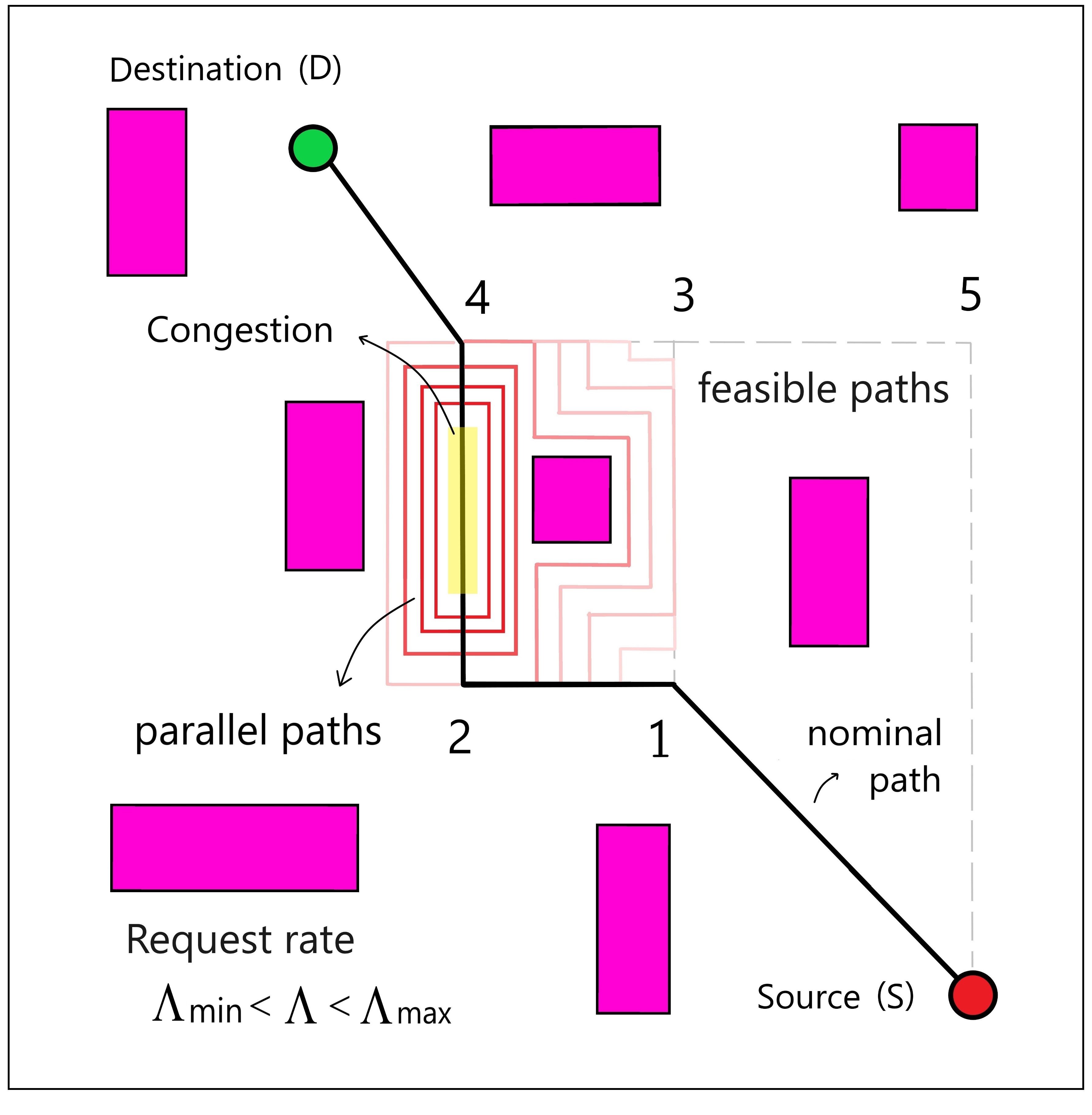}
    \caption{} 
    \label{fig: 1b}
     \end{subfigure}
     \caption{a) Cell tessellation of workspace containing the UAS dispatch center (source) and the customer location (destination). b) Several feasible paths may exist between source and destination ($\mathsf{S}$-$\mathsf{1}$-$\mathsf{2}$-$\mathsf{4}$-$\mathsf{D}$, $\mathsf{S}$-$\mathsf{1}$-$\mathsf{3}$-$\mathsf{4}$-$\mathsf{D}$, $\mathsf{S}$-$\mathsf{5}$-$\mathsf{4}$-$\mathsf{D}$). Here, the path shown by the bold line is the nominal path given by the UTM. In the event of congestion on this path, the figure illustrates the local parallel paths that are generated. The expected traffic spread or the expected path spread is the average number of parallel paths on which the UAS traffic would be diverted depending on the request rate $\Lambda$.}
\end{figure}
\FloatBarrier
\vspace{-0.5cm}
\section{Congestion Mitigation Methodology} \label{sec: congestion mitigation methodology}

The nominal path between the source-destination pair is segmented into a connected chain of straight-line segments. The set of cells through which the straight-line segment passes is called a nominal segment, with a start and an end cell defined for each segment. The proposed strategy deals with congestion in each nominal segment separately. A unit vector directed from start to end is called an \textit{upstream} vector, and that directed from end to start is called a \textit{downstream} vector. To identify a region as congested, we first need to geometrically define the region, for which we partition the workspace as follows. 

Assume the UAS are identical in all aspects and have a $L$ slot look-ahead time window to predict where other UAS would be in the next $L\Delta T$ time ahead. The cells in the nominal segment (start and end included) that are $S = 2L+1$ timeslots apart are called nodes (shown in Fig. \ref{fig: 2}a). With each node as a geometric center, the set of cells enclosed in a square region with edge length $S$ timeslots (an edge parallel to the segment) is called a zone. The eight square regions (of $S$ timeslots edge length) surrounding the zone are called neighborhood zones. We may successively define the neighborhood for all neighborhood zones. The zone and its corresponding node are denoted with $Z$ and $\widetilde{Z}$, respectively. We assume a finite grid of identical zones for each nominal segment to facilitate alternative path generation in the event of congestion. The grid is so positioned that the nominal segment bisects the grid (shown in Fig. \ref{fig: 2}). The length of the grid (that is, the length of the nominal segment) is significantly larger than the width. The start cell, end cell, and width of the grid for two connected nominal segments are chosen such that the respective grids do not intersect. The zones through which the nominal segment passes are called nominal zones. A stream-level tuple notation $Z = (X, Y)$ is used for referring to zones in the finite grid, where stream $X$ is abscissa, $X \in \{-X_e,...,0,..., X_e\}$ and the level $Y$ is ordinate, $Y \in \{1,..., Y_e\}$, and $X_e$ and $Y_e$ are positive integers (refer Fig. \ref{fig: 2}b top view). The nominal zones are called Stream$(0)$ zones denoted with tuple $(0, Y)$, where $Y \in \{1,..., Y_e\}$ is the zone level in the \textit{upstream} direction. Here, $(0,1)$, $(0, Y_e)$ are the nominal zones containing start and end cells, respectively. Similarly, Stream$(\pm 1)$, Stream$(\pm 2)$, and so on Stream$(\pm X_e)$ zones can be identified. By interconnecting nodes of the Stream$(X)$ zones with respective \textit{upstream} nodes, $X\in\{-X_e,...,X_e\}/ \{0\}$, we get $2X_e$ parallel segments to the nominal segment, referred to as Stream$[X]$ segments. The congested regions are defined as follows.
\begin{figure}[h!]
    \centering
    \includegraphics[width = \linewidth]{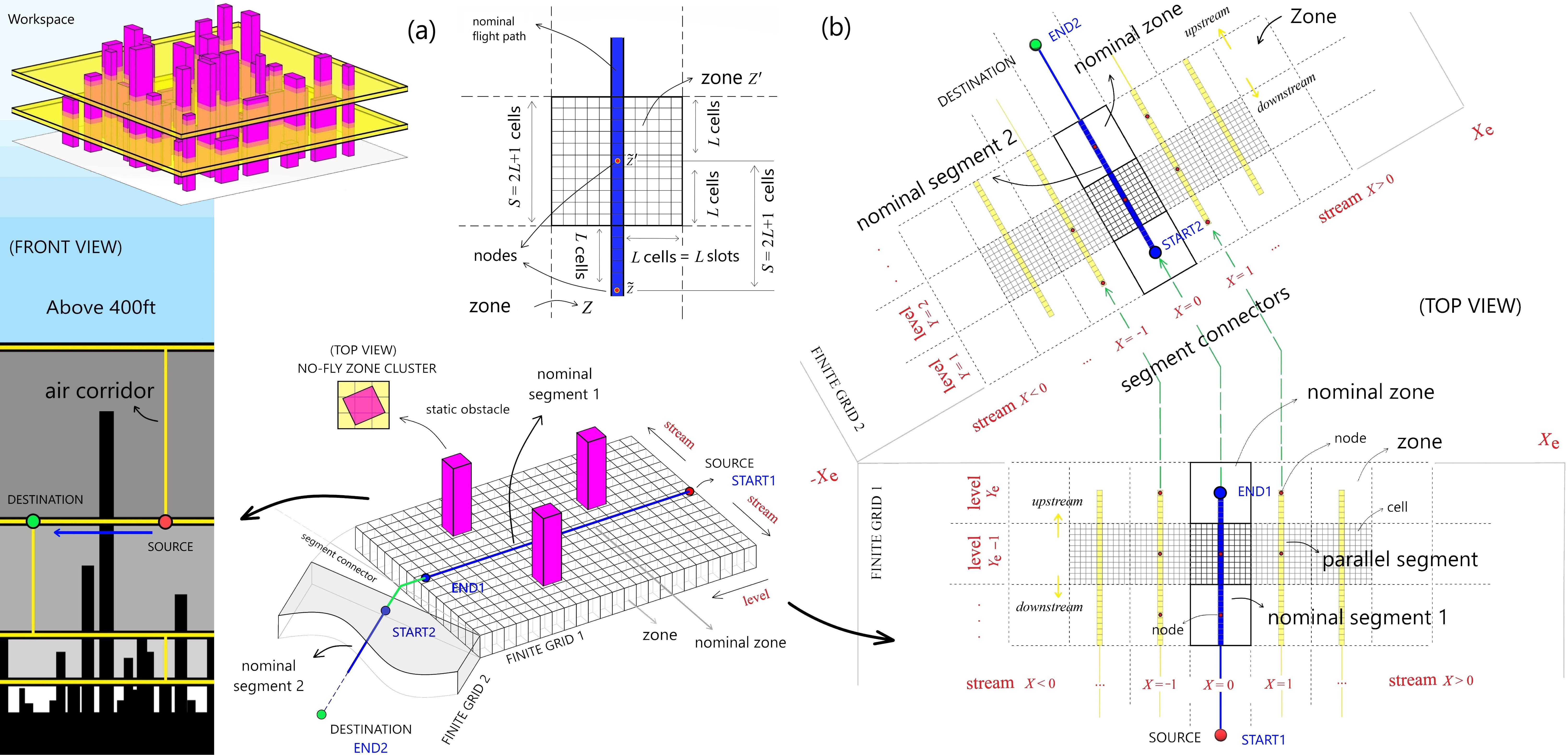}
    \caption{a) Node and zone definition. b) Connected grid and stream-level indexing of zones about each nominal segment.}
    \label{fig: 2}
\end{figure}

\begin{definition}
\textit{A zone is said to be congested if $M$ or more than $M$ UAS are present in the zone.} \label{def: 1}
\end{definition}

Here, $L$ and $M$ are UTM design parameters. Note that the congestion in the zone is due to UAS transitioning on the given source-destination route as well as due to any exogenous UAS coexisting in the workspace. The grid construction is UAS-agnostic. The UAS can extract the grid geometry from the nominal segment information and design parameters shared by the UTM. Within the zone, the heading of the UAS would be restricted so that UAS would only transition on a subset of cells. The heading constraint is position-dependent. For example, when UAS is in the nominal segment, it is restricted to transition in the \textit{upstream} direction. A UAS in a zone can detect or communicate with other UAS coexisting in the zone and respective zone neighborhood. By knowing the UAS present positions, the heading constraint imposed, and that any UAS would complete one transition in one timeslot, we assume the UAS using a linear kinematic model can predict other UAS state information (position and heading) $L\Delta T$ time ahead with sufficient accuracy.

The UAS on reaching a node, needs to take at least $L$ successive transitions to enter any of the neighboring zones. The UAS is equipped with $L$ slot look-ahead prediction window. The UAS can predict other UAS positions and determine which neighboring zones would probably become congested in the next $L$ slots. Thus, when the UAS is at the node, it could selectively choose not to transition towards neighboring zones that will become congested.   
\begin{definition} \textit{From the perspective of a UAS in zone $Z$ (depicted in Fig. \ref{fig: 3}), the neighboring zone $Z'$ is congested if the number of $L$ slot look-ahead UAS predicted positions in zone $Z'$ is greater than or equal to $M$.} \label{def: 2}
\end{definition}

 \begin{figure}[h!]
    \centering
    \includegraphics[width = 0.5\linewidth]{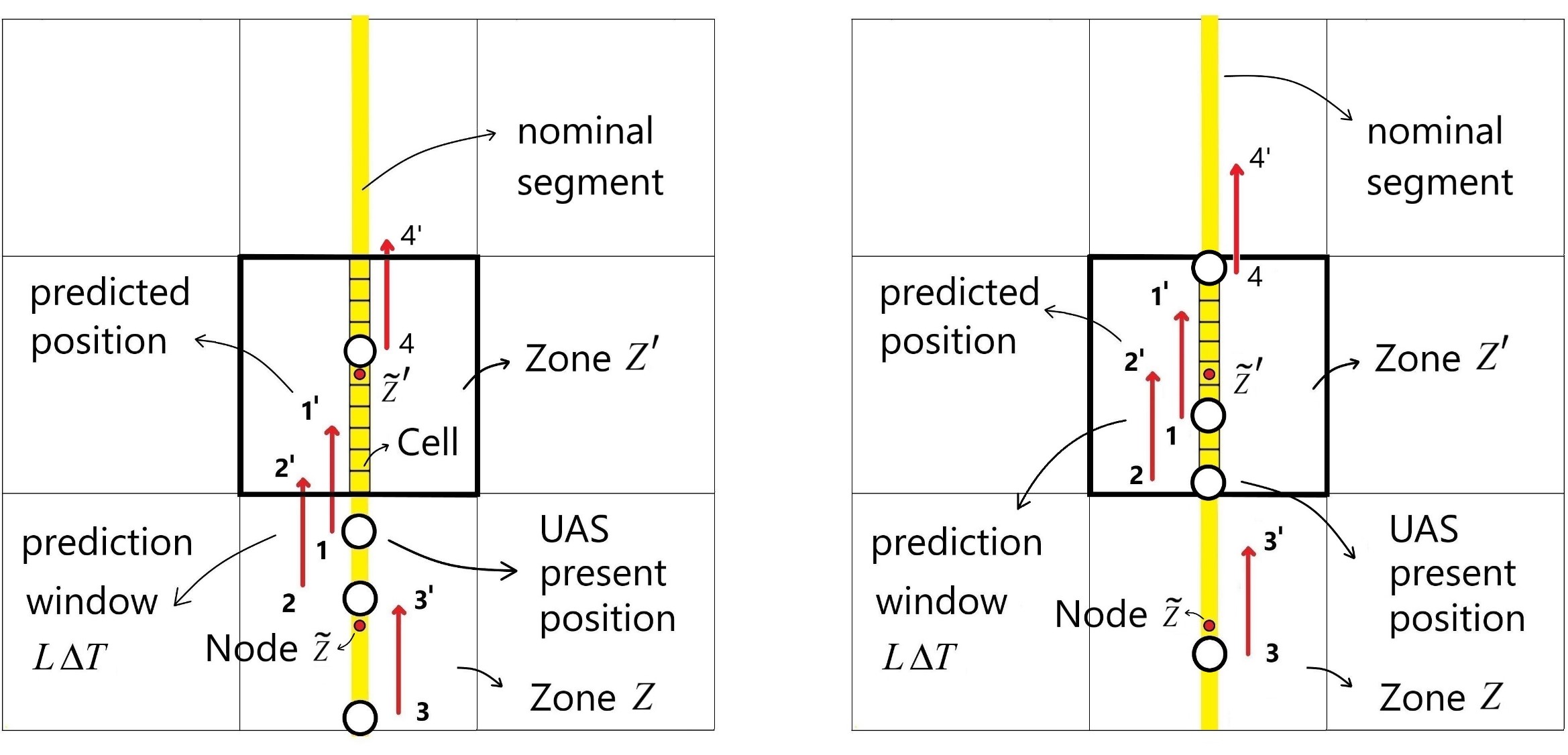}
    \caption{In the figure, with respect to the UAS $i \in \{1,2,3,4\}$, the set $\{\ j'\ \vert\ j \in \{1,2,3,4\} \backslash \{i\}\ \}$ are the $L$ slot look-ahead predicted positions of other UAS. When $M = 2$, from the perspective of UAS 2 present in zone $Z$, zone $Z'$ will not be congested in the next $L$ timeslots as only one predicted position (that of UAS 1) exists in $Z'$. From the perspective of UAS 3 present in zone $Z$, the zone $Z'$ will be congested as two predicted positions (UAS 1 and 2) lie in zone $Z'$. } 
    \label{fig: 3}
\end{figure}

 The congestion definition applies to all the zones in the grid. However, the decisions of UAS present in any zone $Z$ are only influenced by congestion occurrences in two of its neighboring \textit{upstream} zones. These are shown as blue-colored zones in Fig. \ref{fig: 4}. In this figure, the neighborhood zones of $Z= (X, Y)$ that are of significance for the congestion mitigation strategy are denoted as $Z_{\mathrm{UP}}, Z_{\mathrm{IN}}, Z_{\mathrm{ID}}$ and $Z_{\mathrm{ON}}$ (notation described in Table. \ref{tab:table1}). Note that $Z_{\mathrm{IN}}$ and $Z_{\mathrm{ID}}$ are defined only when $X\neq 0$. The zone notation can be extended to respective nodes as well. At node $\widetilde{Z}$, if the UAS decides to avoid zone $Z_\mathrm{UP}$, it generates a piecewise re-route path (a sequence of the number of transitions in a specific heading direction).  
\begin{figure}[h!]
    \centering
   \includegraphics[keepaspectratio, width = 0.6\linewidth]{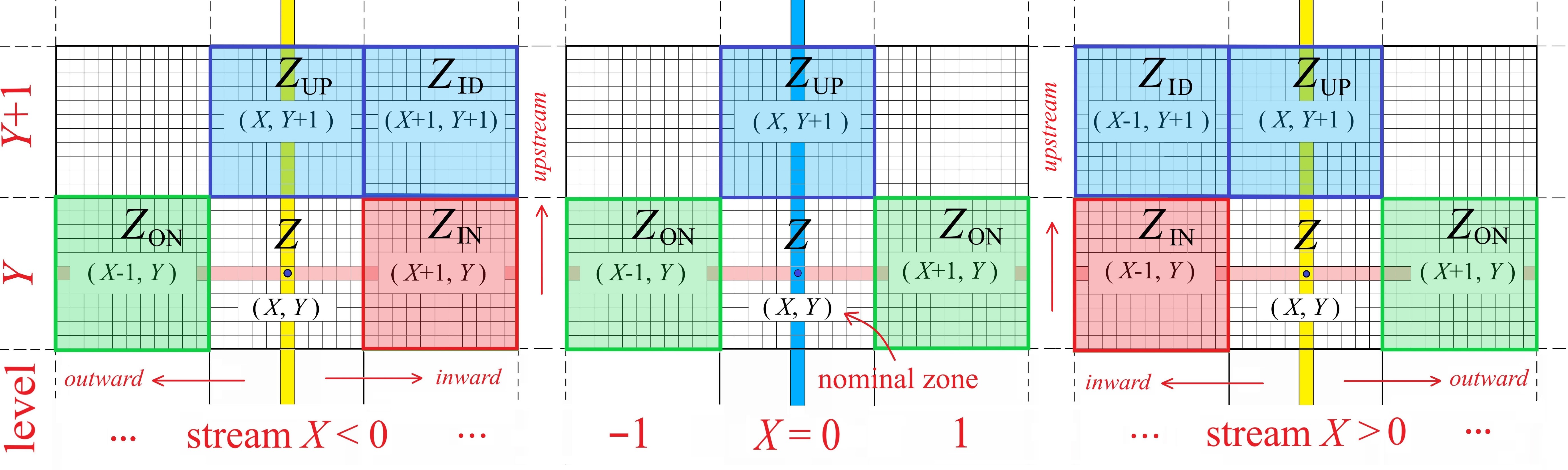}
    \caption{Zone neighborhood of three arbitrary zones $Z = (X,Y)$ in the grid, whose $X<0$, $X=0$,  and $X>0$, respectively. The congestion occurrences in respective \textit{upstream} zones $Z_{\mathrm{UP}}$ and $Z_{\mathrm{ID}}$ influence the decisions of UAS present in $Z$.}
    \label{fig: 4}
\end{figure}
 \begin{table}[h!]
 \caption{\label{tab:table1} Zone neighborhood and path notation for zone $Z$}
 \begin{tabular}{ll}
 \toprule[1pt]
$Z_{\mathrm{UP}}$ & the neighbor in the \textit{upstream} direction. \\\toprule[0.1pt]
$Z_{\mathrm{IN}}$ & the neighbor closer to nominal zone than zone $Z$ in the \textit{inward} direction. The \textit{inward} is a unit vector directed \\ & from $Z = (X,Y)$ towards the nominal zone $(0, Y)$.  \\\toprule[0.1pt]
$Z_{\mathrm{ID}}$ & the diagonal neighbor in the resultant \textit{inward}-\textit{upstream} direction.  \\\toprule[0.1pt]
$Z_{\mathrm{ON}}$ & the neighbor farther from nominal zone than zone $Z$ in the \textit{outward} direction. The \textit{outward} is a unit vector \\ 
& opposite to \textit{inward}, directed laterally away from the nominal zone $(0,Y)$. \\\toprule[0.1pt] 
$\alpha_Z$ & the path consisting $S$ cells that connects node $\widetilde{Z}$ with node $\widetilde{Z}_\mathrm{UP}$ in the \textit{upstream} direction.\\\toprule[0.1pt]
$\beta_Z$ & the path consisting $L$ cells that connects the boundary cell of ${Z}_\mathrm{IN}$ with node $\widetilde{Z}$ in the \textit{outward} direction.\\\toprule[0.1pt]
$\gamma_Z$ & the path consisting $L$ cells that connects node $\widetilde{Z}$ with the boundary cell of ${Z}_\mathrm{ON}$ in the \textit{outward} direction. \\\toprule[0.1pt]
$\widehat{\beta}_Z$& the set 
 of $L$ diagonal parallel paths that one-one connect cells in $\beta_Z$ path with cells of $\alpha_Z$ path in the \\
 &  \textit{outward-upstream} direction.\\\toprule[0.1pt]
$\widehat{\gamma}_Z$ & the set 
 of $L$ diagonal parallel paths that one-one connect cells in $\gamma_Z$ path with cells of $\alpha_Z$ path in the  \\
 & \textit{inward-upstream} direction. \\\toprule[0.1pt]
$\widehat{\delta}_Z$ & the diagonal path consisting $S$ cells that connects the node $\widetilde{Z}$ with node $\widetilde{Z}_\mathrm{ID}$ in the \textit{inward-upstream} direction.\\
\bottomrule[1pt]
 \end{tabular}
 \end{table}
 
The UTM provides a heading-reference directed graph to the UAS present in an arbitrary grid level, as shown in Fig. \ref{fig: 5a}. For the zone $Z = (X,Y),\ X \neq 0$, the graph comprises $\alpha_Z,\beta_Z,\gamma_Z,\widehat{\beta}_Z,\widehat{\gamma}_Z$ and $\widehat{\delta}_Z$ paths (notation described in Table. \ref{tab:table1}). The UAS motion is constrained by this heading-reference graph. For minimal time delays in the flight path, it is preferable that the UAS transition \textit{upstream} as close as possible to the nominal segment. When UAS is at node $\widetilde{Z}$ and $Z_{\mathrm{ID}}$ is uncongested, the $\widehat{\delta}_Z$ path is available for the UAS to shift segments and move closer to the nominal segment. When $Z_{\mathrm{UP}}$ is uncongested, the $\alpha_Z$ path is available to transition \textit{upstream}. In case both $Z_{\mathrm{UP}}$ and $Z_{\mathrm{ID}}$ are congested, the $\beta_Z\rightarrow\gamma_Z$ path   ($\beta_Z$ path $\rightarrow$ node $\widetilde{Z} \rightarrow\gamma_Z$ path, collectively) is the alternative path available for UAS to transition \textit{outward} and avoid the congested zones. If the $Z_{\mathrm{UP}}$ became uncongested while the UAS was in $\beta_Z$ or $\gamma_Z$ path, then $\widehat{\beta}_Z$ or $\widehat{\gamma}_Z$ paths are respectively available for UAS to revert towards $Z_{\mathrm{UP}}$.   

\begin{figure}[h!]
\centering
\begin{subfigure}[b]{0.55\textwidth}
     \centering
    \includegraphics[width = \linewidth]{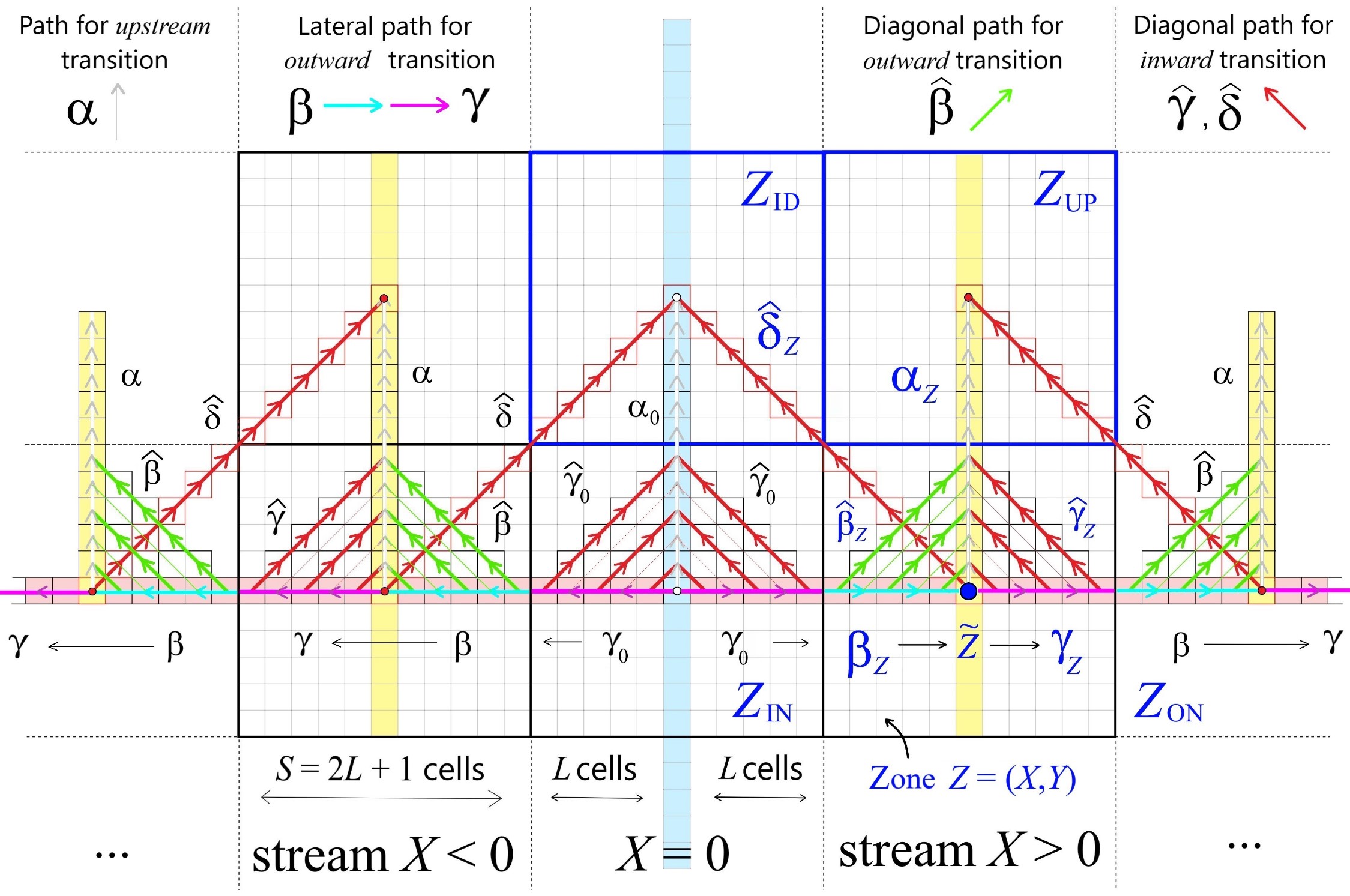}
    \caption{}
    \label{fig: 5a}
\end{subfigure}
     \hfill
\begin{subfigure}[b]{0.4\textwidth}
    \centering
    \raisebox{0.12\height}{\includegraphics[width = \linewidth]{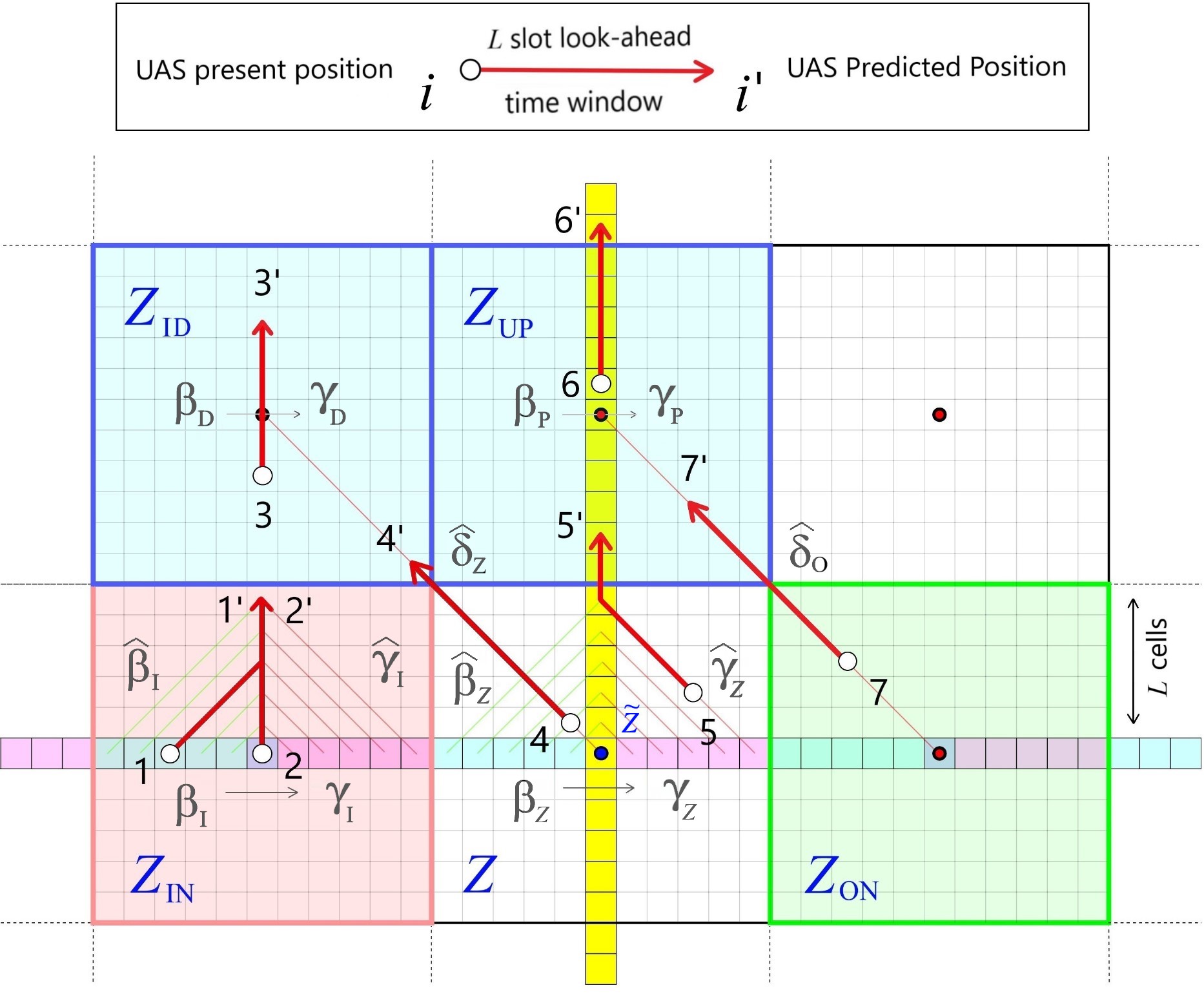}}
    \caption{}
    \label{fig: 5b}
\end{subfigure} 
\caption{a) Heading reference directed graph for UAS present in zone $Z$ in an arbitrary level $Y$. b) The $L = 5$ slot look-ahead predicted positions of UAS, which are restricted on the heading reference graph. If $M = 2$, then by Definition \ref{def: 2}, in the perspective of UAS present at node $\widetilde{Z}$, both zones $Z_{\mathrm{UP}}$ and $Z_{\mathrm{ID}}$ are congested (due to UAS 3, 4, 5, and 7). UAS present on $\beta\rightarrow\gamma$ path (UAS 1 and 2) do not contribute to congestion in $Z_{\mathrm{UP}}$ and $Z_{\mathrm{ID}}$.}
\end{figure}

As a special case, the graph for Stream$(0)$ zones comprises an $\alpha_0$ \textit{upstream} path, two $\gamma_0$ paths with opposite headings directed laterally \textit{outward} from $(0, Y)$ node, and $2L$ number of $\widehat{\gamma}_0$ \textit{inward} diagonal paths. For notation convenience, we use $\alpha_.,\beta_.,\gamma_.,\widehat{\beta}_.,\widehat{\gamma}_.,\widehat{\delta}_.$, with subscript $. =   \mathrm{I}, \mathrm{P}, \mathrm{O}, \mathrm{D}$ referring to the paths in the zones $Z_{\mathrm{IN}}, Z_{\mathrm{UP}}, Z_{\mathrm{ON}}, Z_{\mathrm{ID}} $, respectively. We ignore the subscript when referring to an arbitrary zone.  

The rule-based congestion mitigation strategy for the UAS foreseeing congestion in its $Z_{\mathrm{UP}}$ and $ Z_{\mathrm{ID}}$ zones is as follows. The following rules assume the UAS is transitioning in a Stream$(X)$ zone, $X = \{0,..., X_e\}$. The heading reference graph is symmetric about the nominal segment; hence similar rules when $X = \{-X_e,...,-1\}$ do hold.   

\begin{congestion}
    \textit{The UAS present in zone $Z = (X,Y)$ identify respective $Z_{\mathrm{UP}}$ and $Z_{\mathrm{ID}}$ zones as congested or uncongested following Definition \ref{def: 2}}.
    \label{rule: 1}
\end{congestion}
The UAS is restricted to transition on the heading-constrained graph. Under this structural constraint, among the UAS present in the zone neighborhood, the $L$ slot look-ahead predicted positions of only those UAS transitioning on $\alpha_Z,\widehat{\beta}_Z,\widehat{\gamma}_Z, \widehat{\delta}_\mathrm{O}$ paths lie in zone $Z_{\mathrm{UP}}$ and are responsible for congestion in $Z_{\mathrm{UP}}$. Similarly, predicted positions of UAS transitioning on $\widehat{\delta}_Z, \widehat{\beta}_\mathrm{I},\widehat{\gamma}_\mathrm{I}$ path lie in zone $Z_{\mathrm{ID}}$ (refer Fig. \ref{fig: 5b}).

\begin{congestion}
    In a given timeslot, if the UAS present in $\beta_Z \rightarrow \gamma_Z$ path finds zone $Z_{\mathrm{UP}}$ congested, then the UAS executes one lateral transition in the \textit{outward} direction.
    \label{rule: 2}
\end{congestion} 
In Fig. \ref{fig: 6a}, if $Z_{\mathrm{UP}}$ is congested in timeslot $k$, then in timeslot $k+1$ UAS in positions $1,6,9$ would be in positions $2,7,10,$ respectively. A zone congested at one timeslot may or may not be congested in the next. Only myopic position predictions ($1$ slot look-ahead) are possible for UAS transitioning on $\beta\rightarrow\gamma$, as it is uncertain whether these UAS would continue transitioning in $\beta\rightarrow\gamma$ direction or not. Hence, the predicted positions of UAS in ${\beta_.}\rightarrow{\gamma_.}$, where $. = Z,\mathrm{I},\mathrm{P},\mathrm{O},\mathrm{D}$, are excluded in Definition \ref{def: 2}. If a UAS that has entered the $\beta_Z$ path (UAS in position $1$ in Fig. \ref{fig: 6a}) finds the $Z_{\mathrm{UP}}$ zone congested for $L$ successive slots, it reaches node $\widetilde{Z}$. Similarly, if the UAS that has entered the $\gamma_Z$ path (UAS in position $7$ in Fig. \ref{fig: 6a}) finds the $Z_{\mathrm{UP}}$ zone congested for $L$ successive slots, the UAS ends in the $\beta_{\mathrm{O}}$ path. 
\begin{figure}[h!]
\centering
\begin{subfigure}[b]{0.49\textwidth}
    \centering
    \includegraphics[width = \linewidth]{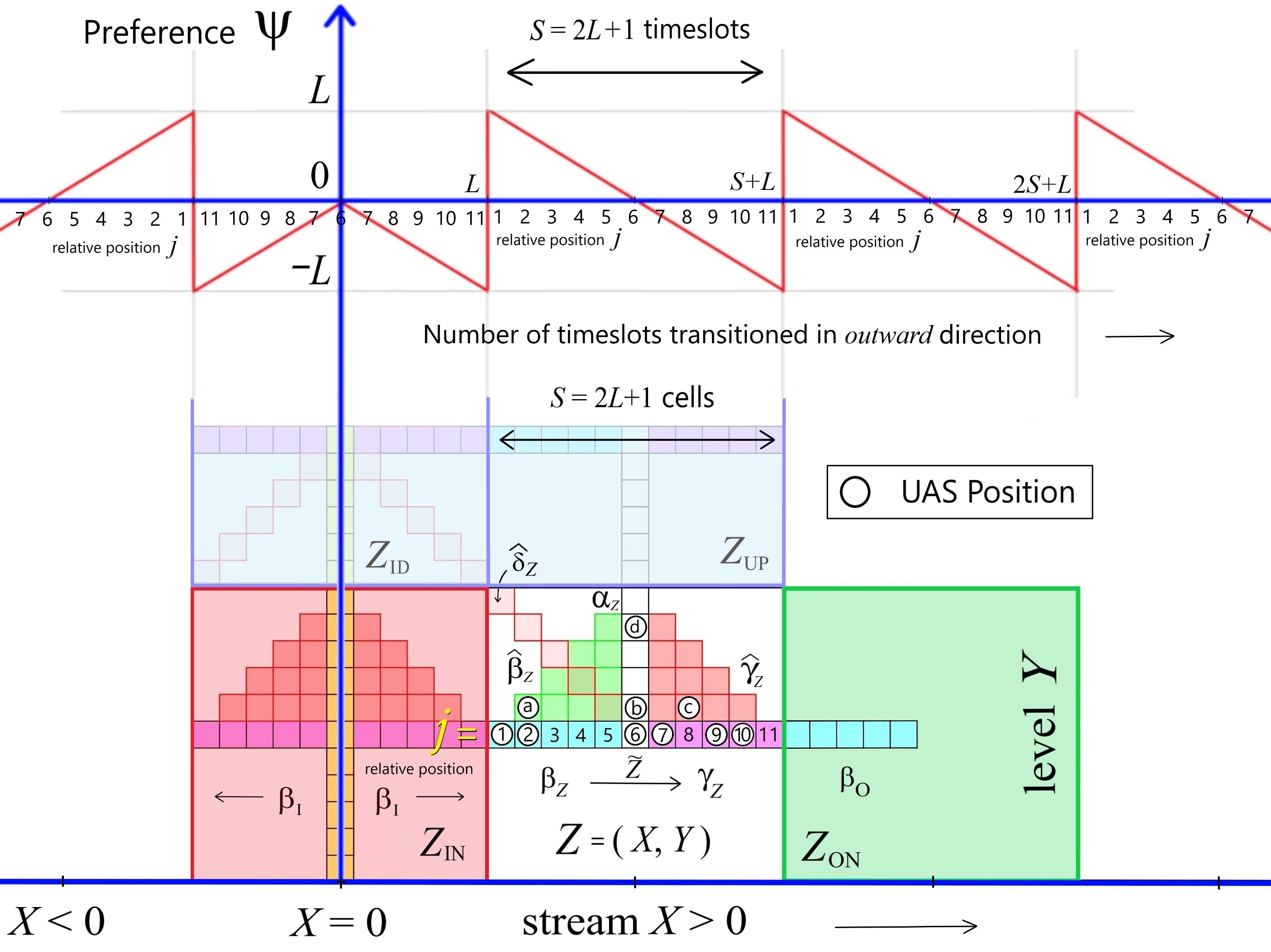}
    \caption{}
    \label{fig: 6a}
\end{subfigure}
     \hfill
\begin{subfigure}[b]{0.49\textwidth}
    \centering
    \raisebox{0.1\height}{\includegraphics[width = \linewidth]{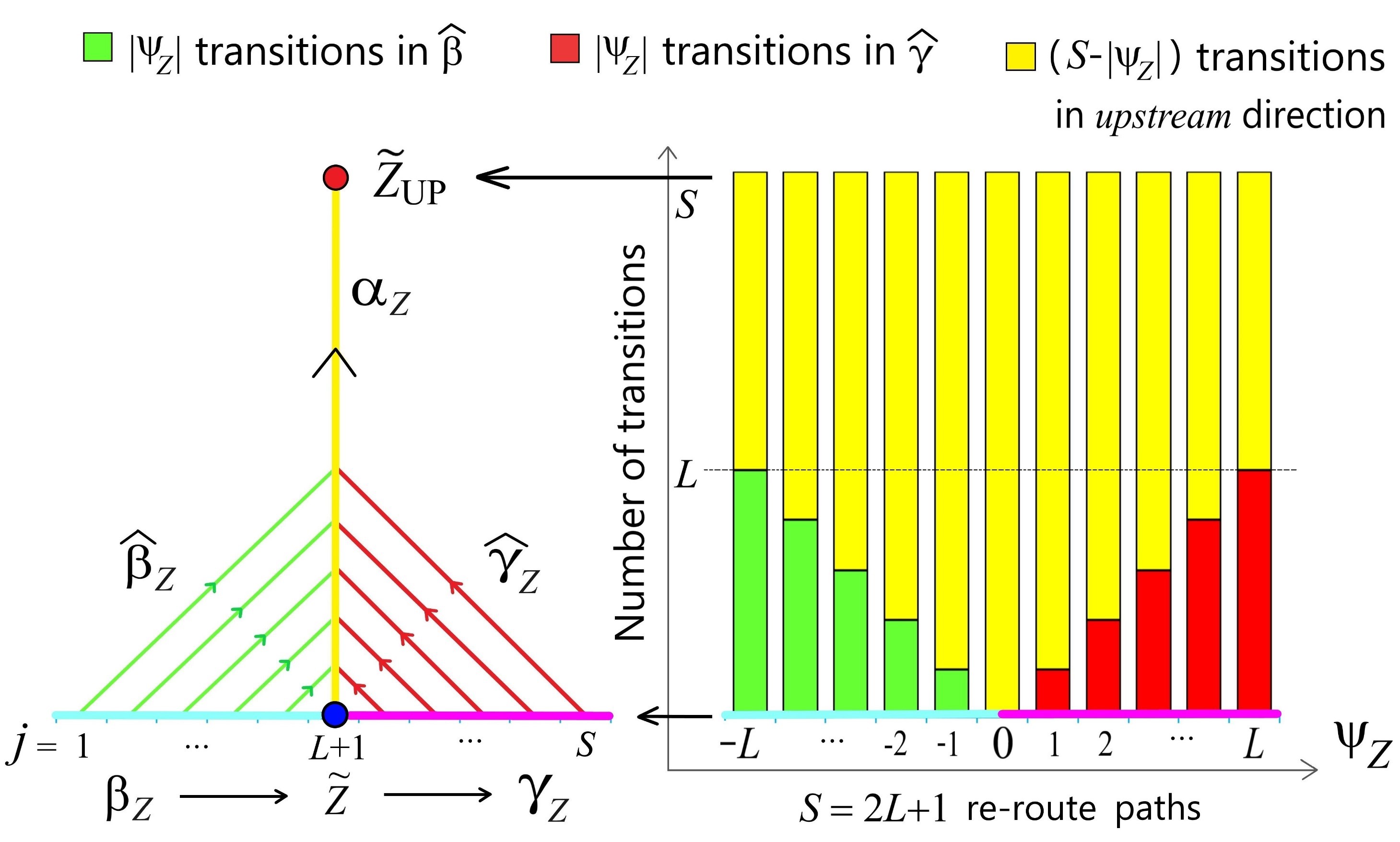}}
    \caption{}
    \label{fig: 6b}
     \end{subfigure}
    \caption{a) The relative position $j = \{1,..., S\}$ in the $\beta_Z\rightarrow\gamma_Z$ path is a cell index with respect to a fixed reference frame in the zone $Z$, where the first cell in $\beta_Z$ path, the node $\widetilde{Z}$, and the last cell in $\gamma_Z$ path are indexed $j = 1, L+1, S$, respectively. The UAS preference $\psi$ to move towards $Z_{\mathrm{UP}}$ varies with its relative position $j$. (Here, $L = 5$ slots, $S = 2L+1 = 11$ timeslots). b) $S$ paths connecting cells in the $\beta_Z\rightarrow\gamma_Z$ path with $\widetilde{Z}_{\mathrm{UP}}$ node.}
    
\end{figure}

If $Z_\mathrm{UP}$ becomes uncongested in timeslot $k$, then any two UAS in the $\beta_Z\rightarrow\gamma_Z$ path that simultaneously move towards $Z_\mathrm{UP}$ would experience conflict. For example, in Fig. \ref{fig: 6a},  UAS in position $2,6,10$ simultaneously reach position `d' by timeslot $k+4$ and conflict with each other. The above risk is overcome by assigning priorities to UAS present in the $\beta_Z\rightarrow\gamma_Z$ path. It is preferable to have self-enforcing priorities from the UAS perspective. The average number of UAS present in a congested zone and the duration for which the zone remains congested are both measures of congestion. A UAS in the \textit{downstream} of node $\widetilde{Z}$ has no prior information about congestion in $Z_{\mathrm{UP}}$ as it is unaffected by congestion until it reaches the node (position $6$ in Fig. \ref{fig: 6a}). However, on reaching $\widetilde{Z}$, each timeslot the UAS sees congestion in $Z_{\mathrm{UP}}$, the UAS executes an \textit{outward} transition on $\gamma_Z$ path, and accrues information about congestion in $Z_{\mathrm{UP}}$. Based on this information, the UAS expects the congestion duration in $Z_{\mathrm{UP}}$ to be at least greater than the time interval it has executed \textit{outward} transitions. Hence, the UAS that has transitioned longer on the $\gamma_Z$ path rather prefers to avoid  $Z_{\mathrm{UP}}$ and progress towards $Z_{\mathrm{ON}}$ as it expects longer congestion durations in $Z_{\mathrm{UP}}$. In contrast, the UAS that enter $\beta_Z$ path (UAS in position $1$ in Fig. \ref{fig: 6a}) have previously avoided $Z_{\mathrm{ID}}$ and have incurred delays, hoping to find $Z_{\mathrm{UP}}$ uncongested. Thus, if $Z_{\mathrm{UP}}$ is uncongested, the UAS in position $1$ actively prefers to move towards $Z_{\mathrm{UP}}$ in comparison to a UAS in the $\gamma_Z$ path. Within the zone $Z$, as UAS transitions on $\beta_Z\rightarrow\gamma_Z$ path, the UAS preference (denoted with $\psi_Z$) to move towards $Z_{\mathrm{UP}}$ linearly decreases with each successive \textit{outward} transition and is modeled as      
\begin{align}
    \psi_Z = L + 1 - j, \ \ \ j \in \mathbb{Z}_{[1, S]} \label{eqn: time-varying preference}
\end{align}
where $j$ is the relative position of UAS in the $\beta_Z\rightarrow\gamma_Z$ path. The variation of $\psi$ with relative position $j$ is shown in Fig. \ref{fig: 6a}. If a UAS has zero $\psi_Z$, it would mean the UAS is currently at node $\widetilde{Z}$. The UAS can execute \textit{upstream} transitions to reach $\widetilde{Z}_{\mathrm{UP}}$. The UAS with positive (negative) $\psi_Z$ can execute \textit{outward} (\textit{inward}) - \textit{upstream} transitions to reach the $Z_{\mathrm{UP}}$ zone. The zone geometry and positional information of the UAS present in the zone are known to every other UAS in the zone. The UAS can compute its preference as well as that of other UAS present in $\beta_Z\rightarrow\gamma_Z$ using Eqn. \ref{eqn: time-varying preference}, enabling the UAS to distributedly decide which UAS should proceed \textit{upstream} when the respective $Z_{\mathrm{UP}}$ zone becomes uncongested.

\begin{congestion}
    \textit{(Priority) If in a given timeslot the $Z_{\mathrm{UP}}$ zone becomes uncongested, then among the UAS transitioning on $\beta_Z\rightarrow\gamma_Z$  path, that UAS which has the highest preference $\psi_{max} \in \mathbb{Z}_{[-L, L]}$ would proceed towards the $Z_{\mathrm{UP}}$ zone. The remaining UAS execute one outward transition on respective paths.} 
    \label{rule: 3}
\end{congestion}
Applying Rule \ref{rule: 3} in Fig. \ref{fig: 6a}, if UAS are present in position $1,6,9$, then UAS in position $1$ proceeds to position `a', whereas UAS in position $6,9$ transition to position $7,10,$ respectively. The UAS in position $6,7$ proceed to position `b' if no UAS are present in positions $1$-$5$ and $1$-$6$, respectively.
For each relative position  $j \in \mathbb{Z}_{[1, S]}$, we have a piecewise re-route path (a sequence of the number of transitions in a specific heading direction) that one-one connects the cells in $\beta_Z\rightarrow\gamma_Z$ path with node $\widetilde{Z}_{\mathrm{UP}}$. Thus, a total of $2L+1 = S$ re-route paths exists. Let $J \subseteq \mathbb{Z}_{[1, S]}$ be the set of positions in the $\beta_Z\rightarrow\gamma_Z$ path where UAS are present.     
\begin{congestion}
    \textit{Following Rule \ref{rule: 3}, the UAS in the $\beta_Z\rightarrow\gamma_Z$ path that has the highest preference opts the below piecewise re-route path to reach $\widetilde{Z}_{\mathrm{UP}}$ node.}
    \begin{align}
        \psi_{\max} = \max\{\psi_Z(j)\ \vert\ j \in J \}
    \end{align}
    \begin{enumerate}
        \item \textit{If $\psi_{\max} = 0$, it means the UAS present at node $\widetilde{Z}$ has the highest preference. The UAS executes $S$ successive transitions in \textit{upstream} direction and reaches $\widetilde{Z}_{\mathrm{UP}}$.}
        \item \textit{When $\psi_{max}>0$, it means the UAS in the $\beta_Z$ path has the highest preference. The UAS executes $|\psi_{\max}|$ successive transitions in \textit{outward-upstream} direction followed by $S-|\psi_{\max}|$  \textit{upstream} transitions and reaches $\widetilde{Z}_{\mathrm{UP}}$. }
        \item \textit{When $\psi_{\max}<0$, it means the UAS in the $\gamma_Z$ path has the highest preference. The UAS executes $|\psi_{\max}|$ successive transitions in the \textit{inward-upstream} direction followed by $S-|\psi_{\max}|$ \textit{upstream} transitions and reaches $\widetilde{Z}_{\mathrm{UP}}$. }   
    \end{enumerate}
    \label{rule: 4}
\end{congestion}
 
Notice that irrespective of the UAS's position $j$ in $\beta_Z\rightarrow\gamma_Z$ path, the UAS takes $\vert \psi_{Z}(j) \vert + (S- \vert \psi_{Z}(j) \vert) = S$ transitions, that is $S$ timeslots to reach the $\widetilde{Z}_{\mathrm{UP}}$ node (refer Fig. \ref{fig: 6b}). 

The following rules are for handling contingencies and exceptional scenarios.

\begin{congestion}
    \textit{(Descend condition) Consider a UAS transitioning \textit{upstream} towards node $\widetilde{Z}$. On reaching node $\widetilde{Z}$, if the UAS finds the $Z_{\mathrm{ID}}$ zone uncongested, with no other UAS present in $\beta_{\mathrm{I}}\rightarrow\gamma_{\mathrm{I}}$ path that could possibly move towards $Z_{\mathrm{ID}}$, then the UAS at node $\widetilde{Z}$ executes $S$ successive transitions on $\widehat{\delta}_Z$ path and reaches the $\widetilde{Z}_{\mathrm{ID}}$ node.}
    \label{rule: 5}
\end{congestion}
Here, by descend, we mean the UAS shifts from the higher stream segment (the present parallel segment containing node $\widetilde{Z}$) to a lower stream segment (segment closer to the nominal segment containing node $\widetilde{Z}_{\mathrm{ID}}$).
While most conflict risks have been mitigated by imposing priority and descend condition, there still is a rare possibility that two UAS transitioning in the \textit{outward} and \textit{upstream} direction simultaneously reach the $\widetilde{Z}$ node in some timeslot $k$ and encounter conflict. Further, following rules \ref{rule: 2}-\ref{rule: 4}, these UAS will thereafter always conflict with each other for the same cell occupancy. The conflict is resolved as follows.

\begin{congestion}
\textit{(Conflict resolution) If a UAS in the $\beta_Z$ path is in conflict with a UAS heading \textit{upstream} at node $\widetilde{Z}$ (refer Fig. \ref{fig: 7}), then irrespective of congestion in $Z_{\mathrm{UP}}$, the UAS in the $\beta_Z$ path is directed onto the $\alpha_Z$ path. The UAS executes $S-1$ transitions in \textit{upstream} direction and reaches $\widetilde{Z}_{\mathrm{UP}}$. 
\begin{enumerate}
    \item if descend condition is satisfied:  the UAS heading \textit{upstream} is directed onto $\widehat{\delta}_Z$ path. It executes $S-1$ transitions in \textit{inward-upstream} direction and reaches $\widetilde{Z}_{\mathrm{ID}}$.
    \item if descend condition not satisfied: The UAS heading \textit{upstream} is directed onto the $\gamma_Z$ path and executes one lateral transition in \textit{outward} direction.  
\end{enumerate}}
\label{rule: 6}
\end{congestion}
The UAS is enclosed in a virtual cylindrical disk of a predefined safety radius. The cells are large enough that during the conflict resolution maneuver, the disks do not overlap. However, the emphasis is on the UAS encountering conflict in time instant $(k-1)\Delta T$ should reach the destined cell centers by time instant $(k+1)\Delta T$. The UAS is free to accelerate or decelerate about the nominal velocity to achieve the above time synchronization. In both Rule \ref{rule: 5} and \ref{rule: 6}, at timeslot $k$ when the descend condition is satisfied or there is a conflict being resolved at node $\widetilde{Z}$, the remaining UAS in the $\beta_Z\rightarrow\gamma_Z$ path execute one transition in the \textit{outward} direction, to avoid any conflict risk with the UAS proceeding towards $\widetilde{Z}_{\mathrm{UP}}$ and $\widetilde{Z}_{\mathrm{ID}}$. 

\begin{figure}[h!]
\centering
\begin{subfigure}{0.3\textwidth}
    \centering
    \includegraphics[width = \linewidth]{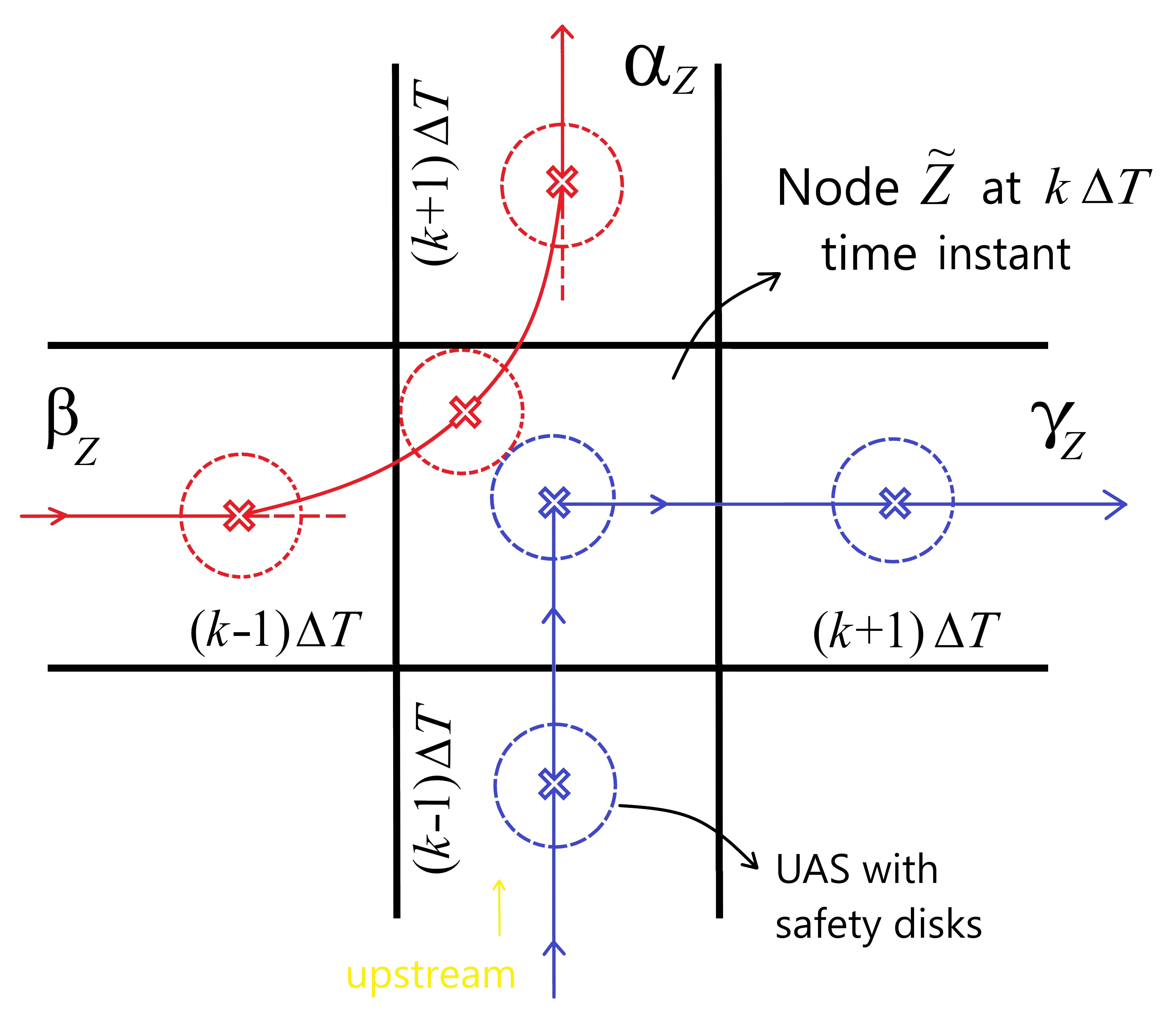}
    \caption{\textit{descend condition} not satisfied}
    \label{fig: 7a}
\end{subfigure}
\hspace{1cm}
\begin{subfigure}{0.3\textwidth}
    \centering
    \includegraphics[width = \linewidth]{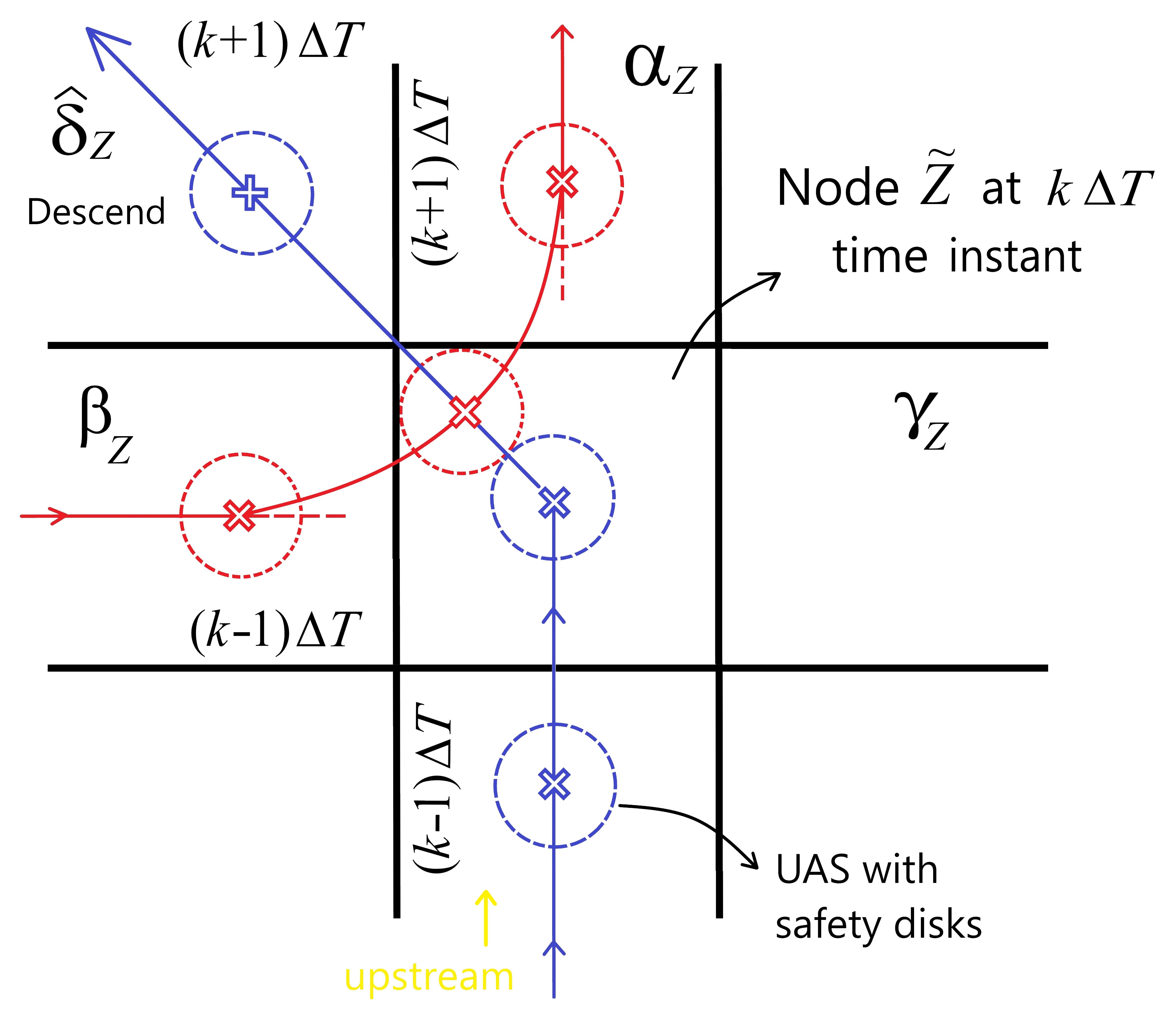}
    \caption{\textit{descend condition} satisfied}
    \label{fig: 7b}
\end{subfigure} 
    \caption{Conflict resolution maneuver executed by  the UAS at time instant $k\Delta T$ on reaching node $\widetilde{Z}$. }
    \label{fig: 7}
\end{figure}

 While UAS is transitioning on the heading-reference graph, the UAS may encounter zones that are occupied by static obstacles. The smallest rectangular zone cluster that encloses the static obstacle is demarcated as a no-fly zone.
\begin{congestion}
    \textit{(Presence of static obstacle) For a UAS in zone $Z$, if the $Z_{\mathrm{UP}}$ is a no-fly zone, then the $Z_{\mathrm{UP}}$ is treated as a congested zone. If the $Z_{\mathrm{ID}}$ or $Z_{\mathrm{IN}}$ zone is a no-fly zone, then the $Z_{\mathrm{ID}}$ is treated as congested zone. If the $Z_{\mathrm{ON}}$ zone is a no-fly zone, then irrespective of congestion in the $Z_{\mathrm{UP}}$ zone, $Z_{\mathrm{UP}}$ is always considered uncongested so that no UAS performs \textit{outwards} transitions and heads towards $Z_{\mathrm{ON}}$} (refer Fig. \ref{fig: 8}). 
    \label{rule: 7}
\end{congestion}
\begin{figure}[h!]
\centering
\includegraphics[width = 0.3\linewidth]{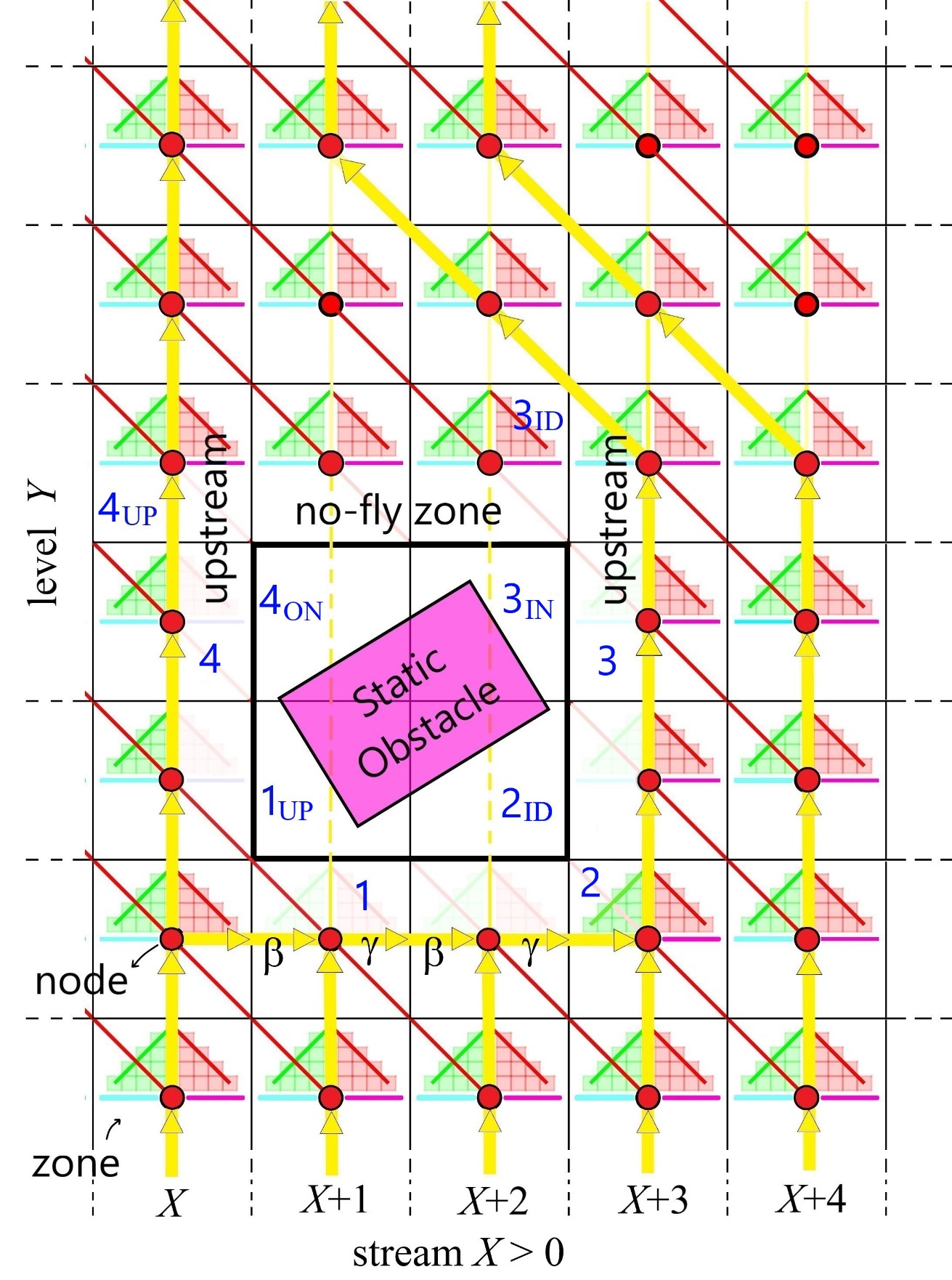}
\caption{For UAS in zones marked $Z = 1,2,3$ and $4$, corresponding $1_{\mathrm{UP}}$, $2_{\mathrm{ID}}$, $3_{\mathrm{IN}}$, and $4_{\mathrm{ON}}$ zones are no-fly zones, respectively. Following Rule \ref{rule: 2}, \ref{rule: 5} and \ref{rule: 7}, no UAS in zone $1$ proceeds towards $1_{\mathrm{UP}}$. No UAS in zone $2,3$ descends to $2_{\mathrm{ID}},3_{\mathrm{ID}}$, respectively. Zone $4_{\mathrm{UP}}$ is always perceived uncongested, hence UAS in zone $4$ would never progress towards $4_{\mathrm{ON}}$. }
\label{fig: 8}
\end{figure}

Consider the UAS in zone $Z$ whose $Z_{\mathrm{UP}}$ is a no-fly zone (zone marked $1$ in Fig. \ref{fig: 8}). As the $Z_{\mathrm{UP}}$ is always perceived congested, these UAS on reaching node $\widetilde{Z}$ would certainly transition along $\beta_Z \rightarrow \gamma_Z$ path every successive timeslot. Hence, the $L$-slot look ahead positions of these UAS can be determined and are included in the Definition \ref{def: 2}. 
\begin{congestion}
     \textit{When UAS in the nominal zone at node $\widetilde{Z} = (0, Y)$ finds the $Z_{\mathrm{UP}}$ zone congested, then the UAS randomly opts either of the two $\gamma_0$ paths and executes one lateral transition in the chosen \textit{outward} direction. The UAS opts for the right \textit{outward} direction with probability $\eta$ and the left \textit{outward} direction with probability $(1-\eta)$.}
     \label{rule: 8}
\end{congestion}
Here, $\eta \in [0,1]$ is a UTM design parameter to introduce bias in the traffic distribution about the nominal segment. It is useful when static obstacles or congestion regions due to exogenous UAS presence in the workspace are known \textit{a priori}.
\begin{congestion}
    \textit{Consider UAS transitioning in the Stream$(\pm1)$ zones (one in each zone). If the UAS simultaneously reach respective nodes $\widetilde{Z} = (\pm1, Y)$ and find the descend condition satisfied, then the probability that UAS in node $(1, Y)$ descends onto the nominal segment, and UAS in node $(-1, Y)$ does not descend is $\eta$. The probability that UAS in node $(-1, Y)$ descends, and UAS in node $(1, Y)$ does not descend is $(1-\eta)$.}
    \label{rule: 9}
\end{congestion}
 
When every UAS in the finite grid follows the proposed rule-based congestion mitigation strategy, the emerging UAS traffic behavior resembles a parallel air corridor network, with no more than $M$ UAS in any zone and the number of active corridor streams (that is, the UAS traffic spread) adapting to the UAS demand ($\lambda$). When a new source-destination route is being planned for a given $\lambda$, the nominal path for the route, the finite grid, the parameters $S$ (zone dimension), $M$ (minimum number of UAS in the zone for congestion), and $\eta$ are chosen such that the expected traffic spread on this route does not interfere with existing congested regions. The expected UAS traffic spread could be estimated using queuing theoretical models presented in section \ref{sec: congestion queuing model}. If $M$ is small, more UAS will transition \textit{outward}, leading to more traffic spread. However, $M$ can be increased to lower the spread at the cost of compromising intrinsic safety. If $\eta \in [0.5,1]$, then the probability that UAS opts right \textit{outward} direction is higher. More traffic will flow to the right of the nominal segment, and hence, there will be more traffic spread to the right. Similarly, if $\eta \in [0,0.5]$, the probability $(1-\eta)$ that UAS opts left \textit{outward} direction is higher. Hence there would be more traffic spread on the left of the nominal segment. 

For each source$(\mathsf{S})$-destination$(\mathsf{D})$ pair existing in the workspace, we have UAS traffic streams emanating at the source and terminating at the destination. The UAS traffic streams between existing $\mathsf{S}$-$\mathsf{D}$ pairs are exogenous traffic streams to the UAS on the new $\mathsf{S}$-$\mathsf{D}$ route. It is preferable that while planning the nominal path for the new route, the nominal path does not intersect exogenous traffic streams. However, in scenarios where a significant reduction in flight time is possible due to shorter path length if the nominal path were passing through exogenous traffic streams (refer Fig. \ref{fig: 9}a-\ref{fig: 9}c), we recommend that path planning impose the following assumptions. 
 
\begin{assumption}
    Each nominal segment in the nominal path can be intersected by at most one exogenous traffic stream. The nominal segment and the exogenous stream must intersect orthogonally. The start cell, the end cell, and the finite grid for the nominal segment must be chosen such that the exogenous stream does not pass through any $\beta\rightarrow\gamma$ path.
    \label{assumption: 1}
\end{assumption}
\begin{assumption}
    The inter-separation distances between UAS in the exogenous streams are geometrically distributed.
    \label{assumption: 2}
\end{assumption}
\begin{assumption}
    The traffic flow on the new source-destination route is affected by UAS in the exogenous traffic stream. However, the UAS in the exogenous traffic stream remain unaffected by the new incoming traffic flow.
    \label{assumption: 3}
\end{assumption}

The exogenous traffic stream could be in the workspace or out of the plane, as depicted by the air corridor in Fig. \ref{fig: 2} front view. The proposed mitigation strategy will be effective even if the above assumptions are violated, as long as the $L$ slot look-ahead positions of UAS in the exogenous streams can be predicted or communicated (refer Fig. \ref{fig: 9}d). Assumptions \ref{assumption: 1}-\ref{assumption: 3} have been stated to develop tractable queuing models for estimating expected traffic spread. Queuing analysis for the more general case is beyond the scope of this paper.
\begin{figure}[h!]
    \centering
    \includegraphics[width = \linewidth]{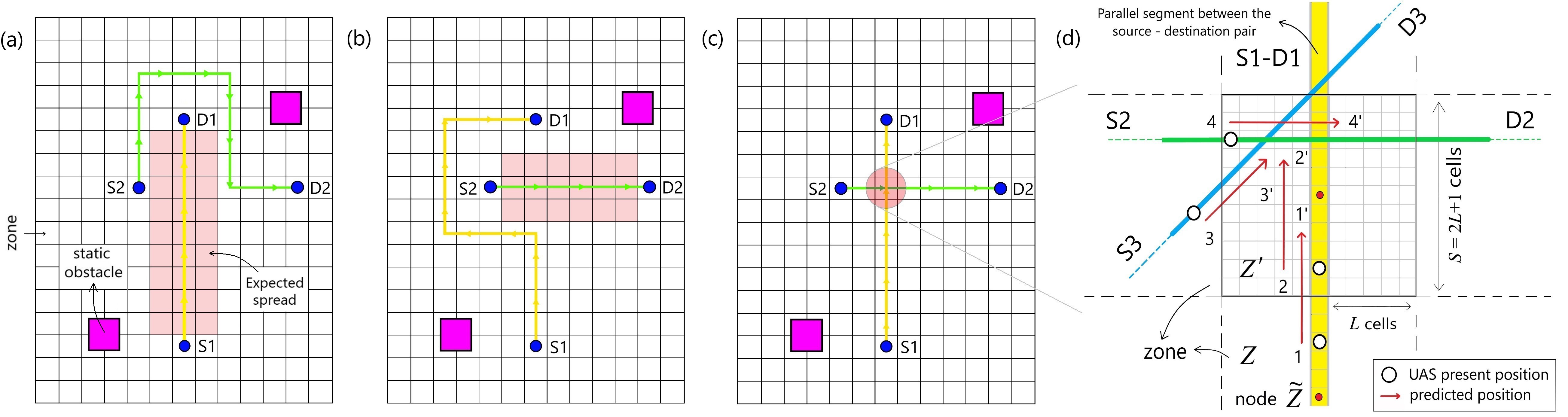}
    \caption{Path planning where the nominal path of a) the new route $\mathsf{S2}$-$\mathsf{D2}$ avoids the exogenous stream $\mathsf{S1}$-$\mathsf{D1}$, b) the new route $\mathsf{S1}$-$\mathsf{D1}$ avoids the exogenous stream $\mathsf{S2}$-$\mathsf{D2}$, c) the new route $\mathsf{S1}$-$\mathsf{D1}$ intersects the exogenous stream $\mathsf{S2}$-$\mathsf{D2}$. d) The UAS transitioning on $\mathsf{S2}$-$\mathsf{D2}$ and $\mathsf{S3}$-$\mathsf{D3}$ routes are exogenous traffic streams to UAS transitioning on the $\mathsf{S1}$-$\mathsf{D1}$ route. For $M \leq 4$, from the perspective of a UAS present at node $\widetilde{Z}$, the zone $Z_{\mathrm{UP}}$ is congested because of UAS 1,2,3 and 4. Note that the inter-separation distance between any two UAS present in the square zone is less than $S$ timeslots.}
    \label{fig: 9}
\end{figure}

The UAS transitioning on the new route does encounter conflicts with UAS in the exogenous traffic streams. However, for these UAS, the proposed congestion mitigation strategy guarantees no more than $M$ conflicts in any zone in the worst case. Here again, the conflicts are resolved by momentarily accelerating and decelerating the UAS, and once the conflict is resolved, the UAS must reach the destined cell center in the same timeslot. If the exogenous traffic stream is out of the plane, then in any timeslot, at most one exogenous UAS exists in the zone. However, if the exogenous stream is in the workspace, the stream intersects the zone orthogonally ($\mathsf{S2}$-$\mathsf{D2}$ shown in Fig. \ref{fig: 9}d), thus at most $S$ exogenous UAS can exist in any timeslot. If $\lambda_e$ is the probability that an exogenous UAS enters a zone in a given timeslot, then the number of exogenous UAS $(N_e )$ in the zone is binomially distributed, with the expected number  of exogenous UAS $(\mathrm{E}[{N}_e])$ in the zone given by,
\begin{align}
    &\mathrm{P}[N_e = n] = {S\choose n}{(\lambda_e)}^n (1-\lambda_e)^{S-n}\ \ , \ \ \mathrm{E}[{N}_e] = S\lambda_e
\end{align}
As a consequence of Assumption \ref{assumption: 3}, $M$ for the zones on the new route intersecting the exogenous stream must be greater than $\mathrm{E}[{N}_e]$. Only then will the UAS transitioning on the new route cross the exogenous stream. 
  
\section{Congestion queuing model} \label{sec: congestion queuing model}

We adopt the discrete-time queuing theory (\citet{atencia2013discrete,gao2004discrete,wittevrongel1999discrete, bruneel1994analysis}) for estimating the expected number of UAS present in any zone $Z = (X, Y)$, $X = \{-X_e,...,0,..., X_e\}$, $Y \in \{1,..., Y_e\}$. As discussed in Fig. \ref{fig: 5b}, for UAS transitioning in $\alpha_Z,\widehat{\beta}_Z, \widehat{\gamma}_Z$ and $\widehat{\delta}_{\mathrm{O}}$ paths, respective $L$ slot look-ahead positions are in zone $Z_{\mathrm{UP}}$ and are responsible for congestion. We associate the zone $Z_{\mathrm{UP}}$ with a queuing system $\overline{Z}$, that comprises cells belonging to $\alpha_Z,\widehat{\beta}_Z,\widehat{\gamma}_Z,\widehat{\delta}_{\mathrm{O}}$ and $\beta_Z\rightarrow\gamma_Z$ paths. The \textit{upstream} and the diagonal paths are servers, and the laterally \textit{outward} paths are queues. Following the rule-based congestion mitigation strategy, the UAS that enter this system of cells ($\overline{Z}$) are arrivals in queuing theory terminology. The \textit{downstream} UAS that enter the system $\overline{Z}$ via node $\widetilde{Z}$ are referred to as nodal arrivals. Along with nodal arrivals, UAS in zones $Z_{\mathrm{IN}}$ and $Z_{\mathrm{ON}}$ enter the system $\overline{Z}$ via a cell that is not $\widetilde{Z}$ node. These UAS are referred to as non-nodal arrivals. The nodal and non-arrivals are denoted with Bernoulli distributed random variables $\{a_{\alpha}, a_{\delta}\}$, and $\{a_{\mathrm{I}}, a_{\mathrm{O}}\}$, respectively, in Fig. \ref{fig: 10}. Note that $a_{.,k} = 1$ if there is an arrival in timeslot $k$, and $0$ otherwise. Applying Rule \ref{rule: 5}
(\textit{descend condition}) in any timeslot $k$, both $a_{\alpha, k},\ a_{\delta, k}$ would simultaneously never be 1. Thus, we
merge the nodal arrivals as $a_k = a_{\alpha,k} + a_{\delta_k},\ a_k \in \{0,1\}$.
\begin{figure}[h!]
    \centering
    \includegraphics[width = 0.65\linewidth]{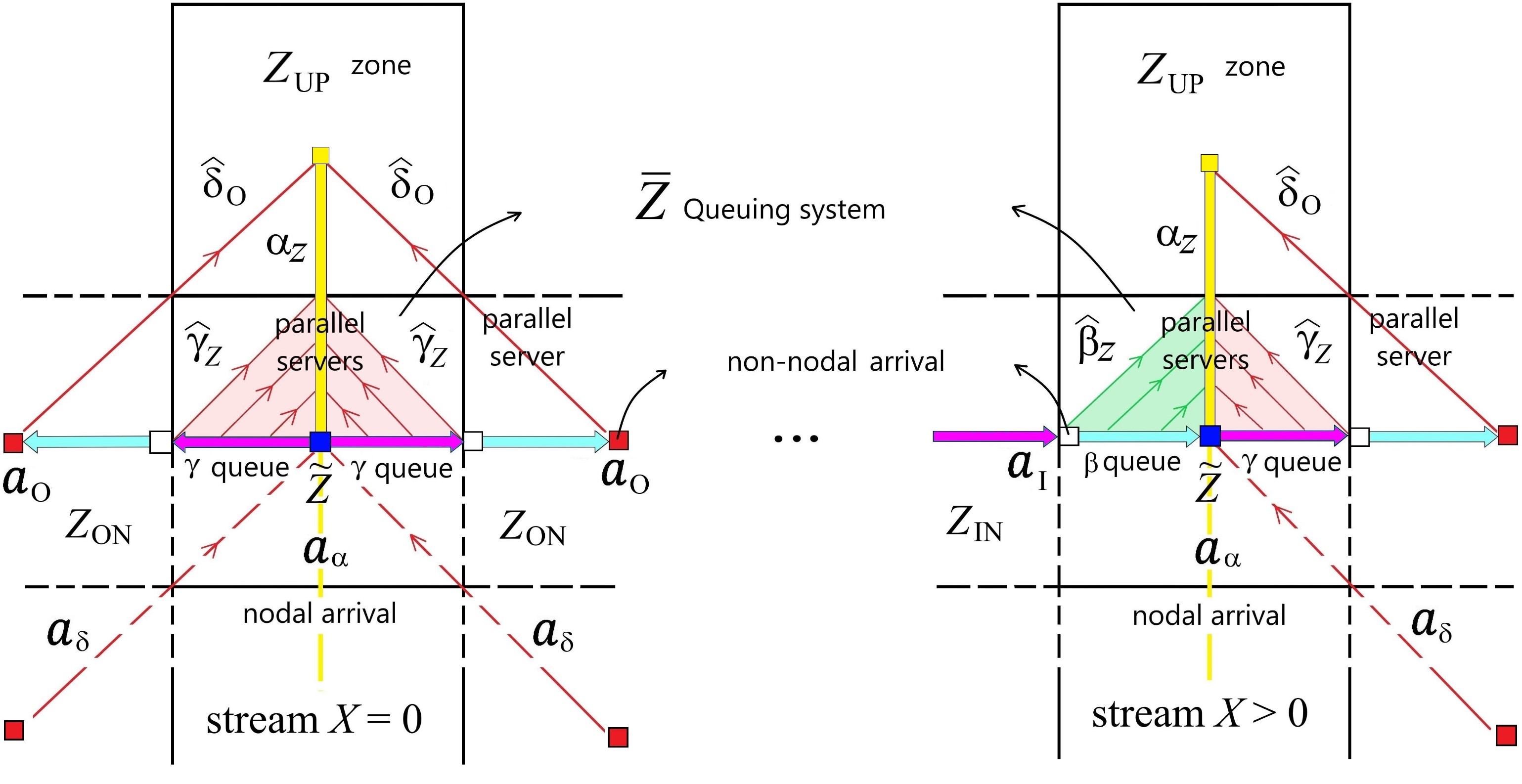}
    \caption{Queuing system $\overline{Z}$ associated with zone $Z_{\mathrm{UP}}$}
    \label{fig: 10}
\end{figure}

The $S$ paths summarized in Fig. \ref{fig: 6b} along with the $\widehat{\delta}_{\mathrm{O}}$ path are the $S+1$ parallel servers of the system. The UAS transitioning on these paths are said to be in service. The UAS that has completed $i$ transitions on any of the $S+1$ paths is said to have received $i$ slots of service, and on reaching the $\widetilde{Z}_{\mathrm{UP}}$ node is said to have departed the $\overline{Z}$ queuing system. A UAS entering service needs $S$ transitions on any of the $S+1$ paths to reach $\widetilde{Z}_{\mathrm{UP}}$ node. Thus, irrespective of the parallel server and the timeslot in which the UAS enters service, the service time is $S$ slots (that is $S\Delta T$ time units).  Let $w_{i,k} \in \{0,1\},\ i \in \{1,...,S-1\}$ are random variables such that, at timeslot $k$, $w_{i,k} = 1$ indicates that there exists a UAS that has received $i$ slots of service and $w_{i,k} = 0$ otherwise. Here, $w_{1,k} = 1$ implies that a UAS has entered service in the previous timeslot, and $w_{S-1,k} = 1$ implies that a UAS would depart the system $\overline{Z}$ in the next timeslot. A UAS that has received $i$ slots of service by timeslot $k$ would have received $i+1$ slots of service by $(k+1)$ timeslot. Hence,
\begin{align}
    w_{i+1,k+1} = w_{i,k}\ \ , \ \ W_{i+1}(z) = W_i(z)  
    \end{align}
where $W_i(z), \ i \in \{1,..., S-1\}$ is the probability generating function (PGF) of limiting probability distribution of the random variable $w_{i,k}$ defined as,   
\begin{align}
    &W_i(z) \triangleq P[w_{i,k} = 0] + z P[w_{i,k} = 1] \ , \ \ \ k \longrightarrow \infty 
    \label{eqn: received service}
\end{align}

The UAS in the exogenous traffic stream passing through zone $Z_{\mathrm{UP}}$ also contribute to congestion in $Z_{\mathrm{UP}}$. Corresponding cells through which the stream passes are appended in system $\overline{Z}$. By Assumption \ref{assumption: 2}, the inter-separation distances between exogenous UAS current positions, as well as between the $L$ slot look-ahead predicted positions, are geometrically distributed. Statistically, there is no difference between the counting processes related to when the present position and the predicted position enter the zone; hence, they can be interchangeably used for the exogenous stream. The UAS in the exogenous stream on entering zone $Z_{\mathrm{UP}}$ are called exogenous arrivals. According to Assumption \ref{assumption: 1}, the exogenous stream intersects the zone orthogonally. When the exogenous stream is in the workspace, the exogenous arrivals need $S$ transitions to exit the zone (that is, the service time is $S$ slots).  Let $e_{i,k} \in \{0,1\}, \ i \in\{0,1,...,S-1\}$ be a random variable such that, $e_{0,k} = 1$ indicates that exogenous arrival at timeslot $k$ has entered service. Further, $e_{i,k} = 1$ indicates that there exists an exogenous arrival that has received $i$ slots of service. When the exogenous stream is out of the plane, then $i = \{0\}$. If $\lambda_e$ is the exogenous arrival rate then,
\begin{align}
    &e_{i+1,k+1} = e_{i,k}\ \ , \ \ E_{i+1}(z) = E_i(z) = (1-\lambda_e) + z\ \lambda_e\ \ ,\ \ \ \ \ \ i \in \{0,...,S-1\}
\end{align}

At any given timeslot $k$, the random variables $w_{i,k},\ \ i \in \{1,...,S-1 \}$ are independent identically distributed (i.i.d). Note that $w_{1,k+1}$ is dependent on $e_{i,k}\ , \ i \in \{0,...,S-1 \}$ and $w_{i,k}$. But, $w_{i,k}$ is independent of $e_{i,k}$ at timeslot $k$. Knowing the PGF of the sum of independent random variables is the product of the PGF of individual random variables;   
the number of UAS in the system currently in service $u_k \in \{0,...,2S-1\}$ and PGF of its limiting probability distribution are given by,  
\begin{align}
    &u_k = \sum_{i = 1}^{S-1} w_{i,k} + \sum_{i = 0}^{S-1} e_{i,k}\ \ \ , \ \ \ U(z) = W_1(z)^{S-1} E_0(z)^{S} \ \ \ ,\ \ \  U(z) \triangleq \sum_{n = 0}^{2S-1} z^n P[u_k = n]\ \ \ , \ \ \ k \longrightarrow \infty
     \label{eqn: UAS in service}
\end{align}
From the congestion Definition \ref{def: 2}, $M$ is the minimum number of UAS predicted positions to be present in a zone for the zone to be defined as congested. If in the timeslot $k$, $M$ UAS are in service, then the queuing system is said to be busy; otherwise it is idle. The queuing system $\overline{Z}$ being busy implies $M$ UAS predicted positions exist in the zone $Z_{\mathrm{UP}}$. 
\begin{definition}
    If the queuing system $\overline{Z}$ is busy, then the associated zone $Z_{\mathrm{UP}}$ is congested.
    \label{def: 3}
\end{definition}
Let $\theta_k$ be a random variable such that $\theta_k = 0$ indicates the system is busy (that is associated $Z_{\mathrm{UP}}$ zone is congested) and $\theta_k = 1$ indicates the system is idle (that is, associated $Z_{\mathrm{UP}}$ zone is uncongested) at timeslot $k$. The PGF of the probability distribution of $\theta_k$, as $k\longrightarrow \infty$ is,   
\begin{align}
    &\theta_k = \begin{cases}
        0 & \text{if } u_k \geq M\\
        1 &  \text{if } u_k < M
    \end{cases} \ \ \ \ , \ \ \ \ 
    \Theta(z) \triangleq \sum_{n = M }^{2S-1} P[u_k = n] + z \sum_{n = 0}^{M-1} P[u_k = n]
    \label{eqn: congestion definition}
\end{align}
From Eqns. \ref{eqn: UAS in service}  and \ref{eqn: congestion definition}, the steady-state probability $\Theta_0$ that the zone $Z_{\mathrm{UP}}$ will be congested in timeslot $k$ is given by,
\begin{align}
   \Theta_0\ =\ P[\theta_k = 0]\ =\  1 - P[\theta_k = 1]\ =\ 1 - \Theta^{(1)}(0)\ =\ 1 - \sum_{n = 0 }^{M-1} P[u_k = n]\ =\ 1 - \sum_{n = 0}^{M-1} \cfrac{U^{(n)}(0)}{n!}
   \label{eqn: congestion probability}
\end{align}
where $(n)$ is the $n^{th}$ derivative with respect to $z$. Using Eqn. \ref{eqn: UAS in service} in Eqn. \ref{eqn: congestion probability} we get probability $\Theta_0$ in terms of $W_1(z)$.
\begin{align}
   \Theta_0 \ =\ & 1  - W_1(0)^{S-1} E_0(0)^S   - \sum_{n = 1}^{M-1} {S-1 \choose n}\mathlarger(1-W_1(0)\mathlarger)^n W_1(0)^{S-n-1} E_0(0)^S \notag \\  
   & - \sum_{n = 1}^{M-1} {S \choose n}\ W_1(0)^{S-1} E_0(0)^{S-n}\mathlarger(1-E_0(0)\mathlarger)^{n} \label{eqn: forward} 
\end{align}

When there are no exogenous arrivals, $E_0(0) = 1$ in Eqn. \ref{eqn: forward}. Here, $W_1(0) = P[w_{i,k} = 0]$ is the steady-state probability that a UAS arrival will not enter service, and $(1-W_1(0))$ is the steady-state probability that the arrival would enter service in a given timeslot. Eqn. \ref{eqn: forward} is a forward relation that captures UAS arrivals entering the service are responsible for congestion. 

If the $\overline{Z}$ system is busy in timeslot $k$, the system remains busy until at least one of the $M$ UAS departs service. The duration (in timeslots) for which the system remains busy is referred to as the congestion period. Following Rule \ref{rule: 2}, the UAS arrivals during the congestion period would successively transition in the $\beta_Z \rightarrow\gamma_Z$ path. These UAS in the $\beta_Z$ or $\gamma_Z$ paths are said to be queuing to enter service, with the number of transitions executed before entering service (in timeslots) being the waiting time. The motive is to find a feedback relation that captures congestion being responsible for an arrival entering the queue. It is evident from Fig. \ref{fig: 10} that queues associated with Stream$(0)$ zones behave differently from those of Stream$(X)$ zones. We present queuing models for Stream$(0)$ and Stream$(X\neq0)$ zones as follows, highlighting respective differences.

\subsection{Stream$\{0\}$ queuing system}

The schematic for the Stream$\{0\}$ queuing system associated with Stream$(0)$ zone is shown in Fig. \ref{fig: 11}. A Stream$\{0\}$ queuing system $\overline{Z}$ has nodal arrival $a_k$ and non-nodal arrival $a_{\mathrm{O},k}$. At the source node, $a_{k}$ arrivals signify UAS being deployed on receiving requests at rate $\lambda$. Thus, $a_{k}$ is Bernoulli distributed with arrival rate $\lambda$ and PGF $A(z) = (1-\lambda) + \lambda z$, whereas arrivals $a_{\mathrm{O}, k}= 0$ with probability 1 and PGF $A_{\mathrm{O}}(z) = 1$. At any other node, the arrival distributions need to be computed considering departures from respective neighboring queuing systems. The Stream$\{0\}$ queuing system has two $\gamma_0$ queues. According to Rule \ref{rule: 8}, if a UAS is queuing, then it is in either of the two queues with probability $\eta$ and $1-\eta$. We merge the two $\gamma_0$ queues and treat them as one. Similarly, following Rule \ref{rule: 9}, we merge the two $\widehat{\delta}_{\mathrm{O}}$ servers of Stream$\{0\}$ queuing system. 
\begin{figure}[h!]
    \centering
    \includegraphics[width = 0.7\linewidth]{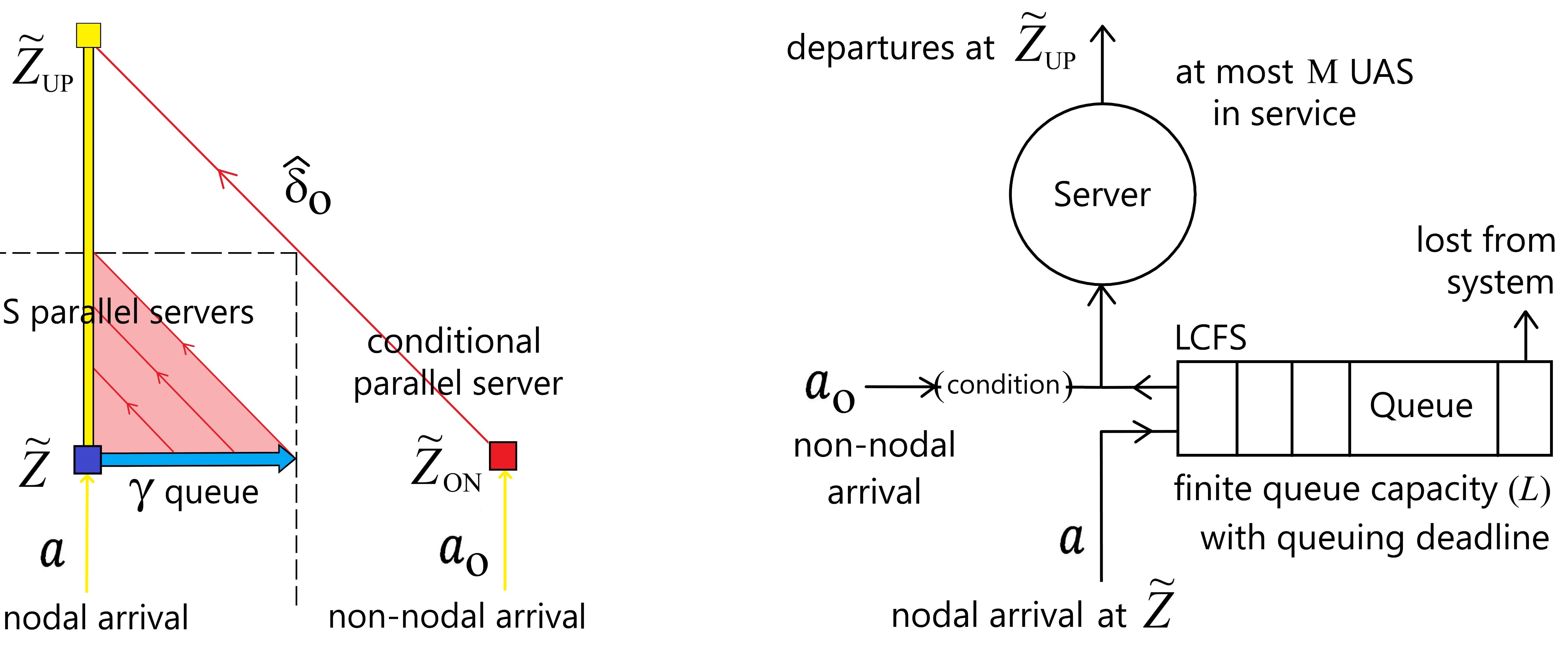}
    \caption{Stream$\{0\}$ queuing system}
    \label{fig: 11}
\end{figure}

Let $v_{j,k} \in \{0,1,...,j\}, \ j \in \{1,..., L\}$ be random variables such that at timeslot $k$, $v_{j,k}$ gives the total number of UAS currently in the queue that have queued for no more than $j$ timeslots since entering the queue. Here, $v_{1,k} = 1$ implies that an arrival found the system busy in timeslot $k-1$, so the arrival entered the queue and has queued for $1$ timeslot since then. $v_{L,k}$ gives the total number of UAS present in the queue. A UAS that has queued for $L$ timeslots (that is, executed $L$ successive transitions on $\gamma_0$ path) exits the $\overline{Z}$ system in the next timeslot. Hence, there is a queuing deadline of $L$ timeslots to enter service, beyond which the UAS is lost from the system. 
As the queue never goes unbounded, the limiting probability distributions as $k\longrightarrow\infty$ exist.

The UAS preference $\psi_Z \in \mathbb{Z}_{[-L,0]}$ decreases with each queued timeslot. A UAS that has queued for the least number of timeslots has the highest preference. Hence, by priority (Rule \ref{rule: 3}), if the $\overline{Z}$ system is idle ($Z_{\mathrm{UP}}$ is uncongested), then UAS with the highest preference enters service. This is fundamentally the Last Come, First Serve (LCFS) discipline. The following dynamical equations summarize the queuing behavior when the system is busy and idle.

\noindent If the system is busy ($\theta_k = 0$) then,
\begin{align}
    v_{1,k+1} = a_ k  \ \ \ \ , \ \ \ \ v_{j+1,k+1} = v_{j,k} + a_ k\ \ \ \ , \ \ \ \  j \in \{1,...,L-1\} \label{eqn: UAS queuing 1}
\end{align}
If the system is idle ($\theta_k = 1$) then following LCFS, we have
\begin{align}
    &v_{1,k+1} = (a_k - 1)^+ = 0\ \ \ \ , \ \ \ \ v_{j+1,k+1} = (v_{j,k} + a_k - 1)^+ \ \ \ \ , \ \ \ \  j \in \{1,...,L-1\} \label{eqn: UAS queuing 4} 
\end{align}
where the notation $(.)^+ = \max\{0,.\}$. Eqn. \ref{eqn: UAS queuing 4} states that when $a_k = 0$, then one UAS from the queue enters service (the one which arrived last). Hence, the number of UAS that have queued at most $j$ slots ($j \in \{1,...,L\}$) reduces by one. When $a_k = 1$, then the arrival enters service, and the number of UAS that have queued $j$ slots by timeslot $k$ would have queued $j+1$ slots by timeslot $k+1$. 
The PGF $V_j(z),\ j \in \{1,..., L\}$ can be derived using the conditional PGF of Eqns. \ref{eqn: UAS queuing 1}-\ref{eqn: UAS queuing 4} as,
\begin{align}
    V_j(z)\ =\ E[z^{v_{j,k+1}}\vert \theta_k = 0] P[\theta_k = 0]  + E[z^{v_{j,k+1}}\vert \theta_k = 1] P[\theta_k = 1] 
\end{align}
The arrivals $a_{k}$ and $v_{j,k}$ are independent at timeslot $k$. Substituting the respective PGFs for $j \in \{1,...,L-1\}$, we get
\begin{align}
    &V_1(z) = A(z)\Theta_0 + (1 - \Theta_0) \ \ \ \ , \ \ \ \ V_{j+1}(z) 
     = \ A(z)V_{j}(z) \Theta_0 + \left[ V_j(0) + \left(\cfrac{V_j(z) - V_j(0)}{z}\right)A(z) \right] (1- \Theta_0) \label{eqn: rouche}
\end{align}

By recursively substituting $V_j(z)$ in $V_{j+1}(z)$, we get $V_L(z)$.  In the timeslot $k$, a UAS will enter service if the system is uncongested and, more importantly, there are UAS available to enter service, that is,
\begin{enumerate}
\item there is a UAS nodal arrival (that is, $a_k = 1$ with probability $(1-A(0))$), or
\item there are UAS in queue (that is, $v_{L,k} > 0$ with probability $(1-V_L(0))$), or 
\item there is a UAS non-nodal arrival (that is, $a_{\mathrm{O},k} = 1$ with probability $(1-A_{\mathrm{O}}(0))$) available to descend towards $Z_{\mathrm{UP}}$ under the condition that there are no UAS nodal arrival and no UAS in the queue (\textit{descend condition}),
\end{enumerate}
As $k \longrightarrow \infty$, the probability $\pi$ that a UAS is available to enter service in a given timeslot is given by, 
\begin{align}
    \sigma_Z = A(0)V_L(0) \ \ \ , \ \ \ \pi\ =\ (1-A(0)) + A(0)(1 - V_L(0)) +  \sigma_Z(1-A_{\mathrm{O}}(0)) =\ (1-\sigma_Z) + \sigma_Z(1-A_{\mathrm{O}}(0)) \label{eqn: UAS available 0}
\end{align}
where $\sigma_Z$ is referred to as the steady-state descend probability that a UAS at node $\widetilde{Z}_\mathrm{ON}$ would descend into $Z_{\mathrm{UP}}$ zone. The PGF of the descend probability distribution is denoted by $B(z) = (1-\sigma_Z) + z \sigma_Z$. The steady-state probability that a UAS would enter service is,
\begin{align}
        \mathrm{P}[w_{1,k} = 1] = (1-\Theta_0)\pi \ \ \ \ ,\ \ \ \ W_1(0) = 1 - \mathrm{P}[w_{1,k} = 1]  \label{eqn: 0 feedback}
\end{align}
Eqn. \ref{eqn: 0 feedback} is the feedback relation that captures congestion being responsible for UAS entering the Stream$\{0\}$ queue. Substituting Eqn. \ref{eqn: 0 feedback} in Eqn. \ref{eqn: forward}, we get a polynomial equation in $\Theta_0$ of degree $(L+1)(S-1)$. There is an indeterminate term ($0/0$) in Eqn. \ref{eqn: rouche}, which we take advantage of using Rouché's theorem. Knowing the steady state exists and $\vert z \vert < 1$, according to Rouché's theorem (\citet{kobayashi1977queueing,bruneel1994analysis}), the roots of the denominator of Eqn. \ref{eqn: rouche} are also roots of the numerator, which leads to pole-zero cancellations and the polynomial in $\Theta_0$ having a zero in $[0,1]$ (existence proof provided in Theorem \ref{thm: Existence} in Appendix). The polynomial equation is solved numerically for $\Theta_0$. Substituting $\Theta_0$ in $\pi,\sigma_Z,V_{j}(z), j \in \{1,L\}$, we get the respective probabilities.

Note that the congestion definition in Eqn. \ref{eqn: forward} is defined solely on $M$ UAS being in the service or not. However, it does not account for when these $M$ UAS have entered service. It would be misleading to assume no correlation between the time instances the UAS enters service, as this affects the congestion period duration. For example, assume $M = 2$ and service time $S = 11$ timeslots. Consider in timeslot `$0$', $Z_{\mathrm{UP}}$ is uncongested, and a UAS enters service. If a second UAS enters service in timeslot `$1$', then as $2$ UAS are in service, the $Z_{\mathrm{UP}}$ becomes congested in timeslot `$1$'. The first UAS will depart the system after receiving $11$ timeslots of service, and with the UAS departing, the $Z_{\mathrm{UP}}$ becomes uncongested in timeslot `$11$'. Thus, in a time interval of $11$ timeslots, the $Z_{\mathrm{UP}}$ is uncongested for one timeslot, and then it remains congested for the next $10$ timeslots. Instead, if the second UAS had entered service in timeslot `$10$', then the $Z_{\mathrm{UP}}$ turns congested in timeslot `$10$'. By timeslot `$10$', the first UAS would have already received $10$ timeslots of service and, in the next timeslot, would depart, leaving the system behind uncongested. In this case, the $Z_{\mathrm{UP}}$ is uncongested for $10$ timeslots, and it remains congested for only one timeslot. We introduce the above correlation by modeling congestion as a Markov modulated on/off regular process (MMRP) (\citet{smiesko2023markov}).
\begin{figure}[h!]
    \centering
    \begin{subfigure}{0.25\textwidth}
      \centering 
      \includegraphics[width = \linewidth]{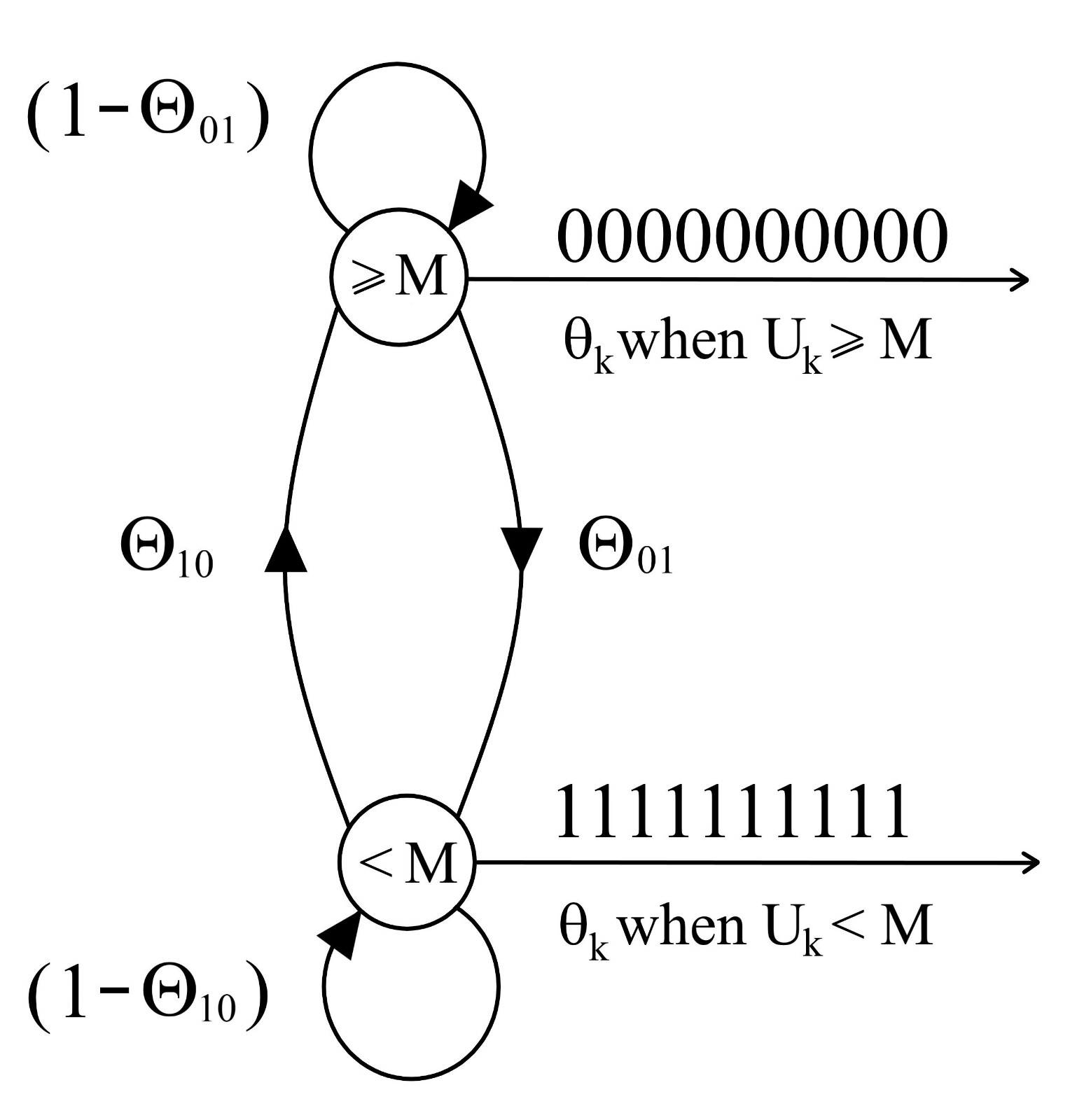}
    \caption{}
    \label{fig: 12a}
    \end{subfigure} \hfill
    \begin{subfigure}{0.65\textwidth}
        \centering
        \raisebox{0.05\height}{\includegraphics[width = \linewidth]{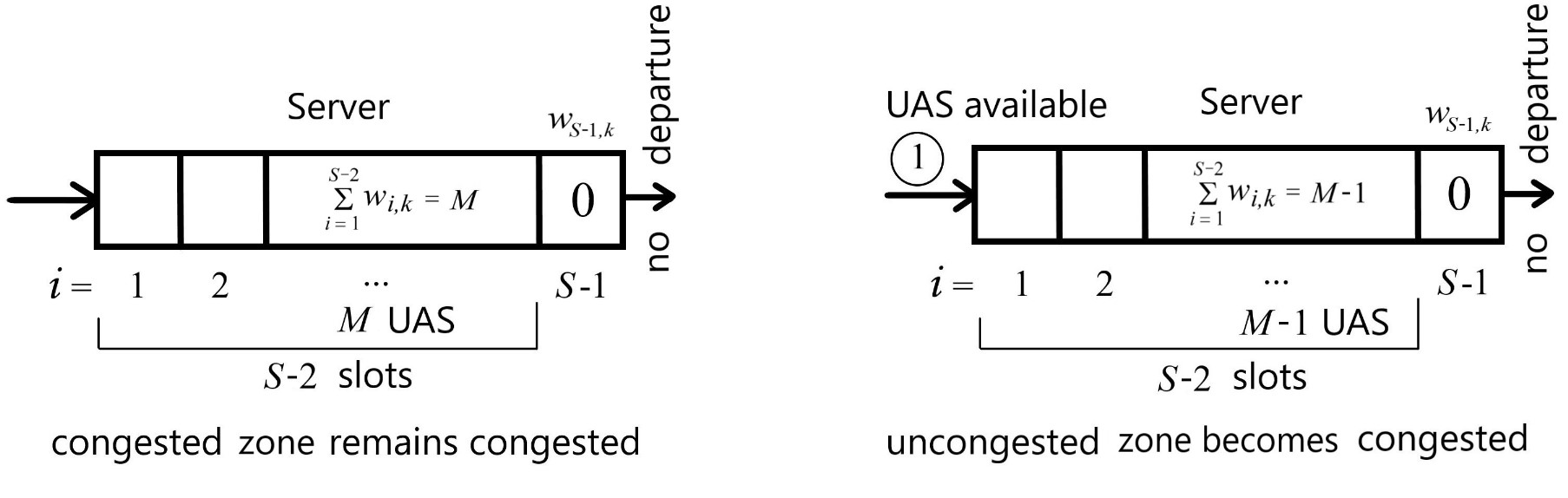}}
    \caption{}
    \label{fig: 12b}
    \end{subfigure}
    \caption{a) Transition diagram for congestion $\theta_k \in \{0,1\}$ modeled as a Markov modulated on/off regular process. b) Condition under which the $Z_{\mathrm{UP}}$ if congested (server busy) remains congested, and if uncongested (server idle) becomes congested the next timeslot. }
\end{figure}

\subsubsection{Markov-modulated regular process}
An MMRP (shown in Fig. \ref{fig: 12a}) is a two-state Markov chain, where the state of the Markov chain depends only on the state of a two-state background process. Here, $\theta_k \in \{0,1\}$ is the Markov state, with $\Theta_0$ and $\Theta_1 = 1-\Theta_0$ being the probability that $Z_{\mathrm{UP}}$ is congested and uncongested, respectively. The number of UAS ($n_k$) present in service is the background process with two states $n_k\geq M$ and $n_k<M$. If $n_k\geq M$ then $\theta_k = 0$ and if $n_k < M$ then $\theta_k = 1$. Let $\Theta_{10} = \mathrm{P}[n_{k+1} \geq M \vert n_k < M] = \mathrm{P}[\theta_{k+1} = 0\vert\theta_k = 1]$ is the transition probability that $Z_{\mathrm{UP}}$, if uncongested in timeslot $k$ becomes congested in timeslot $k+1$, and $\Theta_{01} = \mathrm{P}[n_{k+1} < M \vert n_k \geq M] = \mathrm{P}[\theta_{k+1} = 1\vert\theta_k = 0]$ is the transition probability that $Z_{\mathrm{UP}}$, if congested in timeslot $k$ becomes uncongested in timeslot $k+1$, respectively. Here, $\Theta_{00} = (1-\Theta_{01})$ is the transition probability that $Z_{\mathrm{UP}}$ remains congested and $\Theta_{11}=(1-\Theta_{10})$ is the transition probability that it remains uncongested. Let $\pi_e = (\pi + \lambda_e - \pi\lambda_e)$ denote the steady-state probability that a UAS is available to enter service (including exogenous arrivals) in any given timeslot. 

If $M$ UAS were available to enter service in a time interval of $S-1$ timeslots, then the $Z_{\mathrm{UP}}$ would be congested in the $(S-1)$th timeslot. The $Z_{\mathrm{UP}}$ would remain congested if no UAS are departing in the next timeslot; that is, there are no UAS that have received $S-1$ slots of service as shown in Fig. \ref{fig: 12b}. This event can be summarized in Eqn \ref{eqn: congested-congested}.
\begin{align}
    &\Theta_{00}\ =\ \cfrac{\mathrm{P}[\theta_{S} = 0, \theta_{S-1} = 0]}{\mathrm{P}[\theta_{S-1} = 0]} \ =\ \cfrac{\left[{S-2\choose M} \pi_e^{M} (1-\pi_e)^{S-M-2}\right](1-\pi_e)}{\left[ {S-1 \choose M}  \pi_e^M (1-\pi_e)^{S-M-1} \right]}\ =\ \cfrac{S-M-1}{S-1}
    \label{eqn: congested-congested}
\end{align}

 If $M-1$ UAS were available to enter service in the time interval of $S-1$ timeslots, then the $Z_{\mathrm{UP}}$ would be uncongested in the $(S-1)$th timeslot. Given $Z_{\mathrm{UP}}$ is uncongested, and if none of the $M-1$ UAS already in service departs, and at least one UAS is available to enter the service, then the $Z_{\mathrm{UP}}$ would become congested in the next timeslot (shown in Fig. \ref{fig: 12b}). This situation can be summarized in Eqn \ref{eqn: uncongested-congested}.
\begin{align}
    \Theta_{10}\ =\ \cfrac{\mathrm{P}[\theta_{S} = 0, \theta_{S-1} = 1]}{\mathrm{P}[\theta_{S-1} = 1]} \ =\ \cfrac{\pi_e\left[{S-2\choose M-1}\pi_e^{M-1} (1-\pi_e)^{S-M-1}\right] (1-\pi_e) }{\sum_{n = 0}^{M-1} \left[{S-1\choose n} \pi_e^{n} (1-\pi_e)^{S-n-1}  \right]}
    \label{eqn: uncongested-congested}
\end{align}

The Markov modulated congestion probability following MMRP derivations in \citet{smiesko2023markov} is given by,
\begin{align}
    \Theta_0^*\ =\ \Theta_{00} \Theta_0 + \Theta_{10} \Theta_1 \ \ \ \ , \ \ \ \
    \Theta_1^*\ =\ \Theta_{01} \Theta_0 + \Theta_{11} \Theta_1 \label{eqn: MMRP}
\end{align}

\subsubsection{Markov modulated Bernoulli process (MMBP)} 
In a given timeslot, whether a UAS would enter service or not is correlated to the $Z_{\mathrm{UP}}$ being congested and UAS being available. If $\theta_k = 0$, then $w_{1,k} = 0$. However, when $\theta_k = 1$, if UAS is available then $w_{1,k} =1$ with probability $\pi$ else $w_{1,k} = 0$ with probability $(1-\pi)$. Note that $w_{1,k}$ is an MMBP, which refers to the Markov chain that generates a sequence of zeros when the background process is in one state and generates a Bernoulli process when the background process is in the second state  (as shown in Fig. \ref{fig: 13}). Here, $w_{1,k}\in \{1,0\}$ is the two-state Markov chain with $\theta_k\in\{0,1\}$ being the background process. The Markov modulated probability that a UAS would enter service following MMBP derivations in \citet{smiesko2023markov} is given by
\begin{align}
   \mathrm{P}[w_{1,k} = 1]\ =\ 1 - W_1^*(0)\ =\ [\ \Theta_{11} \Theta_1^* +   \Theta_{01} \Theta_0^*\ ] \pi \label{eqn: MMBP}
\end{align}
\FloatBarrier
\begin{figure}[h!]
    \centering
    \includegraphics[width = 0.25\linewidth]{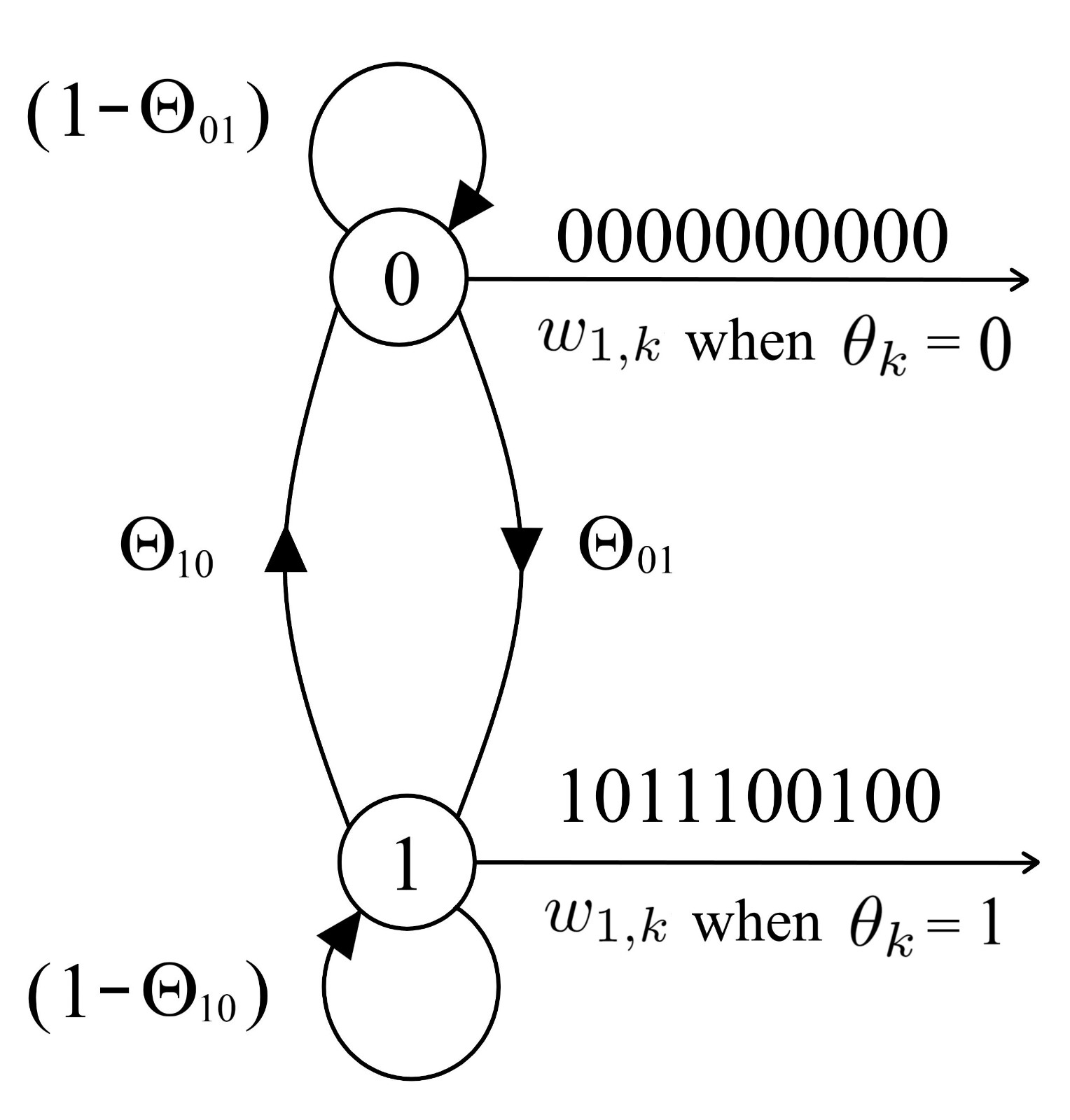}
    \caption{Transition diagram for $w_{1,k} \in \{0,1\}$ modeled as a Markov modulated Bernoulli process.}
    \label{fig: 13}
\end{figure}

The average number of UAS present in the service and the queue can be obtained by substituting $\Theta_0^*$ and $W_1^*(0)$ in $U(z)$ and $ V_L(z)$ and computing $U^{(1)}(1)$ and $V_L^{(1)}(1)$, respectively. The sum of the average number of UAS in the queue and the service is the expected number of UAS present in zone $Z_{\mathrm{UP}}$.

\subsubsection{Queue overflow}
The UAS that have queued for $L+1$ slots are lost from the queue and become non-nodal arrivals $a_{\mathrm{I},k}$ to the neighboring Stream$\{\pm1\}$ queueing systems. The steady-state probability that in timeslot $k$, a UAS queuing would be lost from the system is called Queue overflow probability $\Phi$. Let $q_k \in\{0,1\}$ be a random variable such that in the given timeslot $k$, $q_k = 1$ denotes that there exists a UAS that has queued $L$ timeslots and $0$ otherwise.
\begin{align}
    v_{L,k} &= v_{L-1,k} + q_k
\end{align}
As $v_{L,k} \in \{0,1,...,L\} \ ,\ v_{L-1,k} \in \{0,1,...,L-1\}$, we have, 
\begin{align}
    &\mathrm{P}[v_{L,k} = j]\ =\ \mathrm{P}[v_{L-1,k} = j - 1]\ \mathrm{P}[q_{k} = 1]\ + \mathrm{P}[v_{L-1,k} = j]\ \mathrm{P}[q_{k} = 0] \ \ \ , \ \ \ j = 1,...,L-1 \label{eqn: Lost UAS 1} \\
    &\mathrm{P}[v_{L,k} = L]\ =\ \mathrm{P}[v_{L-1,k} = L-1]\ \mathrm{P}[q_{k} = 1] \label{eqn: Lost UAS 2}
\end{align}
Summing Eqns. \ref{eqn: Lost UAS 1}-\ref{eqn: Lost UAS 2} for $j = 1,...,L$ and simplifying we get,
\begin{align}
    &\sum_{j = 1}^{L} \mathrm{P}[v_{L,k} = j]\ =\ \mathrm{P}[q_{k} = 1]\ \sum_{j = 1}^{L} \mathrm{P}[v_{L-1,k} = j-1]  +  \mathrm{P}[q_{k} = 0]\ \sum_{j = 1}^{L-1} \mathrm{P}[v_{L-1,k} = j-1] \\
    &(1-V_L(0))\ =\ \mathrm{P}[q_{k} = 1] + \mathrm{P}[q_{k} = 0] (1-V_{L-1}(0)) \\
    &\implies \mathrm{P}[q_{k} = 1] = \left(\cfrac{V_{L-1}(0) - V_{L}(0)}{V_{L-1}(0)}\right) \label{eqn: Lost UAS 3}
\end{align}
The UAS corresponding to $q_k = 1$ has the least preference ($\psi_Z = -L$) amongst other UAS present in the queue $(\psi_Z \in \{0,-1,...,-L+1\})$, including the nodal arrival. If the $Z_{\mathrm{UP}}$ is uncongested in timeslot $k$ and there is at least one UAS in the queue that has preference $\psi_Z > -L$. Then, by priority, the UAS that has the higher preference enters service, forcing the UAS corresponding to $q_k = 1$ to queue in the given timeslot. This UAS would have queued for $L+1$ slots and hence is lost from the system. In steady-state, when $Z_{\mathrm{UP}}$ is congested, the fraction of UAS arrivals that entered the queue but failed to enter service when $Z_{\mathrm{UP}}$ was uncongested are lost during the congestion period. Notice that when the $Z_{\mathrm{UP}}$ is congested, as well as when uncongested, the UAS may exit the system. Thus, the Queue overflow probability is given by,
\begin{align}
    \Phi \ =\ (1-\Theta_0^*)\left[\ 1-A(0)\ V_{L-1}(0)\ \right]\ \mathrm{P}[q_{k} = 1]\  +\ \Theta_0^*\ (1-A(0))\ (1-V_L(0))\ W_1^*(0)\ +\ \xi(1-A(0)) \label{eqn: 0 overflow}
\end{align}
The queue overflows and departures are correlated processes that are better modeled as a higher-order Markov model (order $L$ or higher) with covariate dependencies in transition probabilities (\citet{ATAHARULISLAM2006477}). However, for ease of analysis, we use a first-order Markov model and compensate for modeling errors by introducing an empirical correction factor $\xi(x)$ as in Eqn. \ref{eqn: 0 overflow}. The factor $\xi(x)$ is significant only when the UAS arrival rate $x$ is low.  
\begin{align}
    \xi(x) = 0.15\left(\cfrac{x}{M-1 + \zeta}\right)\exp \left( {-S\left(\cfrac{x}{M-1 + \zeta}\right)^3} \right) \approx (2-10\%) x\ \ ,\  \ \zeta = \begin{cases}
        2\eta & \text{if\ } \eta \in [0,0.5] \\
        2(1-\eta) & \text{if\ } \eta \in [0.5,1]
    \end{cases} \label{eqn: correction factor}
\end{align}

Following Rule \ref{rule: 8}, the steady-state probability distribution of the non-nodal arrivals $a_{\mathrm{I},k}$ for the neighboring Stream$\{-1\}$ queuing system is given by 
\begin{align}
    A_{\mathrm{I}}(z) = (1- \eta \Phi) + z \eta \Phi \label{eqn: left non-nodal arrival}
\end{align}
The non-nodal arrival probability distribution $A_{\mathrm{I}}(z)$ for neighboring Stream$\{+1\}$ system is given by
\begin{align}
    A_{\mathrm{I}}(z) = (1- (1-\eta)\Phi) + z (1-\eta)\Phi \label{eqn: right non-nodal arrival}
\end{align}

\subsection{Stream$\{X\}$ queuing system, $X\neq 0$}
The schematic for the Stream$\{X\}$ queuing system associated with Stream$(X)$ zone is shown in Fig. \ref{fig: 14}. A Stream$\{X\}$ queuing system $\overline{Z}$ has nodal arrivals $\{a_k\}$ and non-nodal arrivals $\{a_{\mathrm{I},k}, a_{\mathrm{O},k}\}$. These arrivals are departures and overflows from the preceding systems. Assume that the PGF of these arrivals have already been computed. Let $b_k$ be a random variable such that $b_k = 1$ if a nodal arrival in timeslot $k$ would descend to $Z_{\mathrm{ID}}$ zone and is $0$ otherwise. Corresponding descend probability $\sigma_{\mathrm{I}}$ and PGF $B(z) = (1-\sigma_{\mathrm{I}}) + z\sigma_{\mathrm{I}}$ are available, as $\sigma_{\mathrm{I}}$ is dependent only on the arrivals and UAS queuing in the $\overline{Z}_{\mathrm{IN}}$ system (refer Eqn. \ref{eqn: UAS available 0}). The non-nodal arrival $a_{\mathrm{I},k}$, on finding the $\overline{Z}$ system busy, would enter the $\beta_Z$ queue. The nodal arrival $a_k$, on finding the $\overline{Z}$ system busy and \textit{descend condition} unsatisfied, would enter the $\gamma_Z$ queue. 
\begin{figure}[h!]
    \centering
    \includegraphics[width = \linewidth]{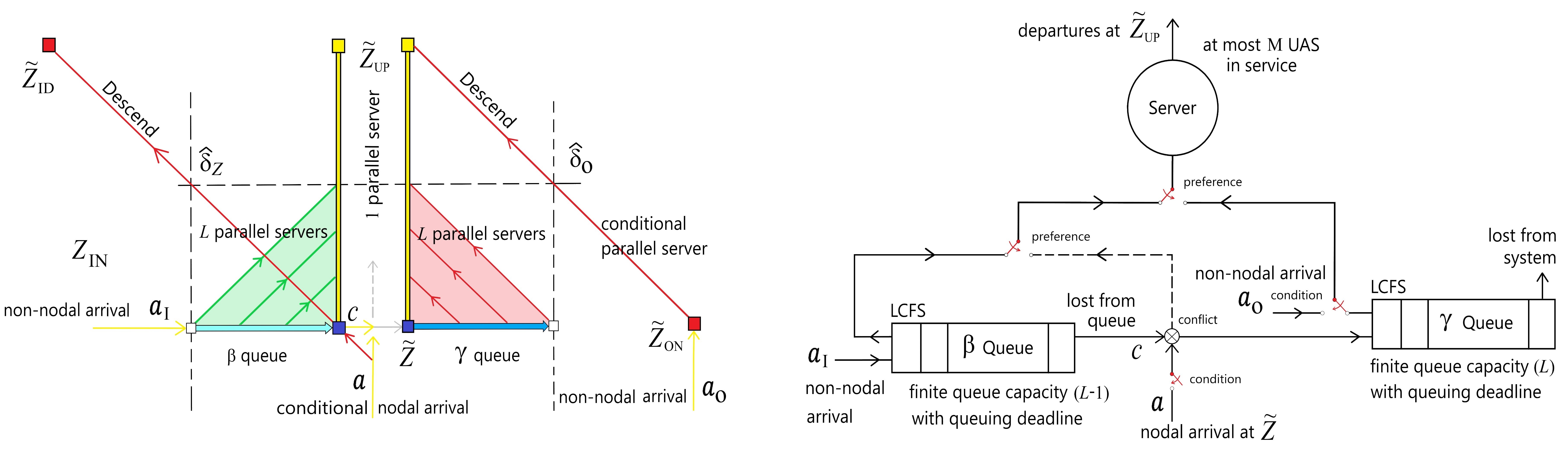}
    \caption{Stream$\{X\}$ queuing system, $X\neq 0$ }
    \label{fig: 14}
\end{figure}

Let ${}^\beta v_{j,k} \in \{0,1,...,j\},\ j \in \{1,..., L-1\}$ be random variables such that at timeslot $k$, ${}^\beta v_{j,k}$ gives the total number of UAS currently in the $\beta_Z$ queue that have queued for no more than $j$ slots since entering the queue. Here, ${}^\beta V_j(z)$ are respective PGFs and, ${}^\beta v_{L-1,k}$ is the total number of UAS queuing in the $\beta_Z$ queue. The non-nodal arrival $a_{\mathrm{I},k}$ on finding $Z_{\mathrm{UP}}$ zone congested for $L$ successive slots would enter the node $\widetilde{Z}$. Thus, for UAS queuing in the $\beta_Z$ queue, there is a $L-1$ slot queuing deadline to enter service, beyond which these UAS are lost from the $\beta_Z$ queue. The UAS that exits the $\beta_Z$ queue re-enters the system via node $\widetilde{Z}$. Hence, these UAS are also treated as nodal arrivals $c_k$ to the system with PGF $C(z)$. If there are simultaneous $c_k$ and $a_k$ nodal arrivals in timeslot $k$, then there is conflict at node $\widetilde{Z}$. In conflict, the nodal arrival $c_k$ enters service (Rule \ref{rule: 6}). Whereas the non-nodal arrival $a_{\mathrm{I},k}$ and the UAS queuing in $\beta_Z$ queue in the given timeslot. In timeslots with no conflict, the UAS in the $\beta_Z$ queue that has queued for the least number of slots has the highest preference and enters service. This is the LCFS discipline. 

Let ${}^\beta q_k \in\{0,1\}$ be a random variable that denotes if there exists a UAS that has queued for $L-1$ slots in $\beta_Z$ queue in timeslot $k$. When $Z_{\mathrm{UP}}$ is uncongested, the UAS corresponding to ${}^\beta q_k = 1$ is lost from the $\beta_Z$ queue if,
\begin{enumerate}
    \item there are non-nodal arrivals ($a_{\mathrm{I},k} = 1$) or
    \item there are UAS in $\beta_Z$ queue (${}^\beta v_{L-2,k} > 0 $) that have higher preference than the UAS corresponding to ${}^\beta q_k = 1$.
\end{enumerate}

In steady-state, when $Z_{\mathrm{UP}}$ is congested, the fraction of UAS arrivals that entered the $\beta_Z$ queue but failed to enter service when $Z_{\mathrm{UP}}$ was uncongested are lost during the congestion period. The above behavior is similar to that of Stream$\{0\}$ queue, where $c_k$ arrivals are queue overflows from $\beta_Z$ queue. Thus, similar to Eqns. \ref{eqn: Lost UAS 3}-\ref{eqn: 0 overflow}, the PGF of nodal arrivals $c_k,\ k \longrightarrow \infty$ is given as follows. 
\begin{align}
    {}^\beta v_{L-1,k} &=\ {}^\beta v_{L-2,k} + {}^\beta q_k \ \ \ \ , \ \ \ \
    \mathrm{P}\left[{}^\beta q_k = 1\right] = \left( \cfrac{{}^\beta V_{L-2}(0) - {}^\beta V_{L-1}(0)}{{}^\beta V_{L-2}(0)} \right) \\
    \mathrm{P}[c_k = 1] &= (1-\Theta_0)\left[  1 - A_{\mathrm{I}}(0) {}^{\beta} V_{L-2}(0) \right] \mathrm{P}\left[{}^\beta q_k = 1\right]  + \Theta_0 (1-A_{\mathrm{I}}(0)) \left(1-{}^\beta V_{L-1}(0)\right) W_1(0) + \xi(1-A_{\mathrm{I}}(0)) \label{eqn: conflict arrival 1} \\
    C(z) &=  \left(1- \mathrm{P}[c_k = 1]\right) + z\ \mathrm{P}[c_k = 1] \label{eqn: conflict arrival 2}
\end{align}

The LCFS dynamic equations for UAS queuing in the $\beta_Z$ queue are as follows.
If the system is busy ($\theta_k = 0$) or there is a conflict ($a_k = c_k = 1$) then
\begin{align}
    {}^{\beta} v_{1,k+1} = a_{\mathrm{I},k}  \ \ \ \ , \ \ \ \ {}^{\beta} v_{j+1,k+1} = {}^{\beta} v_{j,k} + a_{\mathrm{I},k}\ \ \ \ ,\ \ \ \ j \in \{1,...,L-2\}  \label{eqn: gamma queuing 1}
\end{align}
If the system is idle ($\theta_k = 1$) and there is no conflict ($a_k = c_k = 0$ or $a_k \neq c_k$) then
\begin{align}
    {}^{\beta} v_{1,k+1} = (a_{\mathrm{I},k} - 1)^+ = 0  \ \ \ \ , \ \ \ \ {}^{\beta} v_{j+1,k+1} = ({}^{\beta} v_{j,k+1} +a_{\mathrm{I},k} - 1)^+\ \ \ \ ,\ \ \ \ \ j \in \{1,...,L-2\}  \label{eqn: gamma queuing 3} 
\end{align}
Eqns. \ref{eqn: gamma queuing 1}-\ref{eqn: gamma queuing 3} are similar to Eqns. \ref{eqn: UAS queuing 1}-\ref{eqn: UAS queuing 4}. Thus, $ {}^\beta V_j(z),\ j \in \{1,..., L-1\}$ for the $\beta$ queue can be derived using the conditional PGF as,
\begin{align}
    &\Omega = (1-C(0))(1-A(0)) \label{eqn: conflict probability}\\
    &{}^\beta V_1(z) = A_{\mathrm{I}}(z)(\Theta_0 + \Omega - \Theta_0\Omega)  + (1 - \Theta_0)(1-\Omega) \label{eqn: similar equation 1}\\
    &{}^\beta V_{j+1}(z) 
     = A_{\mathrm{I}}(z) {}^{\beta}V_{j}(z) (\Theta_0 + \Omega - \Theta_0\Omega) + \left[{}^\beta V_j(0) + \left(\cfrac{{}^\beta V_j(z) - {}^\beta V_j(0)}{z}\right)A_{\mathrm{I}}(z) \right] (1- \Theta_0)(1-\Omega) \label{eqn: similar equation 2} 
\end{align}

where $\Omega$ is the conflict probability. The conflict probability $\Omega$ is dependent on the probability $(1-C(0))$, which is not known. If $\Omega = 0$, then Eqns. \ref{eqn: similar equation 1}-\ref{eqn: similar equation 2} are same as the Eqn. \ref{eqn: rouche} of Stream$\{0\}$ queuing system. Assuming $\Omega = 0$, by substituting $A(z) = A_{\mathrm{I}}(z)$ in Eqn. \ref{eqn: rouche} and solving for queue overlow probability $\Phi$ using Eqns. \ref{eqn: rouche}-\ref{eqn: 0 overflow}, we get a good estimate for the probability $(1 -C(0)) = \Phi$. Then, conflict probability $\Omega$ can be computed using Eqn. \ref{eqn: conflict probability}. By recursively substituting ${}^\beta V_j(z)$ in ${}^\beta V_{j+1}(z)$, we get the PGF ${}^\beta V_{L-1}(z)$. The nodal arrivals $c_k$ and $a_k$ that fail to enter service in timeslot $k$ would enter the $\gamma_Z$ queue in the next timeslot. Let ${}^\gamma v_{j,k}  \in \{0,1,...,j\},\ j \in \{1,..., L\}$ be random variables such that at timeslot $k$, ${}^\gamma v_{j,k}$ give the total number of UAS currently in the $\gamma_Z$ queue, that have queued for no more than $j$ timeslots since entering the queue. Here, ${}^\gamma V_j(z)$ are respective PGFs and, ${}^\gamma v_{L,k}$ is the total number of UAS queuing in the $\gamma_Z$ queue. The queuing behavior of the $\gamma_Z$ queue when the system is busy or idle, when the \textit{descend condition} is satisfied or not, and when there is a conflict at node $\widetilde{Z}$ or not can be summarized in the following six scenarios. Note that $a_k,\ c_k$ and ${}^\gamma v_{j,k}$ are independent at timeslot $k$. For $j \in \{1,...,L-1\}$,

\begin{enumerate}[nosep]
\item If there is a conflict ($a_k = c_k = 1$) due to simultaneous nodal arrivals $a_k$ and $c_k$ and,
\begin{enumerate}
\item the \textit{descend condition} is satisfied ($b_k = 1$) with probability $\Omega(1-B(0))$, then by Rule \ref{rule: 6}, the nodal arrival $c_k$ enters service and the nodal arrival $a_k$ descends to $Z_{\mathrm{ID}}$ zone. No arrival enters the $\gamma_Z$ queue. The UAS present in the $\gamma_Z$ queue would queue in the given timeslot.
\begin{alignat}{2}
    {}^\gamma v_{1,k+1} = 0\ \ ,\ \ {}^\gamma v_{j+1,k+1} = {}^\gamma v_{j,k} \ \ \ \ \implies \ \ \ \ {}^\gamma V_1(z) = 1 \ \ , \ \ {}^\gamma V_{j+1}(z) ={}^\gamma V_j(z) \label{eqn: beta queuing 1}
\end{alignat}
\item the \textit{descend condition} is not satisfied ($b_k = 0$) with probability $\Omega B(0)$, then the arrival $c_k$ would enter service and arrival $a_k$ would enter $\gamma_Z$ queue. 
\begin{alignat}{2}
{}^\gamma v_{1,k} = 1\ \ , \ \ {}^\gamma v_{j+1,k+1} = {}^\gamma v_{j,k} + a_k\ \ \ \ \implies \ \ \ \ {}^\gamma V_1(z) = A(z)  \ \  , \ \ {}^\gamma V_{j+1}(z) = A(z)\ {}^\gamma V_{j}(z)    
\end{alignat}
\end{enumerate}

\item If there is no conflict ($a_k = c_k = 0$ or $a_k \neq c_k$), that is, both $a_k=1$ and $c_k=1$ do not occur simultaneously and,
\begin{enumerate}
    \item if the system is busy ($\theta_k = 0$) or if the system is idle ($\theta_k = 1$) but there are UAS with higher preference present ($a_{\mathrm{I},k} = 1$ or ${}^\beta v_{L-1,k} > 0$) with probability $\rho$ then the UAS in the $\gamma_Z$ queue do not enter service and, 
    \begin{align}
    \ \ \rho\ =\  \mathlarger{[}\ \Theta_0 + (1-\Theta_0)\mathlarger{(}(1 - A_{\mathrm{I}}(0)) + A_{\mathrm{I}}(0)(1 - {}^\beta V_{L-1}(0)) \mathlarger{)}\ \mathlarger{]}\ =\  1 - (1-\Theta_0)\ A(0)\ {}^\beta V_{L-1}(0)
    \end{align} 
    \begin{enumerate}
        \item if the descend condition is satisfied ($b_k = 1$), then nodal arrival $a_k$ descends to $Z_{\mathrm{ID}}$ zone and the nodal arrival $c_k$ enters the $\gamma_Z$ queue with probability $\rho(1-B(0))(1-\Omega)$.
        \begin{alignat}{2}
        {}^\gamma v_{1,k} = c_k \ , \ {}^\gamma v_{j+1,k+1} = {}^\gamma v_{j,k} + c_k \ \ \implies \ \ {}^\gamma V_1(z) = C(z) \ ,\ {}^\gamma V_{j+1}(z) =\ {}^\gamma V_{j}(z) C(z)
        \end{alignat}
        \item if the descend condition is not satisfied ($b_k = 0$), then nodal arrivals $a_k$ and $c_k$ enter the $\gamma_Z$ queue with probability $\rho B(0) (1-\Omega)$.
        \begin{alignat}{2}
        &{}^\gamma v_{1,k} = a_k + c_k \ &&, \ \ \ \ {}^\gamma v_{j+1,k+1} =\ {}^\gamma v_{j,k} + a_k + c_k \notag \\
        \implies \ \ &{}^\gamma V_1(z) = A(z)C(z) \ &&,\ \ \ \ {}^\gamma V_{j+1}(z) = {}^\gamma V_{j}(z) A(z) C(z)
        \end{alignat}
    \end{enumerate}

    \item if the system is idle ($\theta_k = 1$) and no UAS with higher preference is present ($a_{\mathrm{I},k} = 0$ and ${}^\beta v_{L-1,k} < 0$) with probability $(1-\rho)$ and,
    \begin{enumerate}
        \item if the \textit{descend condition}  is satisfied ($b_k = 1$) with probability $(1-\rho) (1-B(0))(1-\Omega)$, then nodal arrival $a_k$ descends to $Z_{\mathrm{ID}}$. If nodal arrival $c_k = 1$, then following LCFS discipline, the nodal arrival enters service. Otherwise, the UAS in the $\gamma$ queue that has arrived last enters service, thus reducing the number of UAS queuing in the $\gamma$ queue by one.
    \begin{align}
    & {}^\gamma v_{1,k} = (c_k-1)^+ = 0 \ \  ,  &&{}^\gamma v_{j+1,k+1} =\ {}^\gamma (v_{j,k} + c_k -1)^+ \notag \\
    \implies \ & {}^\gamma V_1(z) = 1 \ \ ,   && {}^\gamma V_{j+1}(z) =\ {}^\gamma V_j(0) + \left(\cfrac{{}^\gamma V_j(z) - {}^\gamma V_j(0)}{z}\right)C(z)
    \end{align}
    \item if the \textit{descend condition} is not satisfied ($b_k = 0$) with probability $(1-\rho) B(0)(1-\Omega)$, then the UAS that has arrived last enters service.   
    \begin{align}
    & {}^\gamma v_{1,k} = (a_k + c_k-1)^+ = 0 \ , &&  {}^\gamma v_{j+1,k+1} = {}^\gamma (v_{j,k} + a_k + c_k -1)^+ \notag \\
    \implies \ & {}^\gamma V_1(z) = 1 \ , &&  {}^\gamma V_{j+1}(z) =\ {}^\gamma V_j(0) + \left(\cfrac{{}^\gamma V_j(z) - {}^\gamma V_j(0)}{z}\right)A(z)C(z) \label{eqn: beta queuing 10}
    \end{align}
    \end{enumerate}
\end{enumerate}
\end{enumerate}
The PGF ${}^\gamma V_j(z)$ can be obtained by summing the above conditional PGFs in Eqns. \ref{eqn: beta queuing 1}-\ref{eqn: beta queuing 10} using the total probability theorem. By recursively substituting ${}^\gamma V_j(z)$ in ${}^\gamma V_{j+1}(z)$, we get ${}^\gamma V_{L}(z)$. In the timeslot $k$, a UAS will enter service if there is a conflict or if the system is idle and there are UAS available to enter service. In timeslot $k$, a UAS will be available to enter service if,
\begin{enumerate}
    \item there is a non-nodal arrival ($a_{\mathrm{I},k} = 1$) or,
    \item there are UAS in $\beta_Z$ queue (${}^\beta v_{L-1,k} > 0$) or,
    \item there is a nodal arrival ($c_k = 1$ or $a_k = 1$ with $b_k = 0$) or,
    \item there are UAS in $\gamma_Z$ queue (${}^\gamma v_{L,k} > 0$) or,
    \item there is a non-nodal arrival ($a_{\mathrm{O},k} = 1$) that descends towards $Z_{\mathrm{UP}}$ when $\textit{descend condition}$ is satisfied.
\end{enumerate}
The probability $\pi$ that a UAS is available to enter service in the given timeslot is given by,
\begin{align}
    &\sigma_Z = A_{\mathrm{I}}(0)\ {}^\beta V_{L-1}(0) \ C(0)\ [ 1-B(0)(1- A(0)) ]\ {}^\gamma V_L(0)\ \ \ \ \ , \ \ \ \ \pi = (1-\sigma_Z) + \sigma_Z (1-A_{\mathrm{O}}(0)) \label{eqn: UAS available X}
\end{align}
where $\sigma_Z$ is the descend probability that a non-nodal arrival at node $\widetilde{Z}_\mathrm{ON}$ would descend into $Z_{\mathrm{UP}}$ zone and $1 -B(0) = \sigma_{\mathrm{I}}$ is the descend probability that the nodal arrival at node $\widetilde{Z}$ will descend into $Z_{\mathrm{ID}}$ zone. The probability that a UAS will enter the service in a given timeslot is given by,
\begin{align}
    &P[w_{1,k} = 1] = \Omega + (1-\Omega) (1-\Theta_0) \pi\ , \ \ \ \ \ \ W_1(0) = 1- P[w_{1,k} = 1] \label{eqn: Stream X feedback}
\end{align}

The Eqn. \ref{eqn: Stream X feedback} is the feedback relation that captures congestion responsible for UAS queuing in a Stream$\{X\}$ queue. Substituting Eqn. \ref{eqn: conflict arrival 2} in Eqn. \ref{eqn: Stream X feedback} and simplifying for $W_1(0)$, we get $W_1(0)$ as a function of $\Theta_0$. Substituting Eqn. \ref{eqn: Stream X feedback} in Eqn. \ref{eqn: forward}, we get a polynomial equation in $\Theta_0$. Numerically solving the polynomial equation for zeros, we get the congestion probability $\Theta_0$. The existence proof for the $\Theta_0$ in $[0,1]$ is similar to the proof of Theorem \ref{thm: Existence}. Substituting
$\Theta_0$ in $\pi$, $\sigma_Z$, ${}^\beta V_j(z),\ j \in \{1,...,L-1\}$ and ${}^\gamma V_j(z),\ j \in \{1,...,L\}$, we get the respective probabilities. As $\theta_k$ is an MMRP and $w_{1,k}$ is an MMBP, the respective Markov modulated probabilities $\Theta_0^*$ and $W_1^*(0)$ are given by Eqns. \ref{eqn: MMRP}- \ref{eqn: MMBP}. The average number of UAS present in the service and queue of Stream$\{X\}$ queuing system can be obtained by substituting $\Theta_0^*$ and $W_1^*(0)$ in $U(z)$, ${}^\beta V_{L-1,k}$ and ${}^\gamma V_{L,k}$ and computing $U^{(1)}(1)$, ${}^\beta V^{(1)}_{L-1}(1)$ and ${}^\gamma V^{(1)}_L(1)$, respectively.

The UAS that have queued for $L + 1$
slots in the $\gamma_Z$ queue are lost from the $\overline{Z}$ queuing system and become non-nodal arrivals $a_{\mathrm{I},k}$ for the $\overline{Z}_{\mathrm{ON}}$ queuing system. If the $Z_{\mathrm{UP}}$ is uncongested in timeslot $k$, but there are UAS with higher preference present that would be entering service, then the UAS corresponding to ${}^\gamma q_k = 1$ is lost from the system. When $Z_{\mathrm{UP}}$ becomes congested, under steady-state, the fraction of non-nodal arrivals $a_{\mathrm{I},k}$ and nodal arrivals $a_k$, that end up in the $\gamma_Z$ queue, but fail to enter service when $Z_{\mathrm{UP}}$ was uncongested are lost during the congestion period. The steady-state queue overflow probability $\Phi$ for the $\overline{Z}$ queuing system is given by,
\begin{align}
& {}^\gamma v_{L,k} = {}^\gamma v_{L-1,k} + {}^\gamma q_k \ \ \ \ ,\ \ \ \ \mathrm{P}[{}^\gamma q_k = 1] = \left(\cfrac{{}^\gamma V_{L-1}(0) - {}^\gamma V_L(0)}{{}^\gamma V_{L-1}(0)}\right) \\
& {}^\gamma \Phi = (1-\Theta_0^*)\left(1-W_1^*(0)\right) \mathrm{P}[{}^\gamma q_k = 1] + \Theta_0^*\ ( 1-A_{\mathrm{I}}(0) + (1-A(0))B(0) ) \left(1-{}^{\gamma}V_{L}(0)\right) W_1^*(0) + \xi(1-C(0))
\end{align}

The PGF for the steady-state probability distribution of the non-nodal arrivals $a_{\mathrm{I},k}$ for the neighboring $\overline{Z}_{\mathrm{ON}}$
queuing system is given by 
\begin{align}
    A_{\mathrm{I}}(z) = (1-{}^\gamma \Phi) + z\ {}^\gamma \Phi \label{eqn: non-nodal arrival}
\end{align}

\subsection{Service departures}
Knowing congestion probability $\Theta_0^*$, the PGF of the probability distribution that UAS would depart the queuing system $\overline{Z}$ after receiving service is given by Eqn. \ref{eqn: 0 departures}. The UAS departures from the system $\overline{Z}$ are $a_{k}$ nodal arrivals for its neighboring upstream system $\overline{Z}_{\mathrm{UP}}$ given by
\begin{align}
    A(z) &= W_{S-1}(z) = W_1^*(0) + z(1-W_1^*(0)) \label{eqn: 0 departures}
\end{align}

The algorithmic procedure for estimating the expected UAS traffic spread in the finite grid $Z = (X, Y)$, where $X = \{-X_e,...,0,... X_e\}, \ Y = \{1,..., Y_e\}$, for a given UAS arrival rate $\lambda$ at the source is given in Algorithm \ref{alg: expected spread}.

\begin{algorithm}
    \SetAlgoLined   
    \raggedright For zone $Z = (0,1)$ that encloses the source node, initialize $A(z) = (1-\lambda)+\lambda z, \ A_{\mathrm{O}}(z) = 1$ ; \\
    \raggedright For remaining Level$(1)$ zones, that is $Z = (X,1), \ X \neq 0$, initialize $A(z) = A_{\mathrm{O}}(z) = 1$ ; \\ 
    \raggedright For zones intersected by exogenous traffic stream, initialize $E_{0}(z) = (1-\lambda_e)+\lambda_e z$ ; \\
    \raggedright For all zones, initialize the design parameters $L,\ M$ and $\eta$. Start with zone $Z = (X,Y) = (0,1)$ ;\\
    \While{$Y\in \{1,...,Y_e\}$ }{
            \While{$X\in \{0,...X_e\}$}{

            \uIf{for zone $Z = (X,Y)$, the $Z_{\mathrm{UP}}$ is a no-fly zone}{Assign $\Theta_0^* = 1$ for queuing system $\overline{Z} = (X,Y)$ ; }
            \uElseIf{for zone $Z = (X,Y)$, the $Z_{\mathrm{IN}}$ or $Z_{\mathrm{ID}}$ is a no-fly zone}{Assign $\Theta_0^* = 0$ for queuing system $\overline{Z} = (X,Y)$ ; }
            \Else{Evaluate $\Theta_0^*$ for queuing system $\overline{Z} = (X,Y)$ using Eqn. \ref{eqn: forward} and Eqn. \ref{eqn: MMRP} ; }
            
                 Compute 
                1) the average number of UAS in service $U^{(1)}(1)$ for zone $Z =(X,Y)$ using Eqn. \ref{eqn: UAS in service}; \\
                \ \ \ \ \ \ \ \ \ \ \ \ \ \ \ \ 2) the descend probability $\sigma_Z$ for the zone $Z = (X,Y)$ using Eqns. \ref{eqn: UAS available 0}, \ref{eqn: UAS available X}; \\
                \ \ \ \ \ \ \ \ \ \ \ \ \ \ \ \ 3) the nodal arrivals $A(z)$ for zone $Z = (X,Y+1)$ and $(-X,Y+1)$ using Eqn. \ref{eqn: 0 departures}; \\
                \ \ \ \ \ \ \ \ \ \ \ \ \ \ \ \ 4) the non-nodal arrivals $A_{\mathrm{I}}(z)$ for zone $Z = (X+1,Y)$ and $(-X-1,Y)$ using Eqns. \ref{eqn: left non-nodal arrival}, \ref{eqn: right non-nodal arrival}, \ref{eqn: non-nodal arrival}; \\
                Set $X = X+1$  
            }  
            Expected spread in level $Y$ = Stream$\left[ X_{min} \leq X \leq X_{max} \right]$ segments, where for all $X \in \{-X_e,...,0,...X_e\}$, \\
            \ \ \ \ \ \ \ \ $X_{min} = \min \left\{ X\ \vert \ U^{(1)}(1) \geq 1 \text{\ for\ } Z = (X,Y)  \right\}$ , and $X_{max} = \max \left\{ X\ \vert \ U^{(1)}(1) \geq 1 \text{\ for\ } Z = (X,Y) \right\}$; \\
            
            Set $Y = Y+1$;\\
            \While{$X\in \{0,...X_e\}$}
            {Arrival $A_{\mathrm{O}}(z)$ for zone ($X,Y$) = $A(z)$ of zone ($X+1,Y$) \\
            Arrival $A_{\mathrm{O}}(z)$ for zone ($-X,Y$) = $A(z)$ of zone ($-X-1,Y$)\\
            Set $X = X+1$} 
    }
    \caption{Estimating expected UAS traffic stream spread (the set of active parallel stream segments)}
    \label{alg: expected spread}
\end{algorithm}
 
\section{Simulation Results and Analysis} \label{sec: simulation results}

Consider the workspace enclosing a source-destination pair at altitude $80\ \mathrm{m}$ in the urban environment. The source is at $(0,0,80)\ \mathrm{m}$ where, for every request received, a UAS is deployed to deliver the payload at the paired destination. The time axis is discretized into $\Delta T = 2.5\ \mathrm{s}$ timeslots. We assume one or no request is received in any given timeslot. The average rate $\lambda$ at which the requests are received, analogously the arrival rate $\lambda$ at which the UAS would be deployed at source, varies in the range $[0.05,1]$. For simulation, a horizontal cross-section of the workspace measuring $[-400,\ 400]\ \mathrm{m}$ along $x$ axis, and $[-50,\ 900]\ \mathrm{m}$ along $y$ axis and intersecting a static obstacle at the given altitude is considered. The cross-section is tessellated into cells of dimension $5\ \mathrm{m} \times 5\ \mathrm{m}$. The UAS transitions with $1$ cell per timeslot nominal velocity. The UAS is enclosed in a virtual periphery of safety radius $r = 1\ \mathrm{m}$. The cell dimensions are large enough that there is no breach in the UAS periphery amongst UAS transitioning in respective cell neighborhoods (shown in Fig. \ref{fig: 15}). The UAS has a $L = 5$ slot look-ahead time window, allowing it to predict where other UAS would be in the next $L\Delta T = 12.5\ \mathrm{s}$. Assume that the UTM has prescribed a nominal path between the source-destination pair such that in the cross-section, the set of cells through which the $y$ axis passes is the nominal segment. The cell enclosing the source is the start cell for the nominal segment. The zones are defined with edge length $S = 2L+1 = 11$ cells, and the finite grid (comprising $10$ streams and $10$ levels) is constructed centered on the nominal segment. 
\begin{figure}[h!]
 \centering
    \includegraphics[width = 0.4\linewidth]{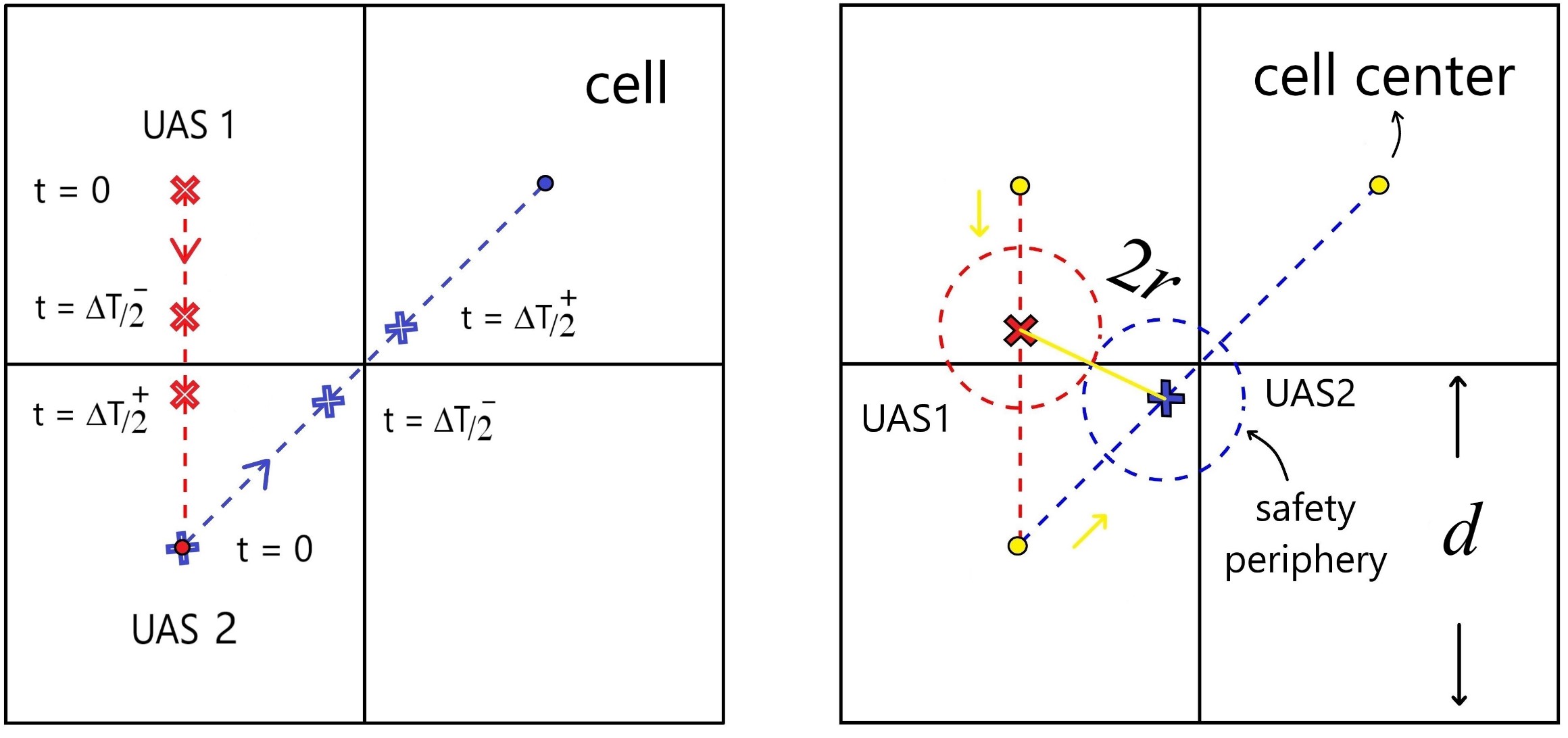}
    \caption{The edge length $d$ of the cell is so chosen that the UAS closest approach does not breach the safety periphery $r$. To ensure this $\mathrm{d} \geq 2\sqrt{5}\ r$ should be satisfied.}
    \label{fig: 15}
\end{figure}

Every UAS deployed at the source follows the proposed congestion mitigation strategy. The simulation is run until around $1200$ UAS are deployed at the source. The congestion probability and expected number of UAS present in the zone are numerically computed by solving Eqn. \ref{eqn: forward} using MATLAB symbolic math toolbox. The theoretical results are validated with the average values of the $1200$ UAS simulation data. 

Figure \ref{fig: 16} shows snapshots of UAS employing the proposed rule-based congestion mitigation strategy in a distributed manner. When every UAS in the source-destination traffic stream follows the proposed strategy, Fig. \ref{fig: 17}-\ref{fig: 19} show that the emerging traffic pattern closely resembles a parallel air corridor network with the number of air corridors adapting to the arrival rate \footnotemark[1]. The emergence of parallel air corridors is a consequence of the proposed strategy, and UTM has no role in it, nor is there any scheduling policy in place. It may be observed that in any zone no more than $M$ UAS are present, thus assuring increased UAS intrinsic safety. The above figures also demonstrate the variation in the number of active corridors (also referred to as traffic stream spread) as a function of the arrival rate ($\lambda$) and the UTM design parameters ($M,\ \eta$). 
\begin{figure}[h!]
    \centering
    \includegraphics[width = \linewidth]{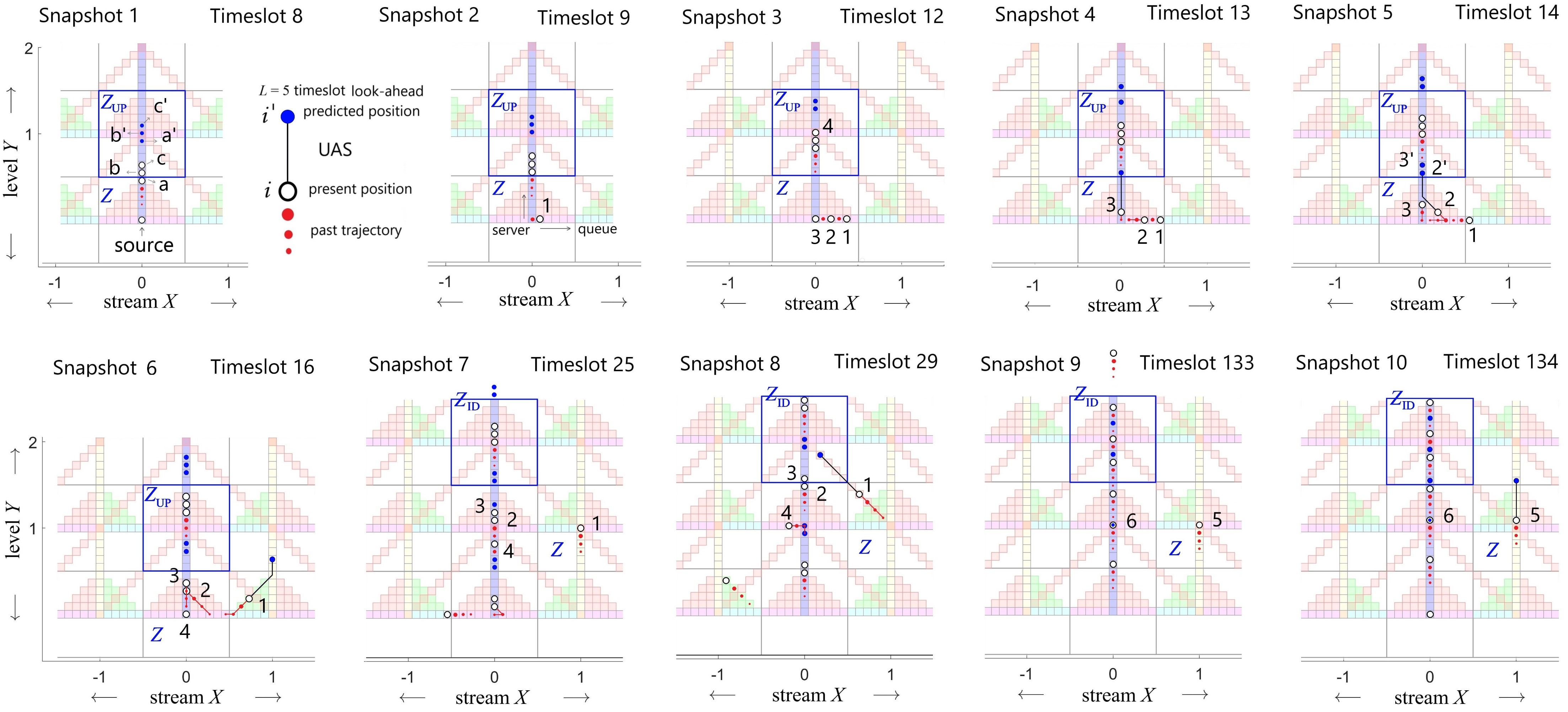}
    \captionof{figure}{The arrival rate $\lambda$ at source is $0.2$. The UTM parameters are $M = 3$ and $\eta = 0.5$.  In snapshot 1, three $L= 5$ slot look-ahead positions (that of UAS a,b,c) are present in $Z_{\mathrm{UP}}$. As $M = 3$, from the perspective of the UAS present at the source node, the zone $Z_{\mathrm{UP}}$ is congested (Rule \ref{rule: 1}). In snapshot 2, when UAS 1 arrived, the $Z_{\mathrm{UP}}$ was congested; hence, UAS 1 entered the queue (Rule \ref{rule: 2}). In snapshot 3, the $L$-slot look-ahead position of UAS in the $\beta \rightarrow \gamma$ path (that of UAS 4) cannot be determined; only two look-ahead positions are in $Z_{\mathrm{UP}}$. In snapshots 3-5, as $Z_\mathrm{UP}$ is uncongested, the UAS 2, 3 enter service following the LCFS discipline (Rule \ref{rule: 3}). As there is a queuing deadline of $L = 5$ slots, UAS 1 is lost from the $\overline{Z}$ queuing system. In snapshot 6, the past trajectory of UAS 1,2,3 shows the path opted corresponding to $\psi > 0$, $\psi < 0$ and $\psi = 0$, respectively (Rule \ref{rule: 4}). In snapshot 7-8, from the perspective of UAS 1, as $Z_{\mathrm{ID}}$ is uncongested and no UAS are present in the $\beta_Z \rightarrow \gamma_Z$ path, the descend condition is satisfied. Hence, UAS 1 descends to $Z_{\mathrm{ID}}$ (Rule \ref{rule: 5}). Whereas, in snapshots 9-10, though $Z_{\mathrm{ID}}$ is uncongested, the descend condition is not satisfied because of UAS 6. Hence, UAS 5 does not descend. } 
    \label{fig: 16}
\end{figure}
\begin{figure}[h!]
    \centering
    \includegraphics[width = \linewidth]{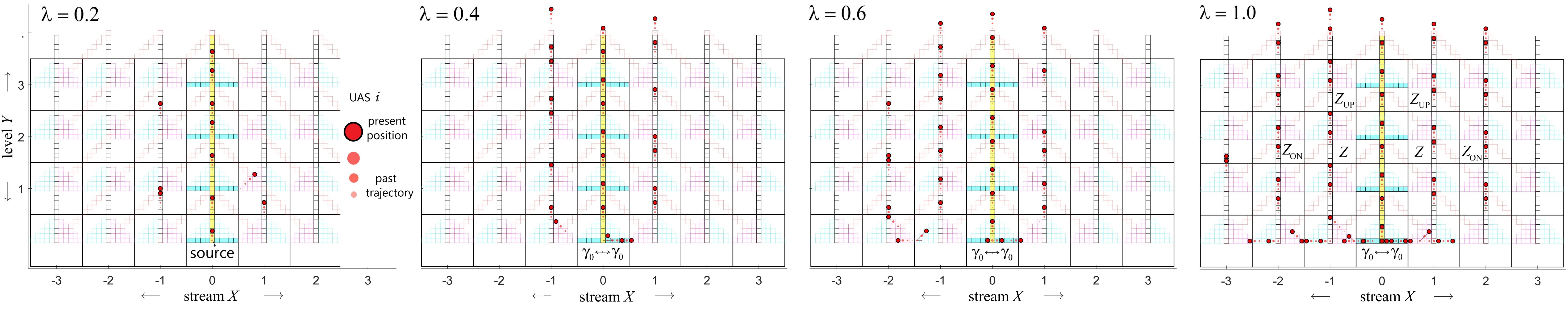}
    \caption{Variation of UAS traffic stream spread as a function of the arrival rate $\lambda \in [\ 0.2,0.4,0.6,1.0\ ]$, when $M = 2$ and $\eta = 0.5$. As the congestion mitigation strategy is being employed, at most $2$ UAS are present in any zone. A UAS queuing in the Stream$\{0\}$ system, would be queuing in either of the two $\gamma_0$ queues with probability $\eta = 0.5$ (Rule \ref{rule: 8}). }
    \label{fig: 17}
\end{figure}
\footnotetext[1]{The simulation videos have been uploaded in the following link \url{https://youtu.be/Yk1vY6nynHg}}

\begin{figure}[h!]
    \centering
    \includegraphics[width = \linewidth]{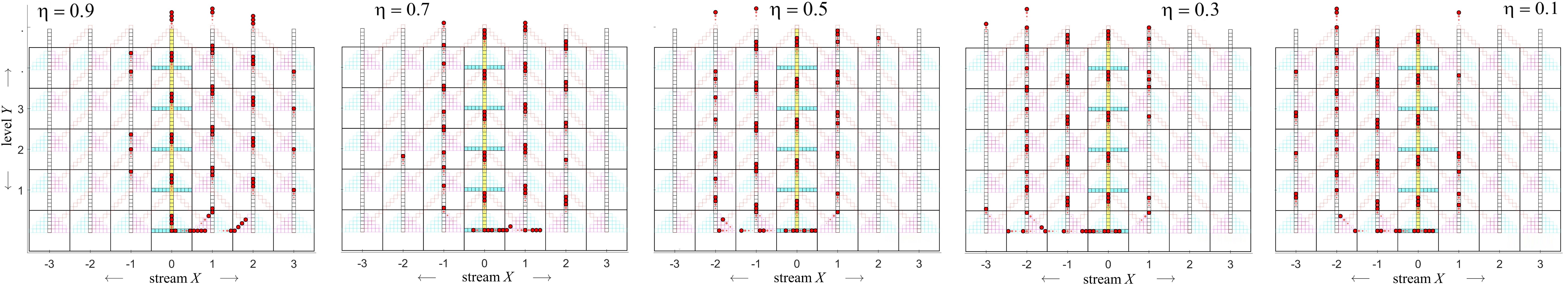}
    \caption{Variation of UAS traffic stream spread as a function of bias $\eta \in [0.1,0.3,0.5,0.7,0.9]$ in traffic distribution about the nominal segment when $M = 3$ and $\lambda = 1.0$. At most $3$ UAS are present in any zone.  }
    \label{fig: 18}
\end{figure}
\begin{figure}[h!]
    \centering
    \includegraphics[width = \linewidth]{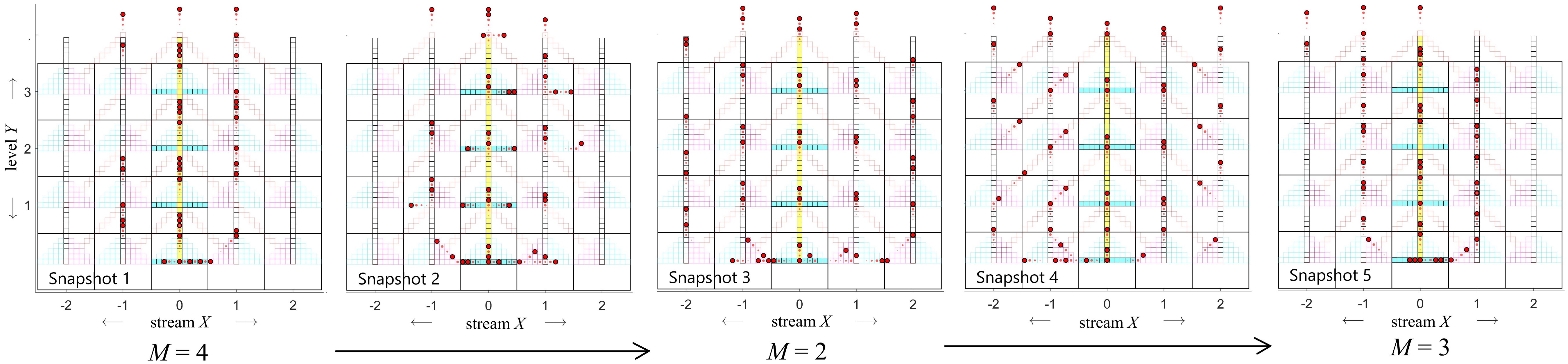}
    \caption{Variation of UAS traffic stream spread as a function of the minimum number of UAS to be present in a zone for congestion $M \in [2,3,4]$ when $\eta = 0.5$ and $\lambda = 0.8$. At most $M$ number of UAS are present in any zone. Snapshots 2 and 4 show the transient behavior of UAS moving outward and descending inward, respectively, adapting to the changing congestion parameter $M$.}
    \label{fig: 19}
\end{figure}

As the arrival rate $\lambda$ increases or $M$ decreases, the congestion frequency and congestion duration in the \textit{inward} stream zones increases. Hence, more UAS shift onto \textit{outward} streams, resulting in increased traffic spread. When $\eta < 0.5$, we have denser traffic flow on the left of the nominal segment and vice versa.
 In the lower levels, the UAS inter-separation distance depends only on the arrival rate. 
When the arrival rate is small, the UAS inter-separation distances are large. Hence, the congestion duration is small (refer Fig. \ref{fig: 20}a). However, in the higher levels, when a zone in the \textit{inward} stream becomes uncongested, the UAS in \textit{outward} streams descend onto the \textit{inward} stream (UAS 1 in snapshot 8 of Fig. \ref{fig: 16}). Thus reducing UAS inter-separation distances and resulting in longer congestion durations (refer Fig. \ref{fig: 20}b). In the $Z = (0,1)$ zone, when $\lambda = 1$, every timeslot the zone is uncongested, a UAS is available to enter the zone (refer Fig. \ref{fig: 20}c). Thus, the zone remains congested for $S-M$ timeslots. For example, when $S = 11,\ M = 2$ and $\lambda = 1$, the zone is congested for $S-M = 9$ timeslots and hence the congestion probability $\Theta_0 = (S-M)/S = 0.82$ (as validated in the $\Theta_0$ plot of Fig. \ref{fig: 21}, when  $\lambda = 1$). The rearrangement of UAS in the traffic stream happens such that beyond a certain level $Y$, we observe no UAS moving \textit{outward} and no UAS descending \textit{inward}. 
\begin{figure}[h!]
    \centering
    \includegraphics[width = 0.8\linewidth]{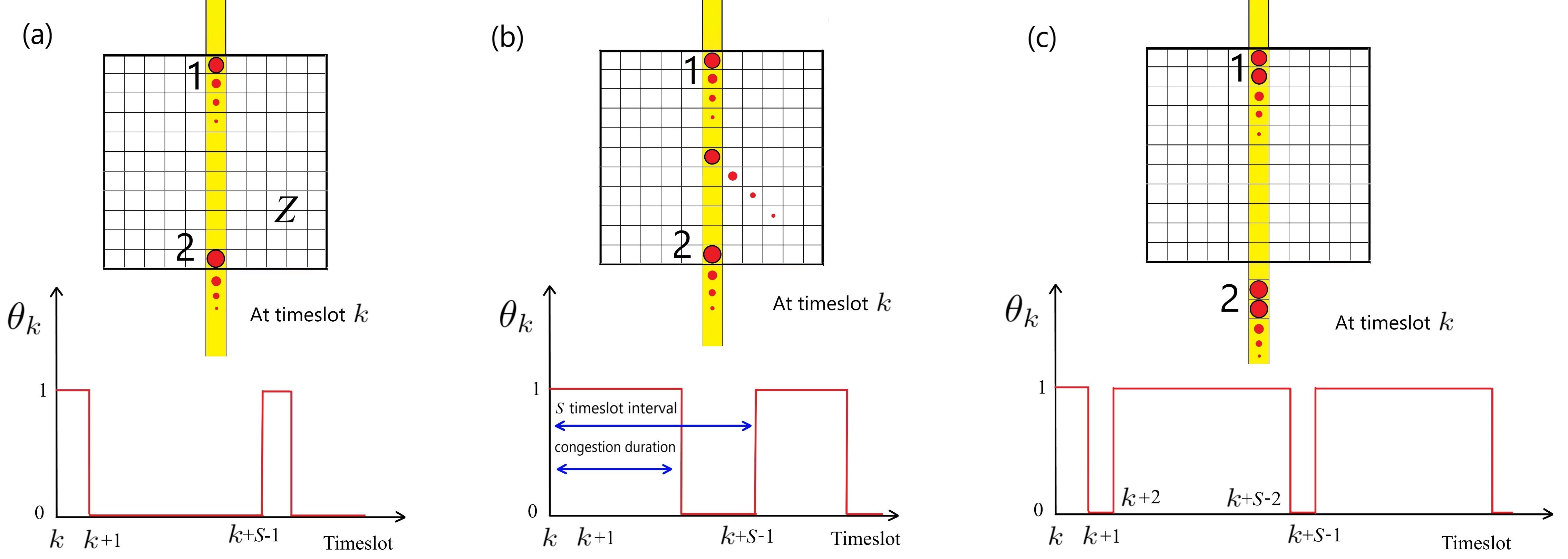}
    \caption{The congestion duration when $M = 2$ and a) $\lambda$ is small, b) $\lambda$ is small, but there are UAS in the \textit{outward} stream zone that would descend, c) $\lambda = 1$, a UAS is available to enter every timeslot the zone is uncongested. }
    \label{fig: 20}
\end{figure}

 The same is validated in Fig. \ref{fig: 21}. The average number of UAS in the queue (that is, UAS moving \textit{outward}) and the queue overflow is nearly zero for level $2$ and above zones. In the congestion probability plot, when $\lambda < 0.2$, there is no congestion, so no UAS shifts onto \textit{outward} streams, and hence, no UAS is available to descend. The inter-separation distances depend only on arrival. Thus, the congestion duration remains unaffected as levels increase. When $\lambda \in (0.2, 0.5)$, there is congestion, but the duration is small. There are UAS that descend and increase the congestion duration in the higher levels. However, on average, the increase is not significant as with the increase in the congestion duration, the probability that a UAS would descend decreases. When $\lambda > 0.6$, the congestion durations are long enough that most of the UAS arriving in the lower levels shift onto \textit{outward} streams. However, these UAS rarely descend (refer Fig. \ref{fig: 21}c). Hence, again, the congestion duration remains unaffected as the level increases. Thus, as $\lambda\rightarrow 0$ and as $\lambda\rightarrow 1$, the congestion probabilities of all $(0,Y)$ zones converge.
 \begin{figure}[h!]
    \centering
     \includegraphics[width = 0.95\linewidth]{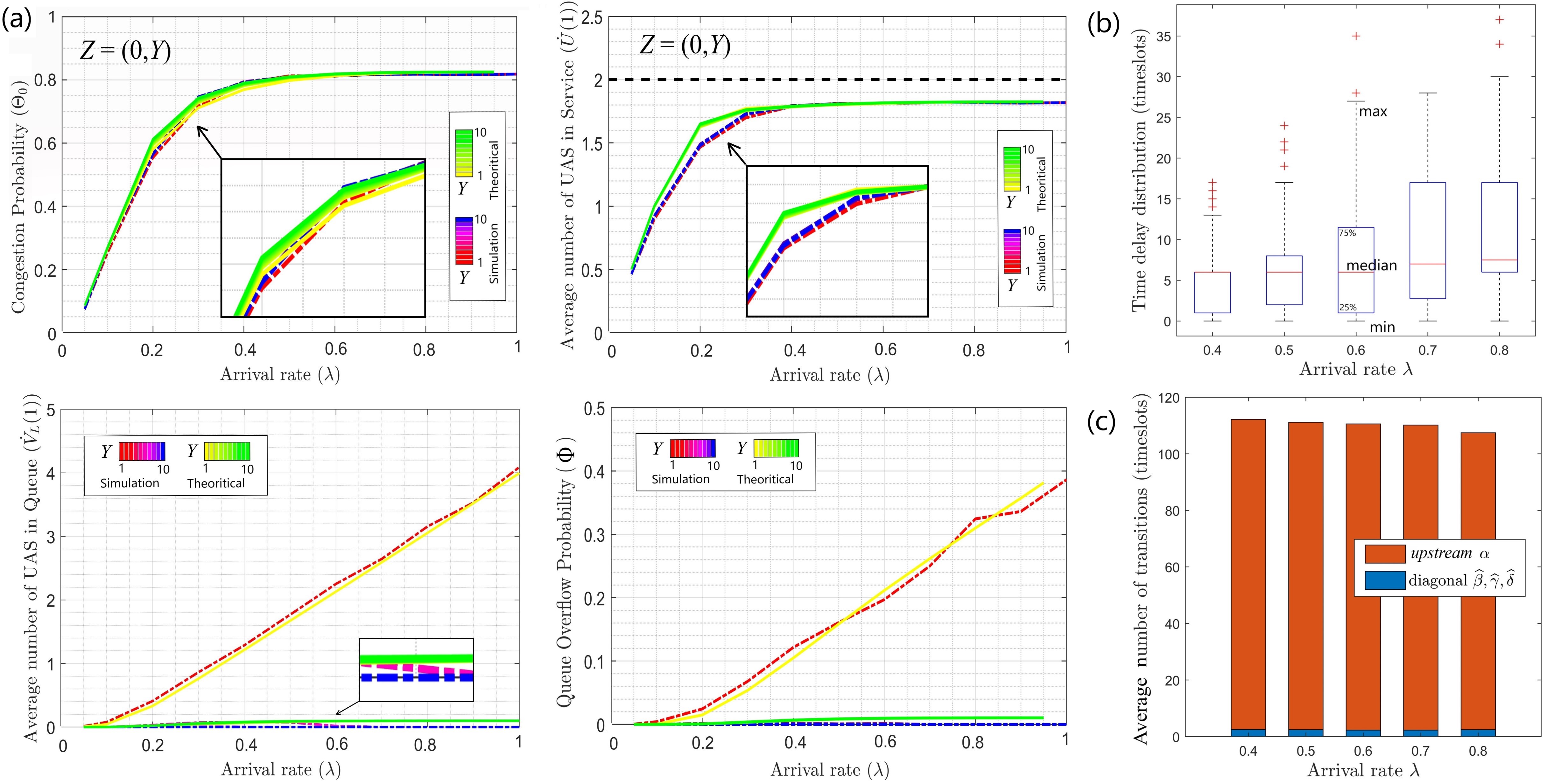}
    \caption{a) Theoretical and simulation comparison of congestion in $Z = (0,Y)$, $Y \in \{1,...,10\}$ for varying arrival rate $\lambda \in [0.05, 1]$ and minimum number of UAS in zone for congestion $M = 2$. Here, $S = 2L+1 = 11$ cells, and $\eta = 0.5$, b) The time delay introduced in the UAS path as a consequence of the UAS executing lateral transitions in the event of congestion and the number of \textit{upstream} and, c) diagonal transitions executed by the UAS to reach a Level $10$ zone (that is a total of $10\times S = 110$ transitions).}
    \label{fig: 21}
\end{figure}

When executing a diagonal transition, to travel 1 unit distance in the \textit{upstream} direction, the UAS needs $\sqrt{2}$ times higher velocity than when executing an \textit{upstream} transition. Thus, diagonal transitions are energy-intensive maneuvers. However, the number of diagonal transitions executed is low (refer Fig. \ref{fig: 21}c) and mostly confined to the lower levels.

Unlike arrival rate $\lambda$, the UTM can take strategic decisions for parameters $M$ and $\eta$. Figure \ref{fig: 22} shows Rule \ref{rule: 7} being employed by UAS in the presence of a static obstacle, due to which the UAS circumnavigates the obstacle. Instead, the UTM could choose an $\eta$ such that the traffic spread would not pass through the static obstacle (as in Fig. \ref{fig: 18}). In Fig. \ref{fig: 22}, the UAS in zone $1$ always perceives zone $1_{\mathrm{UP}}$ to be congested. These UAS would certainly enter zone marked $2_{\mathrm{UP}}$. The respective $L$-slot look-ahead positions can be determined (indicated with blue dots in zone $2_{\mathrm{UP}}$) and contribute to congestion in Definition $\ref{def: 2}$. Thus, the UAS in zone $2$ are affected by UAS laterally transitioning in the $\beta\rightarrow\gamma$ path of zone $1$. The traffic flow in the urban airspace is weather-dependent. As per UAS safety standards, the allowable congestion in a region is constrained by wind conditions. As in \citet{zhou2020resilient}, the weather is modeled as a first-order Markov chain with zero, moderate, and heavy wind as states. The UTM may associate each state with a $M$ value (shown in Fig. \ref{fig: 23}a). Corresponding to prevailing wind conditions, the UTM communicates respective $M$ values with the UAS. The UAS accordingly adapts, ensuring no more than $M$ UAS are present in any zone, as shown in Fig. \ref{fig: 19}. Another use case for changing $M$ would be at the destination. $M$ must be gradually incremented over the levels (as shown in Fig. \ref{fig: 23}b) so that when $M = S$ in the final level containing the destination node, all UAS would enter the destination. 
\begin{figure}[h!]
    \centering
    \includegraphics[width = 0.7\linewidth]{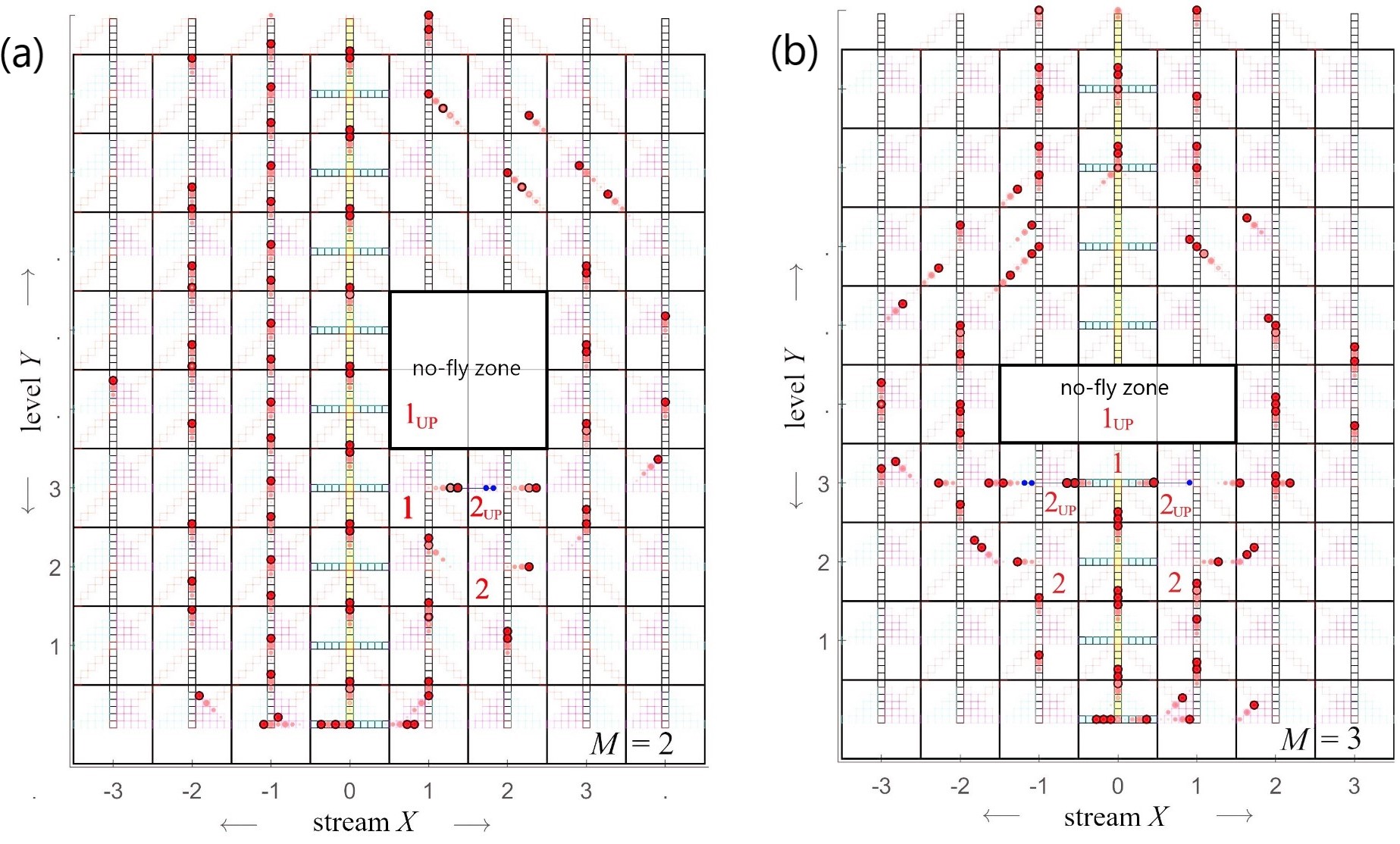}
    \caption{Traffic behavior in the presence of a static obstacle when $\lambda = 0.8$, $\eta = 0.5$ and a) $M = 2$ and b) $M = 3$. }
    \label{fig: 22}
\end{figure} 
\begin{figure}[h!]
    \centering
    \includegraphics[width = 0.6\linewidth]{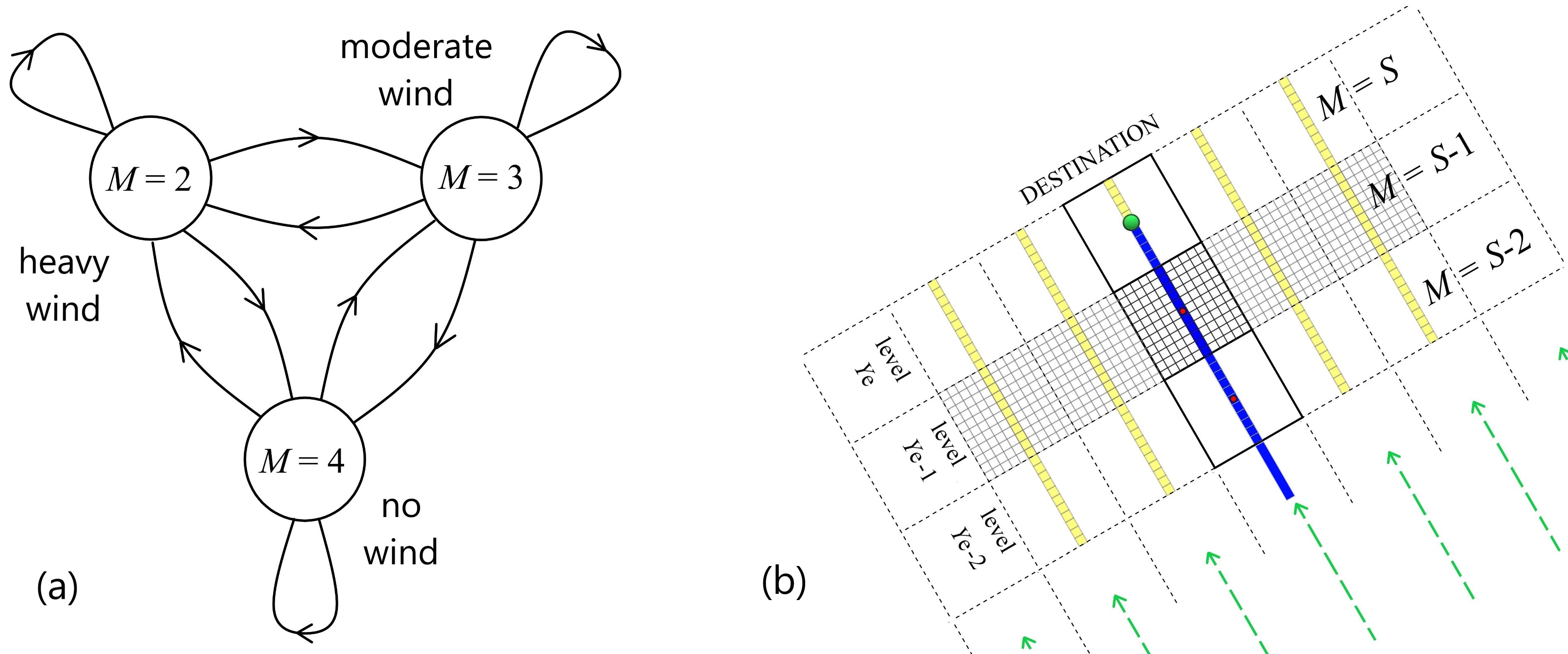}
    \caption{Use cases for controlling UTM parameter $M$, a) weather-dependent congestion regulation where wind conditions are modeled as a Markov chain and b) UAS scheduling at the destination. }
    \label{fig: 23}
\end{figure}

For any given $L,\ M$ and $\eta$, the UTM could estimate the expected traffic spread and the expected number of UAS in the zone beforehand using the queueing models discussed in section \ref{sec: congestion queuing model}. Based on these estimates, the UTM could decide on the look-ahead window $L$ to define the zone and then decide on the parameters $M$ and $\eta$ for the appropriate use case. Figures \ref{fig: 24}-\ref{fig: 25} provide validation for theoretical estimates obtained from the queuing models. 

When $\lambda = 1$, there is an arrival entering the zone $Z = (0,1)$ in every timeslot. As mentioned before, the zone would be congested for $S-M$ timeslots, and hence congestion probability $\Theta_0$ is $(S-M)/S$. The same can be verified in Fig. \ref{fig: 24}a for $\lambda = 1$ at different values of $M$. The average number of UAS in service never exceeds $M$. The UAS queuing would be in either of the two $\gamma_0$ queues, each of capacity $L$, with probability $\eta$ and ($1-\eta$). Hence, the average number of UAS in the queue never exceeds $L$. For zone $Z = (\pm X,1)$, where $X>0$, the UAS would be queuing in both $\gamma$ and $\beta$ queues. The average number of UAS queueing never exceeds $2L-1$. Only a fraction of the UAS arrivals to zone $Z = (0,1)$ are lost from the $\gamma_0$ queue to become non-nodal arrivals for $Z = (\pm1,1)$ zones. For a given arrival rate $\lambda$, the zone $Z = (\pm1,1)$  has less number of UAS available to enter service than zone $Z = (0,1)$. Hence, zone $Z = (\pm1,1)$ has lower congestion probability $\Theta_0$ than  zone $Z = (0,1)$. For the same reason, zones $(X+1,1)$ and $(-X-1,1)$ zones have lower $\Theta_0$ than \textit{inward} neighboring zones $(X,1)$ and $(-X,1)$, respectively. When $\eta = 0.5$, we have equal $\Theta_0$ in $(\pm X,1)$ zones. When $\eta = 0$ all UAS lost from $Z = (0,1)$ become non-nodal arrivals to $Z = (-1,1)$ else only $(1-\eta)$ fraction of UAS lost become non-nodal arrivals. For any given $\lambda$, as more number of UAS are available to enter service when $\eta = 0$, hence in $Z = (-1,1)$ we have more number of UAS in service, resulting in increased congestion, and more number of UAS in the queue when compared to $\eta = 0.5$. \clearpage

\begin{figure}[h!]
    \centering
     \includegraphics[width = \linewidth]{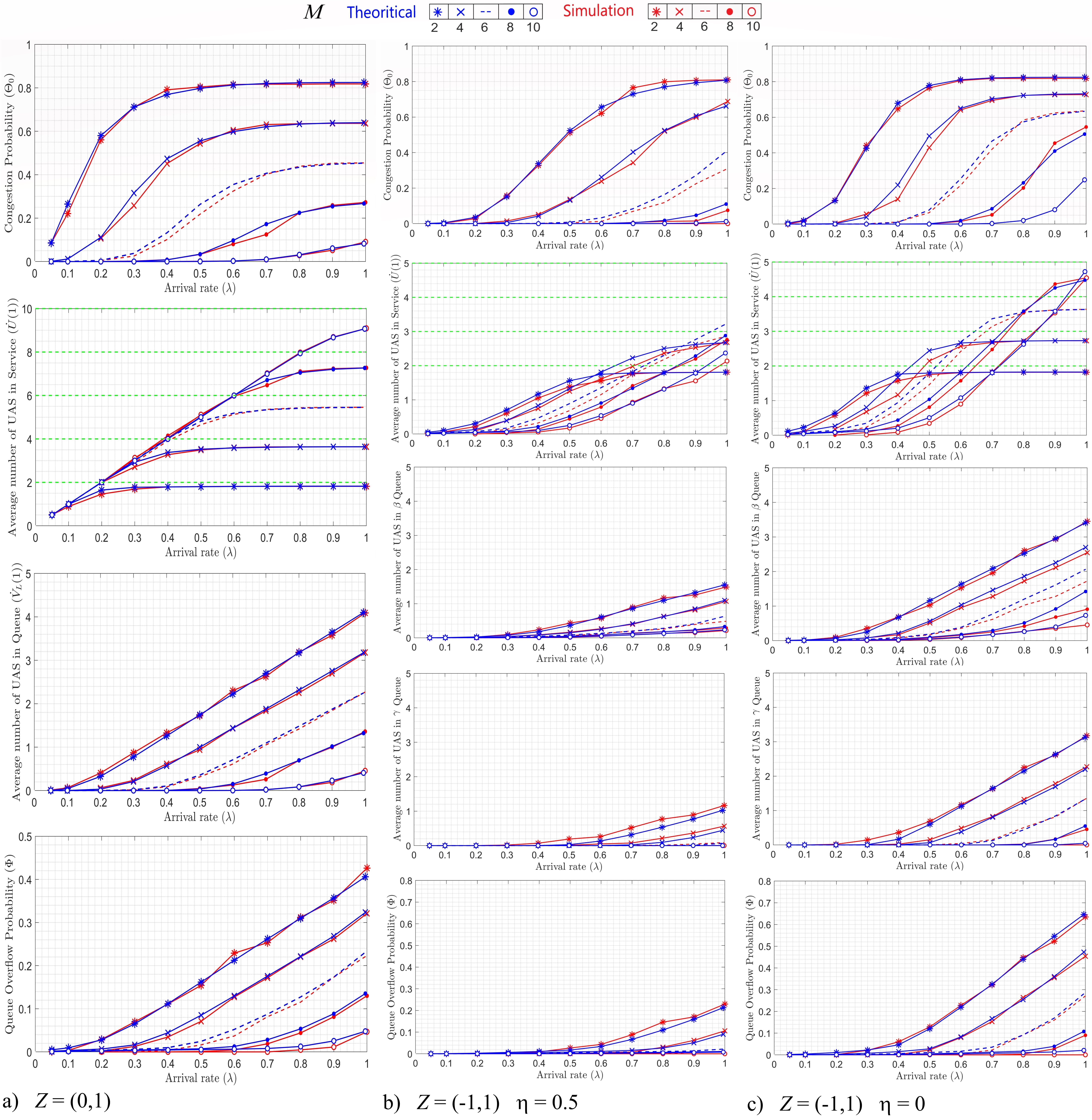}
    \caption{Theoretical and simulation comparison of congestion for $L = 5,\ S = 2L+1 = 11$ and varying arrival rate $\lambda \in [ 0.05, 1 ]$, in a) zone $Z = (0,1)$ when the minimum number of UAS in the zone for congestion $M \in [2,4,6,8,10]$; and in b, c) zone $Z = (-1,1)$ when $M \in [2,3,4,5,6]$ and $\eta = 0.5$ and $\eta = 0$.}
    \label{fig: 24}
\end{figure}

In Fig. \ref{fig: 25}, the expected traffic stream spread at level $1$, for given $M = 2$, and varying arrival $\lambda$ can be determined from the average number of UAS in the service. For example, when $\lambda = 0.8$ in Fig. \ref{fig: 25}a, the $U^{(1)}(1) \geq 1$ for $X \in \{-4,-3,-2,-1,0\}$, hence the expected traffic spread is in Stream$[X_{min}\leq X \leq X_{max}]$ segments, where $X_{min} = -4$ and $X_{max} = 0$. As $\eta = 0.1$, we have denser traffic flow on the left of the nominal segment, hence more spread towards $X<0$ streams. In Fig. \ref{fig: 25}b, $\eta = 0.3$ and $U^{(1)}(1) \geq 1$ for $X \in \{-3,-2,-1,0,1\}$. Hence, the expected spread is in Stream$[-3 \leq X \leq 1]$ segments. In Fig. \ref{fig: 25}c, as $\eta = (1-\eta) = 0.5$, the $U^{(1)}(1) \geq 1$ for $X \in \{-2,-1,0,1,2\}$. The expected spread is in Stream$[-2\leq X\leq 2]$ segments; that is, we have equal spread distribution on the left and the right of the nominal segment.\clearpage

\begin{figure}[h!]
    \centering
     \includegraphics[width = \linewidth]{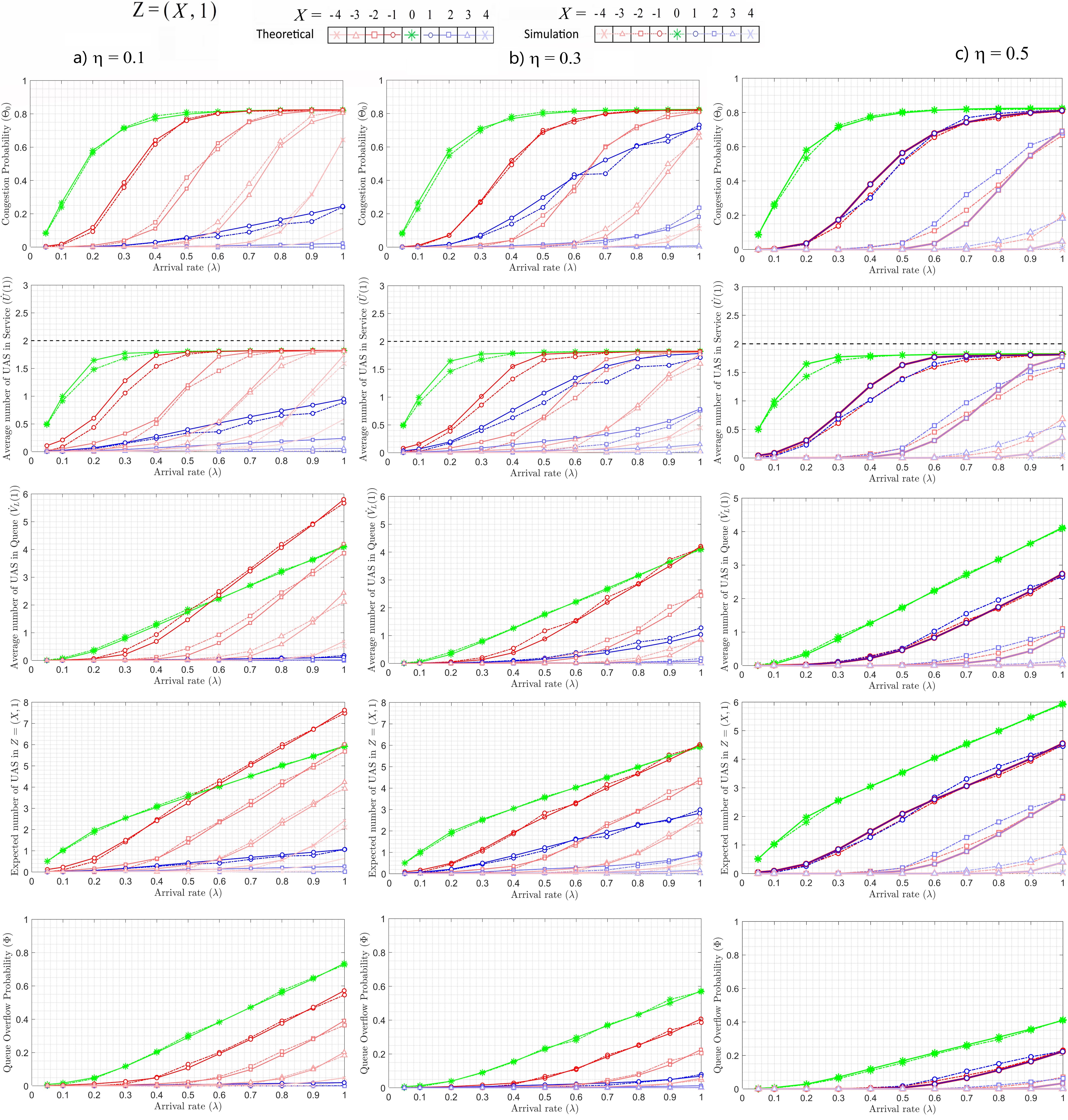}
    \caption{Theoretical and simulation comparison of expected spread in Level$(1)$ zones, that is, $Z = (X,1),\ X \in \{-5,...,0,...,5\}$ when the minimum number of UAS in the zone for congestion is $M = 2$ and the bias in traffic distribution $\eta \in [\ 0.1, 0.3, 0.5\ ]$ (denser traffic flow towards $X < 0$ streams) for varying arrival rate $\lambda \in [\ 0.05, 1 ]$. }
    \label{fig: 25}
\end{figure}

\subsection{Presence of exogenous traffic stream}
Following Assumptions \ref{assumption: 1}-\ref{assumption: 3}, an exogenous traffic stream orthogonally intersecting the nominal segment at level $2$ is considered (refer Fig. \ref{fig: 26}a). The exogenous UAS arrival rate is $\lambda_e = 0.2$. The expected number of exogenous UAS in the level $2$ zones is $S\lambda_e = 2.2$. As $M$ must be greater than $2.2$ for UAS to cross the exogenous traffic stream, $M$ is chosen to be $4$. The exogenous UAS transitions with a nominal velocity of $1$ cell per timeslot. The exogenous UAS do not follow the proposed mitigation strategy and are hence unaffected by congestion in the zone. However, the exogenous UAS contributes to congestion. When $M$ or more than $M$ look-ahead positions (exogenous UAS included) are present in a zone, then the UAS laterally transitions to avoid such zone (as shown in snapshot 1 of Fig. \ref{fig: 26}a). 

Unlike in the presence of a static obstacle, where it is certain that a UAS in $\beta\rightarrow \gamma$ path would always laterally transition every timeslot, here in the presence of an exogenous stream, the lateral transition of UAS is conditional (congestion-dependent). The $L$-slot look-ahead predicted positions of these UAS cannot be determined and are excluded in Definition $2$. This results in the UAS transitioning laterally and those transitioning \textit{upstream} experiencing conflicts (as shown in snapshot 1-2 of Fig. \ref{fig: 26}a). A conflict among UAS belonging to the same stream is referred to as internal conflict. A conflict between a UAS and an exogenous UAS is referred to as exogenous conflict. Though conflicts occur, the congestion mitigation strategy guarantees no more than $M$ conflicts in any zone, that is UAS intrinsic safety is improved. Further, the mitigation strategy helps reduce the total number of conflicts in the stream (as shown in Fig. \ref{fig: 26}b). These advantages come at the expense of additional time delay in the UAS path and energy cost proportional to the number of diagonal transitions (refer Fig. \ref{fig: 26}b). Here, the total number of \textit{upstream} and diagonal transitions executed to reach Level $5$ zone is $5\times S = 55$). Note that the additional cost is insignificant when computed over the entire flight path. Depending on the inter-separation distance distribution between the exogenous UAS, the traffic spread may increase or decrease as shown in snapshots 3-4 of Fig. \ref{fig: 26}a. However, the queueing models can be used to estimate the expected traffic spread as validated in Fig. \ref{fig: 26}c. 
\begin{figure}[h!]
    \centering
    \includegraphics[width = \linewidth]{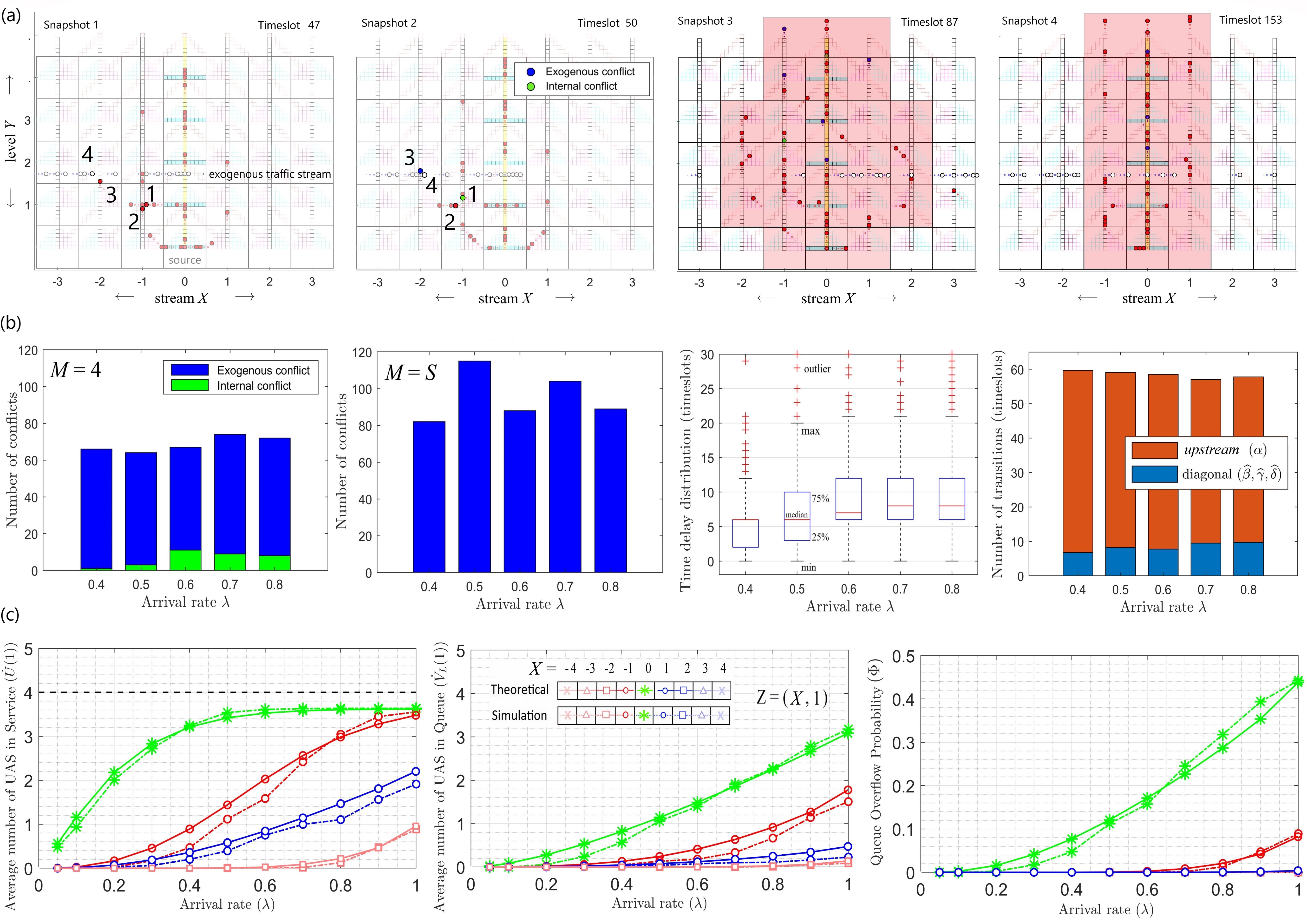}
    \caption{Congestion mitigation being employed by UAS in the presence of an exogenous traffic stream when $M = 4$ and $\eta = 0.3$. a) In snapshot 1, UAS 1 conflicts with UAS 2 belonging to the same stream at timeslot 47, and UAS 3 would conflict with an exogenous UAS 4 at timeslot 49. Rule \ref{rule: 6} is employed by UAS to resolve the conflict. Snapshot 2 shows the UAS positions after the conflict is resolved. b) Comparison of the number of conflicts with ($M = 4$) and without ($M = S$) the proposed congestion mitigation strategy and trade-off concerning time-delay and energy cost associated with the maneuver. c) Validation of expected traffic stream spread estimation in the presence of an exogenous stream ($\lambda_e = 0.2$).  }
    \label{fig: 26}
\end{figure}

\section{Conclusion} \label{sec: conclusion}

In this paper, we present an airspace design for the altitude-layered urban airspace and propose a novel preference-based distributed congestion mitigation strategy. The proposed strategy offers a unified framework for handling contingencies such as congestion in the urban airspace, conflicts among exogenous UAS traffic streams, and static obstacle avoidance. The advantages of the proposed approach are many-fold. The airspace design is the responsibility of UTM, which could be done offline, and the UTM parameters could be communicated to the UAS as needed. The strategy is decentralized, with individual UAS being responsible for mitigating congestion with minimal role played by UTM. Hence, the strategy is scalable to the rapidly growing number of UAS. The strategy improves UAS intrinsic safety by guaranteeing no more than $M$ UAS are present in any zone, also reducing the number of UAS conflicts. However, this comes at the expense of a minimal time delay and energy costs introduced in the UAS flight path associated with lateral and diagonal maneuvers. The emerging traffic behavior when every UAS employs the proposed strategy closely resembles a parallel air-corridor network, with the number of active air-corridors (referred to as traffic stream spread) adapting to the UAS demand. The UAS traffic is directional; thus, there is no compromise on the UTM throughput. The paper also presents queueing theory models for estimating the expected traffic spread for any given UAS demand. The estimates help the UTM reserve the airspace beforehand for any foreseen congestion. Extending the above strategy to the general scenario where more than two exogenous UAS traffic streams intersect (may not be perpendicular) and affect each other would be the future direction for this work. Further, proposing a probabilistic confidence definition for congestion instead of the present deterministic Definition \ref{def: 1} would help further reduce the internal conflicts. There may be other factors that affect the UAS preference to avoid or enter a congested region that can be incorporated into the above strategy.  

\section*{Declaration of Competing Interest}
The authors declare that they have no known competing financial interests or personal relationships that could have appeared to influence the work reported in this paper.

\section*{Acknowledgment}
The research reported here was funded by the Commonwealth Scholarship Commission and the Foreign, Commonwealth and Development Office in the UK. Sajid Ahamed is grateful for their support. All views expressed here are those of the author(s), not the funding body.

\section*{Appendix}
\begin{theorem}
    For Stream$\{0\}$ queuing system, the polynomial equation in $\Theta_0$ has at least one real root in $[0,1]$ \label{thm: Existence}
\end{theorem}
\begin{proof}
Consider a simplistic case where the queuing system has $L = 3$ timeslots queuing deadline. The PGF of the arrival process is given by $A(z) = (1-\lambda) + \lambda z$. Recursively, substituting $V_j(z)$ in $V_{j+1}(z)$, $j = 1,2,$ we get,
\begin{align}
    &V_1(z) = A(z)\Theta_0 + (1-\Theta_0) \ \ , \ \ V_2(z) = A(z) V_1(z)\Theta_0\ +  \left[ V_1(0) +\left(\cfrac{V_1(z) - V_1(0)}{z}\right) A(z) \right](1-\Theta_0) \label{eqn: V2 queue}
    \\
    &V_3(z) = A(z) V_2(z)\Theta_0\ +  \left[ V_2(0) +\left(\cfrac{V_2(z) - V_2(0)}{z}\right) A(z) \right](1-\Theta_0) \label{eqn: V3 queue} 
\end{align}
The computation of $V_3(z)$ requires $V_2(0)$, which on substituting $z = 0$ in Eqn. \ref{eqn: V2 queue} is found to be indeterminate ($0/0$) because of the $( V_1(z) - V_1(0) )/z$ term. However, when we substitute $A(z)$,  $V_1(z)$, $V_1(0)$ and expand the term,
\begin{align}
    V_1(z) - V_1(0) = \cancel{(1-\lambda)\Theta_0} + z\lambda\Theta_0 + \cancel{(1-\Theta_0)} - \cancel{(1-\lambda)\Theta_0} - \cancel{(1-\Theta_0)} = z\lambda\Theta_0
\end{align}
we observe that the numerator has a zero at $z = 0$. Similarly, in Eqn.  \ref{eqn: V3 queue}, the term 
$V_2(z) - V_2(0) = z(\ z\lambda^2\Theta_0^2 + 2\lambda(1-\lambda)\Theta_0^2 + 2\lambda\Theta_0(1-\Theta_0) \ )
$ has a zero at $z=0$. Irrespective of the number of times $V_j(z)$ is recursively substituted in $V_{j+1}(z)$, $j = 1,2,...,L-1$, there would be pole-zero cancellations in $( V_j(z) - V_j(0) )/z$. 

The pole-zero cancellation is a consequence of Rouché's theorem. Rouché's theorem states that if $f(z)$ and $g(z)$ are two analytic functions inside and on a simple closed curve $C$ such that $|f(z)| > |g(z)|$ for all $z \in C$, then both $f(z)$ and $f(z) + g(z)$ have the same number of zeros inside $C$. Consider the simple curve $|z| = 1$.
\begin{align}
    &V_j(z) \ =\  A(z)V_{j-1}(z)\Theta_0\ + \left[ V_{j-1}(0) + \left( \cfrac{V_{j-1}(z) - V_{j-1}(0)}{z}\right)A(z)\right](1-\Theta_0) \notag \\
    &=\ V_{j-1}(0)(1-\Theta_0)  \left[ 1 + \cfrac{ A(z) \mathlarger{\mathlarger{[}} (V_{j-1}(z) - V_{j-1}(0))(1-\Theta_0)  + z\Theta_0 V_{j-1}(z)\mathlarger{\mathlarger{]}} }{zV_{j-1}(0)(1-\Theta_0)}  \right] = \ V_{j-1}(0)(1-\Theta_0) \left[ 1 + \cfrac{g(z)}{f(z)}\right] 
\end{align}
\noindent When $|z| = 1$, we have $|f(z)| = V_{j-1}(0)(1-\Theta_0)$. If $z = 1$, then $\mathlarger{\vert}A(z)\mathlarger{\vert} = | 1-\lambda + \lambda z | = 1 ,\ \mathlarger{\vert}V_{j-1}(z)\mathlarger{\vert} = 1$.
\begin{align}
     \mathlarger{\vert}g(z)\mathlarger{\vert} = \mathlarger{\vert}A(z)\mathlarger{\vert} \mathlarger{\vert} (V_{j-1}(z) - V_{j-1}(0))(1-\Theta_0)  + z\Theta_0 V_{j-1}(z) \mathlarger{\vert} = 1 - V_{j-1}(0)(1-\Theta_0) < V_{j-1}(0)(1-\Theta_0)
\end{align}
If $|z| = 1$ but $z \neq 1$, then $|A(z)| < 1-\lambda + \lambda|z| < 1$, $|V_{j-1}(z)| < 1$.
\begin{align}
    \mathlarger{\vert}g(z)\mathlarger{\vert} &= \mathlarger{\vert}A(z)\mathlarger{\vert} \mathlarger{\vert} (V_{j-1}(z) - V_{j-1}(0))(1-\Theta_0)  + z\Theta_0 V_{j-1}(z) \mathlarger{\vert} \notag \\
    & < \mathlarger{\vert}A(z)\mathlarger{\vert} \left(\ \mathlarger{\vert} V_{j-1}(z) \mathlarger{\vert} (1-\Theta_0) + V_{j-1}(0) (1-\Theta_0) + \Theta_0 \mathlarger{\vert} V_{j-1}(z) \mathlarger{\vert}\ \right) \notag \\
    & < \mathlarger{\vert}A(z)\mathlarger{\vert}\mathlarger{\vert} V_{j-1}(z) \mathlarger{\vert} + \mathlarger{\vert}A(z)\mathlarger{\vert} V_{j-1}(0) (1-\Theta_0) \ <\ V_{j-1}(0) (1-\Theta_0)
\end{align}
Since $|f(z)|>|g(z)|$ for all $|z| = 1$, the numerator $f(z)+g(z)$ and denominator $f(z)$ have same number of zeros inside the region $|z| < 1$. Let $z_1, z_2$ be the zeros of $f(z)$ and $g(z)$, respectively. By the $j^{th}$ recursion, the term $(z_2/z_1)^{j}$ would be present in the PGF $V_{j+1}(z)$. For the PGF $V_{j+1}(z)$ to remain bounded, $|z_2|/|z_1|$ should be less than or equal one. Since $z_1 = 0$ and $|z_2| \leq |z_1|$,  the zero $z_2$ must be equal to $z_1$, leading to pole-zero cancellation. Thus, $V_{L}(0)$ can always be determined and is a polynomial in $\Theta_0$ of degree $L$.

By substituting $V_{L}(0)$ in Eqn. \ref{eqn: 0 feedback}, we get $W_1(0)$ as function of $\Theta_0$ of degree $(L+1)$. Further, substituting $W_1(0)$ in Eqn. \ref{eqn: forward} and assuming no exogenous UAS are present (that is, $E_0(0) = 1$), we get, 
\begin{align}
    \Theta_0 &= 1 - W_1(0)^{S-1} - \sum_{n = 1}^{M-1} \binom{S-1}{n} (1-W_1(0))^{n} W_1(0)^{S-n-1} = h(\Theta_0)
\end{align}
The solution $\Theta_0$ for the equation $h(\Theta_0) = \Theta_0$ is the fixed point of the polynomial function $h$. By Brouwer's fixed-point theorem, the continuous function $h: [0,1] \rightarrow [0,1]$ has at least one fixed-point $\Theta_0 \in [0,1]$.
\end{proof}

\printcredits

\bibliographystyle{cas-model2-names}
\bibliography{References.bib}

\end{document}